\documentclass[12pt]{article}
\usepackage[left=2.75cm,right=2.75cm,top=2.75cm,bottom=2.75cm,letterpaper]{geometry}
\usepackage[onehalfspacing]{setspace}
\usepackage{latexsym}
\usepackage{lscape}
\usepackage{alltt}
\usepackage{amsmath,amssymb}
\usepackage{amsthm}
\usepackage{dsfont}
\usepackage{bm}
\usepackage{multirow}
\usepackage{array}
\usepackage{color}
\usepackage{placeins}
\usepackage{longtable}
\usepackage{adjustbox}
\usepackage{pdflscape}
\usepackage{graphicx}
\usepackage{enumerate}
\usepackage{enumitem}
\usepackage[round]{natbib}
\usepackage{subcaption}
\usepackage{booktabs}
\usepackage{caption}
\usepackage{xcolor}
\usepackage{longtable}
\usepackage{threeparttable}
\usepackage{makecell}
\usepackage{array}
\usepackage{caption}
\usepackage{soul}
\usepackage{float}
\usepackage{hyperref}

\definecolor{orange}{rgb}{1.0, 0.55, 0.0}

\newtheorem{theorem}{Theorem}[section]
\newtheorem{assumption}{Assumption}
\newtheorem{lemma}{Lemma}[section]
\newtheorem{proposition}{Proposition}[section]

\newtheorem{remark}{Remark}[section]

\newtheorem{algorithm}{Algorithm}

\newcommand{\beqn}{\begin{equation}}
\newcommand{\eeqn}{\end{equation}}
\newcommand{\bneqna}{\begin{eqnarray}}
\newcommand{\eneqna}{\end{eqnarray}}
\newcommand{\beqna}{\begin{eqnarray*}}
\newcommand{\eeqna}{\end{eqnarray*}}
\newcommand{\bdmath}{\begin{displaymath}}
\newcommand{\edmath}{\end{displaymath}}

\begin{document}

\renewcommand{\thefootnote}{\arabic{footnote}} 

\title{\Large{Fixed-smoothing Uniform Inference for Quantile Regression\thanks{First arXiv version: September 5, 2026. We are grateful to  participants at 2026 NBER-NSF Time Series Conference for helpful comments and suggestions. All the remaining errors are ours.}}}

\author{Kaicheng Chen\thanks{School of Economics, Shanghai University of Finance and Economics (chenkaicheng@sufe.edu.cn).}
\and
Antonio F. Galvao\thanks{Department of Economics, Michigan State University (agalvao@msu.edu).}
\and 
Seunghwa Rho\thanks{College of Economics and Finance, Hanyang University, Korea (srhoecon@hanyang.ac.kr).}
\and
Timothy J. Vogelsang\thanks{Department of Economics, Michigan State University (tjv@msu.edu).}
\and
Jungmo Yoon\thanks{College of Economics and Finance, Hanyang University, Korea (jmyoon@hanyang.ac.kr).}
\\
}

\date{\vspace{-5ex}}

\maketitle

\markboth{\sc }{\sc }

\begin{abstract}
\begin{spacing}{1.15}
This paper develops fixed-smoothing inference methods for time-series quantile regression that are robust to heteroskedasticity and autocorrelation. Our approach is uniformly valid over quantile levels and accounts for dependence over time and across quantiles. It enables uniform confidence bands, Wald and Sup-t tests for joint hypotheses, and tests of shape restrictions, providing a unified framework for assessing heterogeneity in quantile effects. Under weak dependence, uniform inference for quantile regression processes is non-pivotal because the limiting distributions depend on the long-run covariance structure across quantiles. To address this issue, we develop two complementary approaches. The uniform-in-$\tau$ method estimates this structure and simulates the non-pivotal limiting distribution. For certain tests involving finitely many quantiles, the stack-Wald method delivers pivotal fixed-smoothing inference. We establish the asymptotic validity of both approaches. Simulations show improved size control over existing HAC-based procedures while maintaining good power. An application to stock-return prediction reveals substantial heterogeneity in predictive effects.
\end{spacing}

\vspace{0.25cm}


\vspace{0.25cm}

\end{abstract}

\thispagestyle{empty}\setcounter{page}{0}
\baselineskip=19.5pt

\onehalfspacing

\newpage

\section{Introduction}

Quantile regression (QR) provides an effective framework for studying heterogeneous effects in time-series. Recent applications in economics and finance include measures of systematic risk \citep{Adrian.Brunnermeier.2016, Adrian.Boyarchenko.Giannone.2019}, directional predictability \citep{Han.Linton.Oka.Whang.2016}, predictive regressions \citep{Cenesizoglu.Timmermann.2008, Lee16, MaynardShimotsuKuriyama24}, structural change \citep{Qu.2008, SuXiao08}, and quantile vector autoregressions  \citep{White.Kim.Manganelli.2015}. Much of this literature assumes that regression errors are either independent identically distributed (i.i.d.) or martingale difference sequences (MDS) which substantially simplifies inference. 

Inference becomes considerably more challenging in the presence of serial dependence. An emerging literature has developed robust inference procedures for time-series QR. Examples include the heteroskedasticity and autocorrelation consistent (HAC) covariance matrix estimator of \citet{galvao2024hac} and the block bootstrap of \citet{GregoryLahiriNordman18}. However, inference based on HAC standard errors or bootstrap methods often exhibits substantial size distortions in finite samples, particularly when the data are highly persistent. To improve finite-sample accuracy, recent studies have proposed fixed-smoothing inference \citep{hwang2025har}, self-normalized inference \citep{hoga2025selfnormalized}, and dependent wild bootstrap procedures \citep{CaiLong2026}. These methods, however, are designed for \textit{pointwise inference} at a single given quantile level, whereas many empirical questions require \textit{simultaneous inference} across multiple quantiles or \textit{uniform inference} over an entire range of quantiles.

This paper develops fixed-smoothing (fixed-b and fixed-K) inference methods for the QR process that are robust to heteroskedasticity and autocorrelation (HAR). Importantly, our approach is uniform in quantile levels and explicitly accounts for the dependence structure \textit{across} quantiles. For hypotheses involving a finite set of quantiles, we develop simultaneous inference procedures based on Wald tests. For hypotheses defined over a continuum of quantiles, we construct uniform confidence bands and Sup-t tests for the entire quantile process. We also develop tests for shape restrictions, including monotonicity and homogeneity of quantile effects. Together, these procedures provide a unified framework for conducting robust uniform inference in time-series QR and assessing heterogeneity in effects.

To illustrate the importance of uniform inference in QR models, consider Figure~\ref{fig-app1}, which presents QR estimates from predicting monthly S\&P 500 returns using stock variance, a measure of market-wide volatility. \citet{Welch.Goyal.2008} provide a comprehensive analysis of stock return predictability, and we use an updated version of their data set. QR allows researchers to examine how the predictive effects varies across the return distribution. The pattern is highly heterogeneous: the effect of stock variance is negative in the left tail, close to zero around the median, and positive in the right tail. In other words, periods of higher market volatility make next month’s returns more dispersed, producing larger losses in the worst outcomes and larger gains in the best outcomes.

The solid black line in the figure shows the estimated quantile effect $\beta(\tau)$, while the colored dotted lines provide the corresponding 95\% pointwise confidence intervals. The blue lines are based on the HAC method of \cite{galvao2024hac}, whereas the purple lines are constructed using the fixed-b HAR procedures developed in \citet{hwang2025har}. Both confidence intervals are pointwise in the sense that they are valid for each quantile level considered separately.

\begin{figure}[!htbp]
\centering
\caption{Quantile Effects of Stock Variance on S\&P 500 Returns}
\label{fig-app1}\vspace{-0.75cm}
\includegraphics[width=0.6\textwidth, height=0.38\textheight]{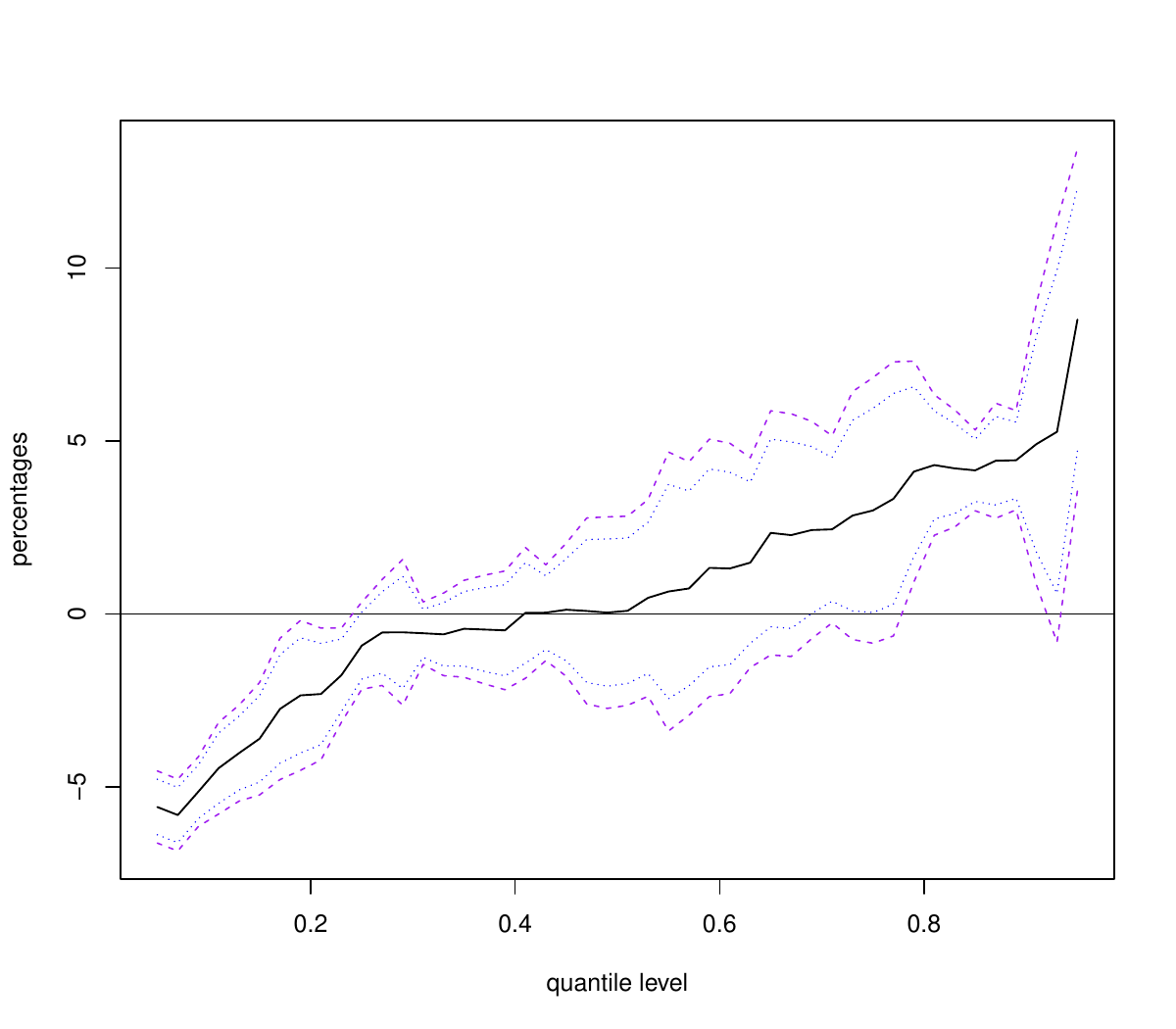}
\par
\begin{minipage}[l]{16.0cm}
{\footnotesize Estimated quantile effects of lagged stock variance on monthly S\&P 500 excess returns (solid black line). The dotted blue lines and dashed purple lines represent 95\% pointwise confidence bands based on the HAC method of \citet{galvao2024hac} and the fixed-b HAR method, respectively. The Bartlett kernel is used with bandwidth $M = 0.2 T$.}
\end{minipage}
\end{figure}

The limitations of pointwise inference are immediately apparent. At $\tau = 0.91$, the estimated effect is $4.91$ percent with a 95\% confidence interval of $(0.83, 8.98)$, suggesting a significant positive effect. Yet at the nearby quantile $\tau = 0.93$, the estimated effect is similar at $5.26$ percent, while the confidence interval becomes $(-0.81, 11.34)$, no longer excluding zero. Should we therefore conclude that stock variance has a significant predictive effect in the right tail, or that the evidence is inconclusive? Because empirical applications rarely provide clear reasons to focus on a particular quantile level, pointwise inference alone provides an incomplete picture.

Researchers are often interested in hypotheses concerning several quantiles simultaneously. For example, given two quantile levels $\tau_1$ and $\tau_2$, one may wish to determine whether both effects differ from zero, whether $\beta(\tau_1) > \beta(\tau_2)$, or whether the two effects are equal. Such questions require simultaneous inference across quantiles rather than separate inference at each quantile level. 

More broadly, many empirical questions concern an entire region of the distribution. Let $\mathcal{T}$ denote a subset of the unit interval. One may wish to test whether a predictor has zero effect over $\mathcal{T}$, that is, $H_0 : \beta(\tau) = 0$ for all $\tau \in \mathcal{T}$ against the alternative that the effect differs from zero for some $\tau \in \mathcal{T}$. Alternatively, one may wish to test whether the effect is uniformly nonnegative or whether it is increasing across quantiles. In these cases, the object of interest is the entire quantile process, $\beta(\tau)$, not its value at a single quantile level. Valid inference in such settings requires procedures that control error uniformly over $\tau$. Uniform inference is essential in a variety of applications, including the analysis of expected shortfall, whose estimand involves $\int_{\mathcal{T}} \beta(\tau) d\tau$ over tail regions $\mathcal{T}$ (\citealp{PATTON2019388}), and portfolio comparison problems, where expected utility of a portfolio depends on the entire return distribution rather than a few selected quantiles (\citealp{Harvey.Siddique.2000}; \citealp{Dittmar2002}).

The main challenge in extending fixed-smoothing inference to the entire QR process is that the resulting limiting distributions are generally non-pivotal across quantiles. This contrasts with conditional mean models where nuisance long-run variances often cancel from the fixed-smoothing limits. In QR under weak dependence, the key nuisance parameter is the long-run covariance kernel of the \textit{quantile score process}, which captures both serial dependence \textit{and} dependence across quantiles. For inference at a fixed quantile level, the cross-quantile dependence is irrelevant, and the fixed-smoothing statistics have pivotal limiting distributions \citep{hwang2025har,hoga2025selfnormalized}. For uniform inference over quantiles, however, the cross-quantile dependence enters the limiting distributions and renders them non-pivotal. To generate critical values for non-pivotal asymptotic limits, we provide a simulation algorithm that uses a consistent estimator of the long-run covariance kernel to simulate the limiting Gaussian process that appears in the limits of the test statistics. We refer to it as the uniform-in-$\tau$ method. 

For certain test statistics involving a \textit{finite} collection of quantiles, asymptotically pivotal limits are obtained using what we call the ``stack-Wald'' method which preserves the dependence structure across quantiles. These limits take ``standard'' fixed-smoothing forms and asymptotic critical values are available in the existing literature.

From the broader perspective this paper contributes to three strands of the literature: fixed-smoothing HAR inference, time-series QR, and simultaneous inference for quantile processes. In time-series settings \citet{KieferVogelsang02,KieferVogelsang05} introduced fixed-b inference for kernel HAR tests and \cite{mueller-lrv} introduced fixed-K inference for orthonormal-series (OS) HAR tests. Subsequent work established higher-order asymptotic properties and developed test-optimal bandwidth selection procedures 
\citep{J2004,SPJ2008,ZhangShao2013,sun2011robust,sun2013os,S2014b,S2014,Lazarus.Lewis.Stock.Watson.2018,Lazarus2021}. \citet{hwang2025har} obtained fixed-smoothing results for quantile regressions that hold for a given quantile. Our paper contributes to this literature by obtaining fixed-smoothing results for tests that involve multiple or a continuum of quantiles - a practically important, and theoretically nontrivial, extension of the given-quantile case.

The second strand is the literature on time-series QR. Existing work on quantile processes typically assumes serially uncorrelated score processes, in which case the limiting process reduces to a Kiefer process, making inference substantially more tractable
\citep{Qu.2008,SuXiao08,GalvaoKatoMontesRojasOlmo14}. By contrast, we establish weak convergence of the partial-sum empirical process under general weak dependence, providing a foundation for inference on QR processes with serially dependent errors. Recent advances in robust inference for time-series QR include HAC estimation \citep{galvao2024hac}, bootstrap procedures \citep{GregoryLahiriNordman18, CaiLong2026}, fixed-smoothing inference methods \citep{hwang2025har}, and self-normalized inference \citep{hoga2025selfnormalized}. These approaches, however, are primarily designed for pointwise inference at a fixed quantile level. Our paper complements this literature by explicitly accounting for the intersection of dependence across time and dependence across quantiles. This allows joint hypothesis tests and uniform confidence bands that involve the entire QR process.

The third strand is the literature on simultaneous inference. When the object of interest is a collection of parameters or a shape restriction on a function, pointwise confidence intervals can be misleading because they do not account for multiple comparisons. Simultaneous confidence bands have a long history in statistics, and the Sup-t approach has recently received renewed attention in economics \citep{montiel2019simultaneous}. For QR, uniform inference procedures have been developed by, e.g., \citet{2002.Koenker.Xiao}, \citet{ChernozhukovFernandezVal05}, \citet{2015.Qu.Yoon}, and \citet{Belloni.Chernozhukov.2019}. These methods provide inference that is uniform over quantile levels, but are largely developed under independent sampling or dependence structures that do not require HAR inference. Our paper bridges this gap by developing uniform inference procedures for quantile processes that remain valid in the presence of heteroskedasticity and autocorrelation.

We conduct a simulation study to evaluate the finite-sample performance of the proposed methods. The results show that the fixed-b and fixed-K approaches substantially reduce the size distortions associated with existing small-b and large-K inference procedures. Both stack-Wald and Sup-t tests exhibit good finite-sample performance in terms of size and power, and similar results are obtained for the tests of shape restrictions on quantile effects. We implement data driven bandwidth selection procedures proposed in the literature, and find that these bandwidth rules work relatively well in practice.

Our empirical application revisits the predictive QR framework of \citet{Cenesizoglu.Timmermann.2008} for stock returns. We examine predictive effects across quantiles and forecast horizons while allowing for serial dependence in the regression errors. The results show that the stock variance has the strongest predictive effect at short horizons, particularly in the tails of the return distribution, and that the effect weakens as the forecast horizon increases. The application illustrates how uniform inference can reliably uncover heterogeneity in predictive effects across both quantiles and horizons.

The remainder of the paper is organized as follows. Section \ref{sec:model} introduces the econometric framework and the long-run variance estimators. Section \ref{sec:fixedb_asym} develops the fixed-smoothing asymptotic theory for the QR process. Section \ref{sec:inference} develops the inference procedures, including uniform-in-$\tau$ tests, the stack-Wald tests, uniform and simultaneous confidence bands, tests of shape restrictions, and data-dependent bandwidth selection. Section \ref{sec:sim} reports simulation results, and Section \ref{sec:app} presents an empirical application to predictive QR. Section \ref{sec:conclusion} concludes. The Appendix provides the proofs, while additional theoretical and simulation results are presented in the Supplemental Appendix. The \texttt{R} package \texttt{QRfs} implements the proposed inference methods.\footnote{The package is available at \url{https://github.com/seunghwaRho/QRfs}.}

\section{Econometric Model}
\label{sec:model}
When modeling economic behavior, empiricists commonly estimate quantile regression models that allow for heterogeneity across conditional quantile functions. Let $y_t$ be a real-valued dependent variable, and $x_t$ be a $p \times 1$ vector of covariates. Let $\mathcal{T}$ be a closed subset of $(0,1)$, for example, $\mathcal{T} = [\epsilon,1-\epsilon]$ for a small $\epsilon > 0$. For $\tau \in \mathcal{T}$, the $\tau$-th conditional quantile function of $y_t$ is given by $Q_{y_{t}}(\tau|x_t) = x_{t}^{\prime}\beta_{0}(\tau)$, $t = 1,\ldots,T$. The quantile error $e_{t}(\tau)$ is defined as: $e_{t}(\tau) = y_t - x_{t}^{\prime}\beta_{0}(\tau)$. This innovation term is allowed to be serially correlated. The $\tau$-th conditional quantile of the error is assumed to be zero, $Q_{e}(\tau|x_t) = 0$, a normalization condition. The parameter of interest is $\beta_{0}(\tau)$, a $p \times 1$ vector, which is allowed to depend on the quantile level. The coefficient vector $\beta_{0}(\tau)$ is estimated by \begin{equation}
\widehat{\beta}(\tau) = \underset{\beta(\tau)}{\textnormal{argmin}}  \sum_{t=1}^{T} \rho_{\tau}\left(y_{t} - x_{t}^{\prime}\beta(\tau)  \right),
\label{eq-beta.hat}
\end{equation}
where $\rho_{\tau}(u) = u \cdot (\tau - I(u \leq 0))$. For a given $\tau$, it is well known that, under regularity conditions, as $T \to \infty$, $\sqrt{T}(\widehat{\beta}(\tau) - \beta_{0}(\tau)) \stackrel{d}{\rightarrow} N(0, \Sigma(\tau) )$, 
where the asymptotic covariance matrix is given by $\Sigma(\tau) = D(\tau)^{-1} \Lambda(\tau) D(\tau)^{-1}$. The outer term in the asymptotic covariance matrix, $D(\tau) = \textnormal{E}\left[f(0|x_t)x_t x_{t}^{\prime}\right]$, depends on the conditional density function evaluated at the $\tau$-quantile.\footnote{More generally, $D(\tau) = \textnormal{E}\left[f(Q_{e}(\tau|x_t)|x_t) x_{t}x_{t}^{\prime} \right]$. Because $Q_{e}(\tau|x_t) = 0$, the two expression coincide.} It accounts for potential conditional heteroskedasticity. Estimating $D(\tau)$ is relatively straightforward following the literature even in time-series settings.
The middle term, $\Lambda(\tau)$, is the heteroskedasticity and autocorrelation (HAR) component of $\Sigma(\tau)$, capturing the long-run variance of the QR score. Define the $p \times 1$ QR score vector as $$z_{t}(\tau) =  \left(\tau - I(e_{t}(\tau) \leq 0)\right) \cdot x_t.$$ The long-run variance of $z_{t}(\tau)$ is then defined as  
$$\Lambda(\tau) = \sum_{j=-\infty}^{\infty} \textnormal{Cov}(z_{t}(\tau) , z_{t-j}(\tau)).$$%

The main objective is to conduct inference for the $\beta(\tau)$ parameter. Consider testing the null hypothesis, $H_{0}:r\left(\beta(\tau)\right)
=0$, against the alternative, $H_{A}:r\left( \beta(\tau)\right) \neq0$, where $r\left( \beta(\tau)\right) $ is a $q\times1$ vector $\left( q\leq p\right) $ of
continuously differentiable functions, with a first derivative matrix $%
R\left( \beta\right) =\partial r\left( \beta\right) /\partial \beta^{\prime}$
of full rank $q$. The Wald statistic for inference on the QR coefficients, using $\widehat{\Lambda}(\tau)$, a HAR long-run variance estimator of $\Lambda(\tau)$, is given by 
\begin{align}
\mathcal{W}(\tau) = T \ r(\widehat\beta(\tau))^{\prime} \left[R(\widehat\beta(\tau)) \, \widehat{\Sigma}(\tau) \, R(\widehat\beta(\tau))^{\prime}\right]^{-1} r(\widehat\beta(\tau)), \label{eq_wald}
\end{align}
where $\widehat{\Sigma}(\tau) = \widehat{D}(\tau)^{-1} \widehat{\Lambda}(\tau) \widehat{D}(\tau)^{-1}$ is an estimator of $\Sigma(\tau)$. In the case of a single restriction, $q=1$, the corresponding $t$-statistic is 
\begin{align}
	t(\tau) = \frac{\sqrt{T} \ r(\widehat\beta(\tau))}{\sqrt{R(\widehat\beta(\tau)) \, \widehat{\Sigma}(\tau) \, R(\widehat\beta(\tau))^{\prime}}}.  
	\label{eq_t-stat}
\end{align}
It is typical to use a kernel-based estimator for $D(\tau)$ of the form:
\begin{align}
	\widehat{D}(\tau)=\frac{1}{hT}\sum_{t=1}^{T} k\left(\frac{\widehat{e}_{t}(\tau)}{h}\right)x_{t}x_{t}^{\prime}, 
\label{eq-D.hat}
\end{align}
where $k(\cdot)$ is a density kernel function and $h$ is a bandwidth. 

For the long run variance matrix, $\Lambda(\tau)$, we focus on kernel-based HAR estimators and orthonormal series (OS) HAR estimators. The kernel-based estimator is given by
\begin{align}
\widehat{\Lambda}^{\rm KE}(\tau) & = \frac{1}{T}\sum_{t=1}^{T}\sum_{s=1}^{T}\mathcal{K}\left(\frac{t-s}{M}\right)\widehat{z}_{t}(\tau)\widehat{z}_{s}(\tau)^{\prime},
\label{eq-Lambda.hat}
\end{align}%
where $\mathcal{K}$ is a symmetric kernel function in $L^2(\mathbb{R})$ satisfying $\mathcal{K}(0)=1$, $|\mathcal{K}(\cdot)|\le 1$, and that $\mathcal{K}(\cdot)$ is continuous at $0$. The kernel heteroskedasticity and autocorrelation consistent (HAC) estimator in (\ref{eq-Lambda.hat}) was studied by \citet{galvao2024hac} under traditional asymptotics for the bandwidth. They established conditions for the consistency of $\widehat{\Lambda}(\tau)$. However, the finite-sample accuracy of the consistency approximation is often unsatisfactory. The resulting t-tests tend to exhibit rejection frequencies above nominal levels in the presence of positive autocorrelation, sometimes considerably so.

The OS HAR estimator is given by
\begin{align}
    \widetilde{\Lambda}^{\rm OS}(\tau) = \frac{1}{K} \sum_{k=1}^K   \widetilde{\Lambda}_k(\tau) = \frac{1}{K} \sum_{k=1}^K  \Phi_T^k(\tau)  \Phi_T^k(\tau)^{\prime},
\label{eq-Lambda.hatOS}
\end{align}
where $K<\infty$\footnote{For hypothesis tests that require inversion of the variance estimator, $K$ must be at least as large as the number of restrictions to ensure that the variance estimator is full rank.} and 
$\Phi_T^k(\tau) =\frac{1}{\sqrt{T}}\sum_{t=1}^{T} \phi_k\left(\frac{t}{T}\right)\widehat{z}_{t}(\tau)$, $ k=1,2,\ldots$ The functions $\{\phi_k(r):r\in [0,1], k=1,2,...\}$ are a sequence of orthonormal basis functions in $L^2[0,1]$ satisfying $\int_0^1 \phi_k(r) dr =0$ for each $k$. 

To achieve more accurate inference and provide better finite-sample approximations than traditional methods, we employ fixed-b and fixed-K asymptotic theory. Under fixed-b asymptotics, results are obtained for the sequence of bandwidths $M = [bT]$ where $b$ is held fixed. Under fixed-K theory, $K$ is held fixed as $T$ increases. These fixed-smoothing approaches generate limiting distributions for test statistics that depend on the tuning parameters (kernel, bandwidth, $K$) that can be used as reference distributions to generate tuning-parameter-dependent critical values. In contrast, under traditional asymptotics, $b\rightarrow 0$ and $K\rightarrow \infty$ as $T\rightarrow \infty$ and have been labeled ``small-b'' and ``large-K'' asymptotics. Because small-b and large-K limits rely on consistency of the variance estimator, they necessarily do not depend on the kernel, bandwidth, or $K$.\footnote{For small values of $b$ or large values of $K$, the fixed-smoothing and traditional limits are often similar.}

It is important to note that the formulas we use for  $\widehat{\Lambda}^{\rm KE}(\tau)$ and $\widetilde{\Lambda}^{\rm OS}(\tau)$ are defined using the empirical scores, $\widehat{z}_{t}(\tau)$, that have \textit{not} been recentered using the sample average of $\widehat{z}_{t}(\tau)$. This contrasts with \cite{hwang2025har} where recentered empirical scores were used. Because our uniform-in-$\tau$ fixed-smoothing limits, for a given quantile, simplify to the same fixed-smoothing limits obtained by \cite{hwang2025har}, our results show that the fixed-smoothing limits obtained by \cite{hwang2025har} hold with and without recentering of the empirical scores.\footnote{Intuitively, this is not surprising given that the sample average of $\widehat{z}_{t}(\tau)$ is $o_p(1)$.}

\section{Uniform Fixed-Smoothing Asymptotic Theory}
\label{sec:fixedb_asym}

This section presents some useful results that hold \textit{uniformly} across quantiles. These results are used to obtain fixed-b and fixed-K asymptotic limits in the inference section that follows. Consider the partial sum empirical process for the score $z_{t}(\tau)$. For $r \in [0,1]$ and $\tau \in \mathcal{T}$, consider the following partial summation
\begin{equation*}
S([r T],\tau) = T^{-1/2} \sum_{t=1}^{[r T]} z_{t}(\tau) = \frac{1}{\sqrt{T}} \sum_{t=1}^{[rT]} (\tau - I(y_{t} \leq x_{t}^{\prime}\beta_{0}(\tau))) x_{t}.
\end{equation*}
Let $G(r,\tau)$, $(r, \tau) \in [0,1] \times \mathcal{T}$, be a $p$-vector of mean zero Gaussian process with covariance function: 
\begin{align}
    \textnormal{Cov}\left(G(r_1,\tau_1),G(r_2,\tau_2)\right) = (r_1 \wedge r_2) \, \Lambda(\tau_1 , \tau_2), 
    \label{cov}
\end{align}
where $\Lambda(\tau_1 , \tau_2)$ is the long-run covariance kernel defined by 
\begin{equation}
\Lambda(\tau_{1} , \tau_{2}) = \sum_{j=-\infty}^{\infty} \operatorname{Cov}(z_{0}(\tau_{1}), z_{j} (\tau_{2})). \label{eq:cov_kernel}
\end{equation}
Let $\Rightarrow$ denote weak convergence, and let
$ \ell^{\infty}([0,1]\times\mathcal{T})$ denote the space of bounded
vector-valued functions on $[0,1]\times\mathcal{T}$, equipped with the
uniform norm $\|\cdot\|_{\infty}$. Consider the following assumption.

\begin{assumption}
    \label{assum_gaussian_limit}
    As $T \rightarrow \infty$, $S([r T],\tau) \Rightarrow G(r,\tau)$ \hskip 0.10in in \ $\ell^{\infty}([0,1] \times \mathcal{T})$.
\end{assumption}

For weakly dependent processes, \cite{Philipp.Pinzur.2019}, \cite{Dehling.Taqqu.1989}, and \cite{Shao.Yu.1996} provide sufficient conditions on mixing coefficients and moment restrictions under which such weak convergence results hold. We treat this as a high-level assumption. Its validity can be established under Assumption~3 below, as shown in Proposition~\ref{prop:donsker} in the Supplemental Appendix.

Related work by \citet{Qu.2008} and \citet{SuXiao08} consider the QR process under a martingale difference sequence (MDS) assumption on the errors. This assumption leads to a substantially simpler limiting process and, consequently, permits nuisance-parameter-free inference. Our setting allows for weak temporal dependence, under which these simplifications are generally unavailable.

\begin{assumption}
\label{assum_asymp_linear}
    As $T \rightarrow \infty$, the following linear representation holds uniformly in $\tau \in \mathcal{T}$:
\begin{equation*}
\sqrt{T}\left(\widehat{\beta}(\tau) - \beta(\tau)\right) = D(\tau)^{-1} \frac{1}{\sqrt{T}} \sum_{t=1}^{T}\left(\tau - I(e_{t}(\tau) \leq 0)\right) x_{t} + r_{T},
\end{equation*}
where the remainder term $r_{T} = o_{p}(1)$ uniformly in $\tau$. 
\end{assumption}

This is a high-level assumption. Under standard conditions for QR, the above linear representation can be rigorously established. We refer the reader to \citet{GalvaoKatoMontesRojasOlmo14} for primitive sufficient conditions in a $\beta$-mixing setting. Under Assumption \ref{assum_asymp_linear}, we obtain $\sqrt{T}\left(\widehat{\beta}(\tau)-\beta_0(\tau)\right)\Rightarrow D(\tau)^{-1} G(1,\tau)  \; {\rm in} \ \   \ell^{\infty} (\mathcal{T})$. 

Let $x_{tj}$ denote the $j$-th component of $x_{t} = \left(1,x_{t2},\ldots,x_{tp}\right)^{\prime}$ and consider the following assumption:

\begin{assumption}
\label{assum_reg}
\begin{enumerate}
\item[i.] $\textnormal{E}\left[\| x_{t} \|^{4}\right] < \infty$. For any $t \neq s$, $\textnormal{E}\left| x_{t i_{1}} x_{t i_{2}} x_{t i_{3}} x_{s j_{1}} x_{s j_{2}} x_{s j_{3}} \right| < \infty$ for $1 \leq i_1, i_2 , i_3 \leq p$ and $1 \leq j_1, j_2 , j_3 \leq p$.

\vskip-0.1in

\item[ii.] The process $(y_{t},x_{t}^{\prime})^{\prime}$ is strictly stationary and $\beta$-mixing with coefficient $\beta(l)$ such that there exist constants $a \in (0,1)$ and $B \geq 0$ such that $\beta(l) \leq B a^{l}$. 

\vskip-0.1in

\item[iii.] The conditional density $f(y|x)$ is continuously differentiable w.r.t $y$ for all values of $x \in \mathcal{X}$. There exists a constant $\bar{f}>0$ such that $f(y|x) \leq \bar{f}$ for all $(y,x) \in R \times \mathcal{X}$. There exists $c>0$ such that $f(y|x) > c$ for all $y \in [Q(\epsilon|x),Q(1-\epsilon|x)]$ and $x \in \mathcal{X}$.

\vskip-0.1in

\item[iv.] There exists a constant $\check{f}>0$ such that $f(y_{t},y_{s}|x_{t},x_{s}) \leq \check{f}$ for all $y_{t}$, $y_{s}$, $x_{t}$, $x_{s}$. 

\vskip-0.1in

\item[v.] There exists $c>0$ such that $|f^{\prime}(y|x_{t})| < c$ for all $y \in [Q(\epsilon|x),Q(1-\epsilon|x)]$ and $x \in \mathcal{X}$.

\vskip-0.1in

\item[vi.] $\textnormal{plim}_{T \to \infty} T^{-1} \sum_{t=1}^{[rT]} f(0|x_{t}) x_{t} x_{t}^{\prime} = r D(\tau)$ uniformly in $r$ with $D(\tau)$ positive definite.

\vskip-0.1in

\item[vii.] $\textnormal{plim}_{T \to \infty} T^{-1}\sum_{t=1}^{[rT]} x_{tl} x_{ti} x_{tj} = r \sigma_{lij}$ uniformly in $r$ for any $l,i,j = 1,\ldots, p$ where $\sigma_{lij} = \textnormal{E}[x_{tl} x_{ti} x_{tj}] < \infty$.
\end{enumerate}
\end{assumption}

Assumption \ref{assum_reg} is standard in the QR literature. We impose a $\beta$-mixing condition on the underlying process because our analysis relies on the uniform central limit theorem for $\beta$-mixing processes developed in \citet{arcones1994central}. 
Assumption \ref{assum_reg} also includes regularity conditions on the regressors and the underlying density function, ensuring that the limiting variance-covariance matrix is well defined. Assumption 3(vi) is standard in the fixed-smoothing literature. Because $x_{t}$ includes a constant regressor, Assumption 3(vii) includes $\textnormal{plim}_{T \to \infty} T^{-1}\sum_{t=1}^{[rT]} x_{t} x_{t}^{\prime} = r \textnormal{E}\left[x_{t} x_{t}^{\prime} \right]$ as a special case. Assumptions 3(v) and 3(vii) are used to control the effect of the estimation error in the QR coefficients.

We now present the main theoretical results. The proofs of theorems and propositions are provided in the Appendix. Define the tied-down version of the Gaussian process by $\widetilde{G}(r, \tau) = G(r, \tau) - r G(1, \tau)$. The partial sum empirical process, $\widehat{S}(t,\tau)= T^{-1/2} \sum_{t=1}^{[r T]} \widehat{z}_{t}(\tau)$, satisfies a functional central limit theorem.

\begin{theorem}
\label{thm_1}
Suppose Assumptions \ref{assum_gaussian_limit}, \ref{assum_asymp_linear}, and \ref{assum_reg} hold. As $T \to \infty$, 
\begin{enumerate}
\item $ \widehat{S}([r T],\tau) = S([r T],\tau) -  r S(T,\tau) + u_{T} \, $ where $\, u_{T}  = o_{p}(1)$ uniformly in $(r, \tau) \in [0,1] \times \mathcal{T}$. 
\item If the remainder term in Assumption \ref{assum_asymp_linear} satisfies $r_{T} = o_{p}(T^{-1/4} \log(T))$, then $\, u_{T} = O_{p}(T^{-1/4} \log(T))$ uniformly in $(r, \tau) \in [0,1] \times \mathcal{T}$. 
\item  $\widehat{S}([r T],\tau) \Rightarrow G(r, \tau) - r G(1, \tau) = \widetilde{G}(r, \tau)$.
\end{enumerate}
\end{theorem}

\begin{remark}
\citet{bai1996} and \citet{Qu.2008} derive analogous results for the partial sum empirical process constructed from estimated errors under an MDS assumption. Theorem~\ref{thm_1} generalizes their findings to allow for weakly dependent errors. 
\end{remark}

\begin{remark}
  Theorem~\ref{thm_1} holds regardless of the smoothing asymptotics adopted. Consequently, test statistics constructed as functionals of the partial-sum empirical process generally have limiting distributions that depend on nuisance parameters under both small-b and fixed-b approximations. 
\end{remark}

In the next two theorems, we establish fixed-b and fixed-K asymptotic limits for the kernel and OS HAR variance estimators, respectively, uniform over quantile levels. Both results are new to the literature.

\begin{theorem}
    \label{thm_2}
    Suppose Assumptions \ref{assum_gaussian_limit}, \ref{assum_asymp_linear}, and \ref{assum_reg} hold.
  Let $M = bT$ for fixed values $b\in (0,1]$. As $T \to \infty$, $ \widehat{\Lambda}^{\rm KE}({\tau}) \Rightarrow  P\left(\widetilde{G}(r,\tau),b\right)$ in $\ell^{\infty}(\mathcal{T}) $, where
\begin{enumerate}
    \item If $\mathcal{K}(\cdot)$ is twice continuously differentiable everywhere and $\mathcal{K}''(\cdot)$ is uniformly bounded,
    \begin{align*}
        P(\widetilde{G}(r,\tau),b)= -&\frac{1}{b^2}\int_0^1\int_0^1 \mathcal{K}''\left(\frac{r-s}{b} \right) \widetilde{G}(r,\tau) \widetilde{G}(s,\tau)^{\prime} dr ds .
    \end{align*}
    \item If $\mathcal{K}(x)$ is continuous, $\mathcal{K}(x)=0$ for $|x|\geq 1$, $\mathcal{K}(.)$ is twice continuously differentiable everywhere except for $|x|=1$, $\mathcal{K}'_{-}(x)$, the derivative of $\mathcal{K}(x)$ from the left, and both $\mathcal{K}''(\cdot)$ and $\mathcal{K}'_{-}(.)$ are uniformly bounded. 
    \begin{align*}
        P(\widetilde{G}(r,\tau),b)=- &\frac{1}{b^2}\int\int_{|r-s|<b} \mathcal{K}''\left(\frac{r-s}{b} \right) \widetilde{G}(r,\tau) \widetilde{G}(s,\tau)^{\prime} dr ds  \\
        &+ \frac{\mathcal{K}'_{-}(1)}{b} \int_0^{1-b} \left( \widetilde{G}(r+b,\tau) \widetilde{G}(r,\tau)^{\prime}+\widetilde{G}(r,\tau) \widetilde{G}(r+b,\tau)^{\prime} \right) dr .
    \end{align*}

    \item If $\mathcal{K}(\cdot)$ is the Bartlett kernel,
\begin{align*}
    P(\widetilde{G}(r,\tau),b)= &\frac{2}{b}\int_0^1 \widetilde{G}(r,\tau) \widetilde{G}(r,\tau)^{\prime} dr \\
    &-\frac{1}{b}\int_0^{1-b} \widetilde{G}(r,\tau)\widetilde{G}(r+b,\tau)' dr  - \frac{1}{b}\int_0^{1-b} \widetilde{G}(r+b,\tau) \widetilde{G}(r,\tau)^{\prime}dr .
\end{align*}
\end{enumerate}
\end{theorem}

\begin{theorem}
\label{thm:fixedK}
Under the same setting as Theorem~\ref{thm_1}, suppose that, for each
$k=1,\ldots,K$, $\phi_k$ is continuous on $[0,1]$ and $\sup_{T\geq 1,\; r\in[0,1-1/T]}\left|\frac{\phi_k(r+1/T)-\phi_k(r)}{1/T}\right|<\infty.$
Then, as $T\to\infty$ with fixed $K<\infty$,
\[
\widehat{\Lambda}^{\rm OS}(\tau)
\Rightarrow
\frac{1}{K}\sum_{k=1}^K
G_k(1,\tau)G_k(1,\tau)^{\prime}
\quad \text{in } \ell^{\infty}(\mathcal{T}),
\]
where $\{G_k(1,\tau)\}_{k=1}^{K}$ are independent copies of
$G(1,\tau)$.
\end{theorem}
\noindent Under fixed-smoothing asymptotics, the long-run variance estimators no longer converge to the population covariance kernel; instead, they converge to random functionals of the limiting Gaussian process. 

The next result derives the asymptotic distributions of the Wald and t-test statistics. Let $R_0 = R(\beta_0(\tau))$, and define the following objects for the kernel and OS HAR estimators:
\begin{align*}
    H^{\rm KE}(\tau) &= R_0 D(\tau)^{-1} P\left(\widetilde{G}(r,\tau),b\right)D(\tau)^{-1} R_0^{\prime},\\
     H^{\rm OS}(\tau) &= R_0 D(\tau)^{-1} \frac{1}{K}\sum_{k=1}^KG_k(1,\tau)G_k(1,\tau)^{\prime} D(\tau)^{-1} R_0^{\prime}.
\end{align*}

\begin{proposition}
\label{prop:fixedb}
Suppose Assumptions~\ref{assum_gaussian_limit}, \ref{assum_asymp_linear}, and \ref{assum_reg} hold, together with
the conditions of Theorem~\ref{thm_2} for $\iota={\rm KE}$ and Theorem~\ref{thm:fixedK} for $\iota={\rm OS}$.
Suppose further that: (1) $r(\cdot):\mathbb R^p\to\mathbb R^q$, $q<p$, is
    continuously differentiable with full-rank derivative
    $R(\cdot)$, and $\{\beta_0(\tau):\tau\in\mathcal T\}$
    is contained in a compact subset of $\mathbb R^p$; (2) $\sup_{\tau\in\mathcal T}
    \|\widehat D(\tau)-D(\tau)\|=o_p(1)$ and
    $\inf_{\tau\in\mathcal T}\lambda_{\min}(D(\tau))>0$; (3) $\inf_{\tau\in\mathcal T}
    \lambda_{\min}(H^\iota(\tau))>0$ almost surely. For $\iota={\rm OS}$, suppose also that $K\geq q$.
Then, under the null hypothesis $r(\beta_0(\tau))=0$,
as $T\to\infty$ with $M=bT$ for fixed $b\in(0,1]$ when $\iota={\rm KE}$, and fixed $K$ when
$\iota={\rm OS}$, in $\ell^\infty(\mathcal T)$,
\begin{align}
\mathcal W^\iota(\tau) &\Rightarrow \left(R_0D(\tau)^{-1}G(1,\tau)\right)'
\left[H^\iota(\tau)\right]^{-1} \left(R_0D(\tau)^{-1}G(1,\tau)\right),
\label{wald_limit}\\
t^\iota(\tau) &\Rightarrow \frac{R_0D(\tau)^{-1}G(1,\tau)}
{\sqrt{H^\iota(\tau)}}, \quad\text{if }q=1. \label{t_limit}
\end{align}
\end{proposition}

As displays \eqref{wald_limit} and \eqref{t_limit} show, the limiting distributions of the test statistics are stochastic processes indexed by $\tau$ that depend on the long-run covariance kernel $\Lambda(\tau_{1} , \tau_{2})$ of the process $G(r,\tau)$, through $G(1,\tau)$ and $H(\tau,b)$. Because of this dependence, these limiting processes are not asymptotically pivotal. The remainder of this paper is devoted to developing inference procedures that address this issue.

Readers familiar with fixed-smoothing literature would note that the limits given by \eqref{wald_limit} and \eqref{t_limit} contrast sharply with the conventional fixed-smoothing limits obtained in \citet{KieferVogelsang05,mueller-lrv,sun2013os}. In conditional mean models, nuisance parameters cancel out of the limiting distributions, yielding asymptotically pivotal test statistics. This cancellation is possible because the partial sum process for the conditional mean can be written as the product of a long-run variance and a standard Gaussian process, such as Brownian motion. Such a multiplicative decomposition does not extend directly, at least not in a straightforward way, to QR models across quantiles.

To illustrate this point, we note two special cases in which the limiting process admits substantial simplifications. First, under the MDS assumption considered by \citet{Qu.2008} and \citet{SuXiao08}, the long-run covariance kernel reduces to the short-run covariance kernel and the partial sum empirical process simplifies to a $p$-vector of independent Kiefer processes. Consequently, inference for the quantile regression process becomes nuisance-parameter-free. This simplification, however, relies critically on the absence of serial dependence in the errors and generally does not extend to weakly dependent data.

Second, under weak dependence, when attention is restricted to a single fixed quantile level $\tau$, the marginal limit process admits the representation $G(r,\tau)=\Gamma(\tau)W_p(r)$, where $\Gamma(\tau)\Gamma(\tau)^{\prime} = \Lambda(\tau,\tau)$ and $W_p(r)$ is a $p$-vector of independent standard Brownian motions. This representation underlies existing fixed-smoothing procedures for pointwise inference in QR.\footnote{See equation (16) in \citet{LeeLiaoSeoShin25} and Lemma 4 in Hoga and Schulz (2025). A closely related assumption is implicit in \cite{hwang2025har}, who assume a Gaussian limit for the series-weighted partial sum with a separable long-run variance matrix. Such an assumption is equivalent to $S([r T],\tau) \Rightarrow \Gamma(\tau) W_{p}(r)$ because this convergence is sufficient for a Gaussian limit with separable long-run variance. See \citet{sun2011robust,sun2013os}.} However, it is only a characterization of the marginal distribution at a given quantile level and does not extend to the dependence structure across quantiles because the covariance kernel generally does \textit{not} admit the factorization $\Lambda(\tau_1,\tau_2) = \Gamma(\tau_1)\Gamma(\tau_2)^{\prime}$.

The present paper considers process-level inference under weak dependence, where neither of the above simplifications is available. In general, the limiting Gaussian process cannot be represented as a $p$-vector of Kiefer processes $K(r,\tau)$, nor does it admit a separable representation of the form $G(r,\tau)=\Gamma(\tau)K(r,\tau)$. The reason is that the long-run covariance kernel $\Lambda(\tau_1,\tau_2)$ reflects the joint serial dependence of the quantile score process \textit{across} different quantile levels and is generally non-separable. Consequently, the nuisance parameters contained in $\Lambda(\tau_1,\tau_2)$ do not cancel from the fixed-smoothing limiting distributions. This is intuitively similar to the non-pivotal fixed-smoothing limits obtained by \cite{chen2024fixed} in panel regressions with two-way clustered components although the details are different.

\section{HAR Inference in QR}
\label{sec:inference}

This section develops inference procedures for two classes of hypotheses, a continuum of quantile levels and a finite number of quantile levels. To facilitate visualization, we also construct a simultaneous confidence band for slope coefficients across multiple quantile levels. In addition, we propose a test-based, data-driven bandwidth selection rule for all tests considered in this section using the kernel‑HAR and OS‑HAR estimators.

The first class of hypotheses involves a continuum of quantile levels, $\tau \in \mathcal{T}$, whereas the second involves a finite collection of quantile levels, $\bar{\tau}_{m} = (\tau_1,\ldots,\tau_m)$. For the first class, we develop inference based on the QR process and refer to the resulting procedures as \emph{uniform-in-$\tau$} tests. For the second class of hypotheses, we stack the moment conditions across the finite set of quantile levels and show that the resulting Wald statistics have asymptotically pivotal fixed-smoothing limits. We refer to these procedures as \emph{stack-Wald} tests. Specifically, we consider the following hypotheses:
\begin{enumerate}
    \item[(i)] For a continuum of quantile levels: 
        \begin{itemize}
            \item Significance. $H_0 : r\left(\beta(\tau)\right) = 0$ for all $\tau \in \mathcal{T}$, vs. $H_{A}: r\left(\beta(\tau)\right) \neq 0$ for some $\tau \in \mathcal{T}$.
            \item Homogeneity. For given $j=1,\ldots,p$, $H_{0}: \beta_{j}(\tau_{2}) = \beta_{j}(\tau_{1}) \; \text{for all} \; \tau_{1} \neq \tau_{2} \, \in \mathcal{T}$ against $H_{A}: \beta_{j}(\tau_{2}) \neq \beta_{j}(\tau_{1}) \; \text{for some} \; \tau_{1} \neq \tau_{2} \, \in \mathcal{T}$.
            \item Monotonicity. For given $j=1,\ldots,p$, $H_{0}: \beta_{j}(\tau_{2}) - \beta_{j}(\tau_{1}) \geq 0 \; \text{for all} \; \tau_{2} > \tau_{1} \, \in \mathcal{T}$ vs. $H_{A}: \beta_{j}(\tau_{2}) - \beta_{j}(\tau_{1}) < 0 \; \text{for some} \; \tau_{2} > \tau_{1} \, \in \mathcal{T}$.
        \end{itemize}
    \item[(ii)] For a finite number of quantile levels:
        \begin{itemize}
            \item $H_0: r\left(\beta(\bar\tau_m)\right) = 0$ and $H_A: r\left(\beta(\bar\tau_m)\right) \ne 0$, where $\beta(\bar \tau_m) = (\beta(\tau_{1})',\cdots,\beta(\tau_{m})')^{\prime}$.
        \end{itemize}
\end{enumerate}

The uniform-in-$\tau$ tests explicitly account for the long-run covariance kernel $\Lambda(\tau_1,\tau_2)$ by estimating it consistently and using the estimate to simulate sample paths of the limiting Gaussian process $G(r,\tau)$. These simulated paths are then used to approximate the distributions of functionals of $\mathcal{W}^{\iota}(\tau)$ and $t^{\iota}(\tau)$, thereby yielding critical values for the uniform test statistics. By contrast, the stack-Wald tests circumvent the nuisance-parameter problem by exploiting the asymptotically pivotal fixed-smoothing limit obtained after stacking the moment conditions across a finite set of quantile levels.

\subsection{Uniform-in-$\tau$ tests}
\label{sec:uniform.in.tau}

Consider a hypothesis $H_0 : r(\beta(\tau)) = 0$ for all $\tau \in \mathcal{T}$, vs. $H_{A}: r(\beta(\tau)) \neq 0$ for some $\tau \in \mathcal{T}$. Typical examples include $r(\beta(\tau)) = \beta_{1}(\tau)$ or $r(\beta(\tau)) = \beta_{1}(\tau) - \beta_2(\tau)$ for a t-test, and $r(\beta(\tau)) = \left(\beta_{1}(\tau) , \beta_2(\tau)\right)$ for a Wald test. Because the QR cannot be estimated for a continuum of values for $\tau$, we first describe a two-step algorithm that discretizes the computation of $\widehat{\beta}(\tau)$, $\mathcal{W}^{\iota}(\tau)$, and $t^{\iota}(\tau)$.


\textbf{Step 1:} Partition $\mathcal{T}$ into a grid of $n$ equally spaced quantile levels $\{\tau_{1},\ldots,\tau_{n}\}$. 
For each $j \in \{1,\ldots,n\}$, solve the minimization problem \eqref{eq-beta.hat} to obtain $\widehat{\beta}(\tau_{j})$. Using estimators in \eqref{eq-D.hat}, \eqref{eq-Lambda.hat} and \eqref{eq-Lambda.hatOS}, compute $\widehat{D}(\tau_{j})$ and $\widehat{\Lambda}(\tau_{j})$, which in turn yield $\widehat{\Sigma}(\tau_{j})$. Construct the test statistics $\mathcal{W}^{\iota}(\tau_{j})$ and $t^{\iota}(\tau_{j})$ using \eqref{eq_wald} and \eqref{eq_t-stat}.

\textbf{Step 2:} Apply linear interpolation to obtain $\widehat{\beta}(\tau) = a(\tau) \widehat{\beta}(\tau_{j}) + (1-a(\tau)) \widehat{\beta}(\tau_{j+1})$ where $a(\tau) = (\tau_{j+1} - \tau)/(\tau_{j+1} - \tau_{j})$ for $\tau \in [\tau_{j},\tau_{j+1}]$. Likewise, obtain $\mathcal{W}^{\iota}(\tau) = a(\tau) \mathcal{W}^{\iota}(\tau_{j}) + (1-a(\tau)) \mathcal{W}^{\iota}(\tau_{j+1})$ and 
$t^{\iota}(\tau) = a(\tau) t^{\iota}(\tau_{j}) + (1-a(\tau)) t^{\iota}(\tau_{j+1})$.

The number of grid points, $n$, should be large enough to ensure accurate interpolation. In practice, larger values of $n$ could be used with larger $T$. \citet{NeocleousPortnoy08} show that if $n$ increases faster than $T^{1/4}$, linear interpolation does not result in any efficiency loss asymptotically. In this case, the interpolated estimator $\widehat{\beta}(\tau)$ achieves the same asymptotic efficiency as the estimator obtained by directly estimating $\beta(\tau)$.

We now turn to simulating the limiting stochastic processes of $\mathcal{W}^{\iota}(\tau)$ and $t^{\iota}(\tau)$ which depend on $G(r, \tau)$, a  Gaussian process that depends on the long-run covariance kernel $\Lambda(\tau_1,\tau_{2})$. Valid inference is based on generating sample paths from $G(r, \tau)$ by consistently estimating $\Lambda(\tau_1,\tau_{2})$ using several methods including HAC estimators or parametric approximations.


To illustrate the uniform-in-$\tau$ method, consider first the infeasible benchmark in which the long-run covariance kernel is known. When $G(r, \tau)$ is a square-integrable stochastic process with a continuous covariance function, $\Lambda(\tau_{1},\tau_{2})$ is a Mercer kernel, and can be represented as $ \Lambda(\tau_{1},\tau_{2}) = \sum_{l=1}^{\infty} \lambda_{l} e_{l}(\tau_1) e_{l}(\tau_2)'$ where $\lambda_{l} \geq 0$ are the eigenvalues and $e_{l}(\tau)$ are the associated orthonormal eigenfunctions. Under this representation, the Gaussian process $G(r, \tau)$ admits the Karhunen–Lo\`{e}ve expansion. Let $W_{l}(r)$ be independent Brownian motions, then
\begin{equation*}
G(r, \tau) = \sum_{l=1}^{\infty} \sqrt{\lambda_{l}}W_{l}(r) e_{l}(\tau).
\end{equation*}
This representation is valid because $G(r, \tau)$ is a mean-zero Gaussian process and the series representation reproduces its mean—it is a linear combination of independent, mean-zero normal variables—and covariance 
structure:
\begin{equation*}
\textnormal{Cov}\left(G(r_{1}, \tau_{1}), G(r_{2}, \tau_{2}) \right) = (r_{1} \wedge r_{2}) \sum_{l=1}^{\infty} \lambda_{l} e_{l}(\tau_1) e_{l}(\tau_2)' = (r_{1} \wedge r_{2}) \Lambda(\tau_{1},\tau_{2}).
\end{equation*}

The infinite series can be approximated using the leading $L$ eigenvalues and eigenfunctions: $ G(r, \tau) \approx \sum_{l=1}^{L} \sqrt{\lambda_{l}} W_{l}(r) e_{l}(\tau)$. Let $\widehat{\lambda}_{l}$ and $ \widehat{e}_{l}(\tau)$ denote some feasible approximations to the eigenvalues and eigenfunctions of the covariance kernel (see below). Then, the Gaussian limit process can be approximated by 
$G(r, \tau) \approx \sum_{l=1}^{L}  \sqrt{\widehat{\lambda}_{l}} W_{l}(r)\widehat{e}_{l}(\tau)  $. By generating independent sample paths of the Brownian motions $W_{l}(\cdot)$, $l = 1,\ldots,L$, one can simulate sample paths of $G(r, \tau)$. 

Note that the covariance kernel $\Lambda(\tau_1,\tau_2)$ and its eigenfunction ${e}_{l}(\tau)$ are defined on a continuum of quantile levels. To obtain the feasible approximations for ${\lambda}_{l}$ and ${e}_{l}(\tau)$, we approximate the covariance kernel by a sequence of growing finite-dimensional covariance matrix. In practice, let $\bar{\tau}_n = \{\tau_1,...,\tau_n\}$ denote a grid of quantile levels, and define the associated $np\times np$ covariance matrix:
\begin{align*}
   \Lambda(\bar{\tau}_{n},\bar{\tau}_{n}):= \left[
\begin{array}{ccc}
\Lambda(\tau_{1}, \tau_{1}) & \cdots & \Lambda(\tau_{1}, \tau_{n}) \\
\vdots & \vdots & \vdots \\
\Lambda(\tau_{n}, \tau_{1}) & \cdots & \Lambda(\tau_{n}, \tau_{n}) \\
\end{array}
\right],
\end{align*}%
where the $(i,j)$-th block, $\Lambda(\tau_{i},\tau_{j})$, is defined in (\ref{eq:cov_kernel}). To consistently estimate the covariance matrix, we can use either the kernel or OS variance estimator, and we illustrate using the kernel estimator. Let $\widehat{\Lambda}^{\rm KE}(\bar\tau_{n} , \bar\tau_{n})$ denote the HAC variance matrix estimator with the $(i,j)$-th block given by:
\begin{align}
    \widehat{\Lambda}^{\rm KE}(\tau_i,\tau_j) = \frac{1}{T}\sum_{t=1}^{T}\sum_{s=1}^{T}\mathcal{K}\left(\frac{t-s}{M}\right)\widehat{z}_{t}(\tau_i)\widehat{z}_{s}(\tau_j)^{\prime}.\label{lambda_est}
\end{align}

For estimation of this nuisance parameter, we adopt the HAC procedure of \cite{galvao2024hac}, employing a small-b bandwidth choice with $b  = o(1)$.\footnote{Consistency of the HAC estimator at fixed quantile levels under a small-b approximation is established in \cite{galvao2024hac}. While our theoretical development focuses on a fixed-b limits for tests that are uniform in $\tau$, we additionally use the small-b consistency of the HAC estimator to estimate nuisance parameters needed to simulate the non-pivotal fixed-b limits. A similar approach by \cite{chen2024fixed} lead to improved finite sample performance in two-way clustered panel settings.} 
Let $\widehat{\Lambda}^{\rm KE}(\bar\tau_{n} , \bar\tau_{n}) = \widehat{\mathcal{E}}_n \widehat\Lambda_n  \widehat{\mathcal{E}}^{\prime}_n$ denote the eigenvalue decomposition, where $\widehat{\mathcal{E}}_n$ contains eigenvectors, and $\widehat\Lambda_n$ is the diagonal matrix of eigenvalues. Let $(\widehat{e}_{l}(\tau_1)',\ldots,\widehat{e}_{l}(\tau_n)')^{\prime}$ denote the stacked eigenvectors associated with the $l$-th largest eigenvalues. For $\tau \in [\tau_{j},\tau_{j+1}]$, define its linear interpolation by 
$\widehat{e}_{l}(\tau) = a(\tau) \widehat{e}_{l}(\tau_{j}) + (1-a(\tau))\widehat{e}_{l}(\tau_{j+1})$. Sample paths of $G(r,\tau)$ can be simulated over a grid of $r$ values using the first $L$ eigenvectors associated with the $L$ largest eigenvalues. The truncation level $L$ serves as a tuning parameter and may be selected using data-driven criteria, as explained in Section~\ref{sec:sim}. The tied-down version of the process can be obtained as $\widetilde{G}(r, \tau) = G(r, \tau) - r G(1, \tau)$. Finally, combining the simulated sample paths with the consistent estimator $\widehat{D}(\tau)$, we can approximate the distributions of the test statistics in \eqref{wald_limit} and \eqref{t_limit}. A detailed implementation algorithm is given below. 

\vskip 0.1in

\begin{algorithm}[Uniform-in-$\tau$ Tests]
\label{alg:uniform_tests}
    \begin{enumerate}
	\item[(i)] Compute the HAC estimator $\widehat{\Lambda}^{\rm KE}(\tau_i,\tau_j)$ of ${\Lambda}(\tau_i,\tau_j)$, for all $i,j \in (1,\ldots,n)$, and construct the $np\times np$ matrix $\widehat{\Lambda}^{\rm KE}(\bar{\tau}_n , \bar{\tau}_n)$ as an estimator of $\Lambda(\bar{\tau}_n , \bar{\tau}_n)$. Obtain the leading eigenvalues $\widehat{\lambda}_{l}$ for $l=1,\ldots,L$ and the corresponding eigenfunctions $\widehat{e}_{l}(\tau)$, as described above. 

	\item[(ii)] Choose a fine grid $\{ r_0,r_1,...,r_{\bar{n}} \}$ such that $0=r_0<r_1< \cdots <r_{\bar{n}}=1$. Generate independent Brownian motions $W_{l}(r)$ over this grid and construct $\widehat{G}(r_{i}, \tau) = \sum_{l=1}^{L} \sqrt{\widehat{\lambda}_{l}} W_{l}(r_{i}) \widehat{e}_{l}(\tau)$. 

	\item[(iii)] Using the simulated sample paths $\widehat{G}(r,\tau)$ for all $r\in \{r_0,\ldots,r_{\bar{n}}\}$ and $\tau \in \mathcal{T}$, approximate the integrals in $P(\widetilde{G}(r,\tau),b)$ by Riemann sums. Evaluate the right-hand sides of \eqref{wald_limit} and \eqref{t_limit}, denoted as $\mathcal{W}^{\infty}(\tau)$ and $t^{\infty}(\tau)$, respectively, with $D(\tau)$ and $G(r,\tau)$ replaced by $\widehat{D}(\tau)$ and $\widehat{G}(r,\tau)$. Let $\mathcal{F}(\cdot)$ denote the function that defines the statistic of interest, e.g. $\mathcal{F}(f) = \sup_{\tau \in \mathcal{T}}|f(\tau)|$ or $\mathcal{F}(f) = \int_{\mathcal{T}}f(\tau) d\tau$. Compute $\mathcal{F}(\mathcal{W}^{\infty}(\tau))$ or $\mathcal{F}(t^{\infty}(\tau))$ and record their values.

     \item[(iv)] Repeat steps (ii) and (iii) for sufficiently many times and compute quantiles of $\mathcal{F}(\mathcal{W}^{\infty}(\tau))$ or $\mathcal{F}(t^{\infty}(\tau))$ to generate critical values. Compute $\mathcal{F}(\mathcal{W}^{\iota}(\tau))$ or $\mathcal{F}\left(t^{\iota}(\tau)\right)$ using the two-step algorithm and carry out the desired tests using their rejection rules.

\end{enumerate}
\end{algorithm}


\subsection{Stack-Wald Tests}
\label{sec:stacking}
Suppose we wish to test hypotheses about $\beta(\tau)$ at a finite collection of quantile levels, $\bar \tau_m=\{\tau_1,\ldots, \tau_m\}$. Here, $m$ is fixed, and the quantile levels need not be equally spaced. We can then use the stacking method to avoid estimating the nuisance function $\Lambda(\cdot,\cdot)$. To understand how the stacking approach works, fix $\tau$ and view $G(r, \tau)$ as a stochastic process indexed by $r \in [0,1]$. For a given $\tau$, the process has covariance function $\textnormal{Cov}\left(G(r_1,\tau),G(r_2,\tau)\right) = (r_1 \wedge r_2) \, \Lambda(\tau , \tau)$. This covariance structure suggests that $G(r,\tau)$ admits the representation $G(r,\tau) = \Gamma(\tau) W_{p}(r)$ where $\Gamma(\tau)\Gamma(\tau)^{\prime} = \Lambda(\tau,\tau)$. This observation can be extended to joint hypotheses involving $\beta_{0}(\tau_{1}),\ldots,\beta_{0}(\tau_{m})$. In this case, one can rely on pointwise results for the estimators at these given quantile levels. This approach yields asymptotic pivotal test statistics. 

The test for the joint hypothesis relies on the stacked moment functions:
\[
s_{t}(\bar\tau_{m}) = 
\left(
\begin{array}{c}
x_{t}(\tau_1 - I(y_{t} \leq x_{t}^{\prime}\beta_{0}(\tau_1)))  \\
\vdots \\
 x_{t} (\tau_{m} - I(y_{t} \leq x_{t}^{\prime}\beta_{0}(\tau_{m})))
\end{array}
\right),
\]
which is a $mp \times 1$ vector. Using the moment functions, define the partial sum process:
\begin{equation*}
S([rT],\bar\tau_{m}) = \frac{1}{\sqrt{T}} \sum_{t=1}^{[rT]} s_{t}(\bar\tau_{m}). 
\end{equation*}

\noindent Note that the stacked score vector collects information across quantile levels. Its long-run variance matrix is $\Lambda(\bar\tau_m,\bar\tau_m)$, the finite-dimensional counterpart of the covariance kernel $\Lambda(\tau_1,\tau_2)$. Because we take $\bar\tau_m$ as given and fixed,  $S([rT],\bar\tau_{m})$ is the partial sum of a fixed-dimensional vector of random variables. Under suitable mixing and moment conditions (Theorem 16.4 of \cite{hansen2022econometrics}), this partial sum empirical process can be shown to converge to a mean-zero Gaussian process 
$\bar{\Lambda}^{1/2} W_{mp}(r)$ where $\bar{\Lambda}^{1/2}$ denotes $\Lambda^{1/2}(\bar\tau_m , \bar\tau_m) $ and $W_{mp}$ is $mp$-dimensional standard Wiener process, i.e., $S([rT],\bar\tau_{m}) \Rightarrow \bar{\Lambda}^{1/2} W_{mp}(r)$. Therefore, cross-quantile dependence is transformed into a standard multivariate long-run variance problem.

Let $D(\bar{\tau}_m)$ denote the $mp \times mp$ block diagonal matrix whose $j$-th diagonal block is $D(\tau_j)$. 
By standard fixed-b and fixed-K arguments, we obtain a set of results corresponding to Theorems \ref{thm_1} and \ref{thm_2} for the given set of quantile levels $\bar\tau_m$:
\begin{align} 
     \sqrt{T}\left(\widehat{\beta}(\bar{\tau}_m)-\beta(\bar\tau_m)\right) & \Rightarrow D(\bar\tau_{m})^{-1} \bar{\Lambda}^{1/2} W_{mp}(1), \label{eq_thm_stack1} \\
     \widehat{S}_{[rT]}(\bar\tau_{m}) & \Rightarrow \bar{\Lambda}^{1/2} \left(W_{mp}(r) - r W_{mp}(1)\right) = \bar{\Lambda}^{1/2} \widetilde{W}_{mp}(r), \label{eq_thm_stack2} \\
    \widehat \Lambda^{\rm KE}(\bar\tau_m , \bar\tau_m) & \Rightarrow  \bar{\Lambda}^{1/2} P(\widetilde{W}_{mp}(r),b)(\bar{\Lambda}^{1/2})^{\prime}, \label{eq_thm_stack3} \\
    \widehat \Lambda^{\rm OS}(\bar\tau_m , \bar\tau_m) & \Rightarrow  \bar{\Lambda}^{1/2} \frac{1}{K} \sum_{k=1}^K W_{mp}^{(k)}(1)W_{mp}^{(k)}(1)' (\bar{\Lambda}^{1/2})^{\prime}, \label{eq_thm_stack4}
\end{align} 
where each $W_{mp}^{(k)}(1)$ is an independent copy of $W_{mp}(1)$. Define the Wald statistic for the set of quantile levels $\bar\tau_m$ as: for $\iota ={\rm KE, OS}$,
\begin{align}
	\mathcal{W}^{\iota}(\bar\tau_m) = T \ r(\widehat{\beta}(\bar\tau_m))^{\prime} \left[R(\widehat{\beta}(\bar\tau_m)) \, \widehat{D}^{-1}(\bar\tau_m) \widehat{\Lambda}^{\iota}(\bar\tau_m , \bar\tau_m)\widehat{D}^{-1}(\bar\tau_m) \, R(\widehat{\beta}(\bar\tau_m))^{\prime}\right]^{-1} r(\widehat{\beta}(\bar\tau_m)) \label{eq_wald_m}.
\end{align}
In the case of a single restriction, $q=1$, the corresponding $t$-statistic is 
\begin{align}
	t^{\iota}(\bar\tau_m) = \frac{\sqrt{T} \ r(\widehat{\beta}(\bar\tau_m))}{\sqrt{R(\widehat{\beta}(\bar\tau_m)) \, \widehat{D}^{-1}(\bar\tau_m) \widehat{\Lambda}^{\iota}(\bar\tau_m , \bar\tau_m)\widehat{D}^{-1}(\bar\tau_m) \, R(\widehat{\beta}(\bar\tau_m))^{\prime}}}.  \label{eq_t-stat_m}
\end{align}%
Here, $\widehat D(\bar{\tau}_{m})$ is a consistent estimator of $D(\bar\tau_m)$ with its $i$-th diagonal block defined as 
$\widehat{D}(\tau_i) =\frac{1}{hT}\sum_{t=1}^{T} k_{D}\left(\frac{\widehat{e}_{t}(\tau_i)}{h}\right)x_{t}x_{t}^{\prime}$. Following \citet{Kato12}, $ \widehat D(\bar{\tau}_m) \overset{p}{\to} D(\bar\tau_m)$. $\widehat{\Lambda}(\bar\tau_m , \bar\tau_m)$ is the HAR variance matrix estimator with the $(i,j)$-th block given by \eqref{lambda_est}.

Define a $q \times q$ variance matrix $\Upsilon = \bar{R}_{0} D(\bar\tau_{m})^{-1} \bar{\Lambda}^{1/2} (\bar{\Lambda}^{1/2})^{\prime} D(\bar\tau_{m})^{-1} \bar{R}_{0}^{\prime}$, where $\bar{R}_0:=\partial r(\beta(\bar\tau_m))/\partial \beta(\bar\tau_m)^{\prime}$. Note that $\bar{R}_0 D(\bar\tau_{m})^{-1} \bar{\Lambda}^{1/2}W_{mp}(1)$ is a $q \times 1$ vector of Brownian motions with variance $\Upsilon$. Because a finite linear combination of Gaussian processes is also Gaussian, $\bar{R}_0 D(\bar\tau_{m})^{-1} \bar{\Lambda}^{1/2}W_{mp}(1)$ can be written as $\Upsilon^{1/2} W_{q}(1)$ where $W_{q}(.)$ is a standard $q$-dimensional Wiener process and $\Upsilon^{1/2}(\Upsilon^{1/2})^{\prime} = \Upsilon$. Therefore, combining \eqref{eq_thm_stack1} -- \eqref{eq_thm_stack4} gives
\begin{align}
       \mathcal{W}^{\rm KE}(\bar\tau_m) & \Rightarrow \left(\Upsilon^{1/2} W_{q}(1)\right)^{\prime}  \times \left[\Upsilon^{1/2} P( \widetilde{W}_{q}(r),b) (\Upsilon^{1/2})^{\prime} \right]^{-1} \times \left(\Upsilon^{1/2} W_{q}(1)\right)  \nonumber \\
       &= W_{q}(1)^{\prime}\left[P(\widetilde{W}_{q}(r),b)\right]^{-1}W_{q}(1), \label{wald_limit_pw_ke}
       \end{align}
 
 \begin{align}      
    \mathcal{W}^{\rm OS}(\bar\tau_m) & \Rightarrow \left(\Upsilon^{1/2} W_{q}(1)\right)^{\prime}  \times \left[\Upsilon^{1/2} \frac{1}{K}\sum_{k=1}^K W_q^{(k)}(1)W_q^{(k)}(1)' (\Upsilon^{1/2})^{\prime} \right]^{-1} \times \left(\Upsilon^{1/2} W_{q}(1)\right)  \nonumber \\
       &= W_{q}(1)^{\prime}\left[\frac{1}{K}\sum_{k=1}^K W_q^{(k)}(1)W_q^{(k)}(1)'\right]^{-1}W_{q}(1), \label{wald_limit_pw_os} 
\end{align}%
which are asymptotically pivotal. Likewise,
\begin{align}
       t^{\rm KE}(\bar\tau_m) \Rightarrow & \frac{W_{q}(1)}{\sqrt{P(\widetilde{W}_{q}(r),b)}},\label{t_limit_pw_ke} \\
        t^{\rm OS}(\bar\tau_m) \Rightarrow & \frac{W_{q}(1)}{\sqrt{\frac{1}{K}\sum_{k=1}^KW_q^{(k)}(1)W_q^{(k)}(1)'}},\label{t_limit_pw_os}
\end{align}%
which are again asymptotically pivotal. The limits given by \eqref{wald_limit_pw_ke} and \eqref{t_limit_pw_ke} are equivalent to the fixed-b limits obtained by \cite{KieferVogelsang05}. The inference based on these pivotal fixed-b limits can be conducted using the tabulated fixed-b critical values from, for example, Table B of \cite{Vogelsang2012_wp}. The limits given by \eqref{wald_limit_pw_os} are the same as Hotelling's T-squared distribution, which implies $\frac{K-q+1}{qK} \mathcal{W}^{\rm OS}(\bar\tau_{m}) \Rightarrow F_{q,K-q+1}$. Similarly, the limit given by \eqref{t_limit_pw_os} implies ${t}^{\rm OS}(\bar\tau_{m}) \Rightarrow t_{K}$. Therefore, inference for a finite set of quantile effects is based on standard fixed-smoothing limiting distributions under both kernel-HAR and OS-HAR methods. Because the inference procedures in \cite{hwang2025har} are developed for a single, given quantile ($m=1$), their fixed-smoothing results can be viewed as a special case of the stacking method.

\subsection{Uniform \& Simultaneous Confidence Bands}
\label{sec:band}

This section develops uniform and simultaneous confidence bands for quantile regression coefficients. 
Uniform confidence bands provide a natural visualization of estimation uncertainty across quantiles and offer a useful exploratory assessment of effects throughout the outcome distribution. We begin by considering a simultaneous confidence band for a single slope coefficient evaluated at a finite collection of quantile levels, $\beta_{j}(\bar{\tau}_{m}):=(\beta_{j}(\tau_{1}),\cdots,\beta_{j}(\tau_{m}) )^{\prime}$, where $1 \leq j \leq p$. For a given $\alpha \in (0,1)$, define the rectangular confidence region
$$I_{\alpha}(\bar{\tau}_{m}) = \prod_{i=1}^m [\widehat{\beta}_j(\tau_i)-\widehat{\sigma}_j(\tau_i)c_{\alpha}, \widehat{\beta}_j(\tau_i)+\widehat{\sigma}_{j}(\tau_i)c_{\alpha}] ,
$$
where $\widehat{\sigma}_j(\tau) := \left[\widehat D(\tau)^{-1} \widehat{\Lambda}(\tau)\widehat D(\tau)^{-1}/T\right]_{jj}^{1/2}$ and $c_\alpha > 0$ is a critical value to be specified. If $c_\alpha$ is chosen such that $\text{Pr}\left(\beta_{j}(\bar{\tau}_{m}) \in I_{\alpha}(\bar{\tau}_{m})\right) \geq 1-\alpha$  asymptotically, then $I_{\alpha}(\bar{\tau}_{m})$ is a $(1-\alpha)$ simultaneous confidence band for $\beta_j(\bar{\tau}_{m})$. The following proposition shows how to choose $c_{\alpha}$. Let $t_{j}^{\infty}(\tau)$ denote a realization from the limiting distribution of the t-statistic defined in Step (iii) of Algorithm~\ref{alg:uniform_tests}. 

\begin{proposition}
\label{prop:supt_direct}
    Suppose the conditions of Proposition~\ref{prop:fixedb}
hold. Then,
\begin{align}
	\textnormal{Pr}\left(\beta_{j}(\bar{\tau}_{m}) \in I_{\alpha}(\bar{\tau}_{m})\right) \to \textnormal{Pr}\left( \max_{1\leq i \leq m} \left| t_j^\infty(\tau_i)\right| \leq c_{\alpha}\right) .\label{eq_prop:supt_direct} 
\end{align}
\end{proposition}
\noindent
Choosing $c_\alpha$ as the $(1-\alpha)$ quantile of $\max_{1\leq i \leq m}\left|t_j^\infty(\tau_i)\right|$ yields $\textnormal{Pr}\left(\beta_{j}(\bar{\tau}_{m}) \in I_{\alpha}(\bar{\tau}_{m})\right) =1-\alpha$ asymptotically. This construction is known as the Sup-t approach; see \cite{montiel2019simultaneous}. 

We now extend this method to a uniform confidence band over the entire quantile index $\mathcal{T}$. A $(1-\alpha)$ uniform confidence band for $\beta_{j}(\tau)$, $\tau \in \mathcal{T}$, consists of functions $l_{\alpha}(\tau)$ and $u_{\alpha}(\tau)$ satisfying $\text{Pr}\left(\beta_{j}(\tau) \in [l_{\alpha}(\tau) , u_{\alpha}(\tau)] \; \text{for all} \; \tau \in \mathcal{T}\right) \geq 1-\alpha$ asymptotically. Proceeding analogously to Proposition~\ref{prop:supt_direct}, we take $c_{\alpha}$  as the $(1-\alpha)$ quantile of $\sup_{\tau \in \mathcal{T} } \left| t_j^\infty(\tau)\right|$ and define
\begin{equation*}
I_{\alpha}(\mathcal{T}) = [\widehat{\beta}_j(\tau)-\widehat{\sigma}_j(\tau)c_{\alpha} \, , \, \widehat{\beta}_j(\tau)+\widehat{\sigma}_{j}(\tau)c_{\alpha}] \, , \quad \tau \in \mathcal{T}.
\end{equation*}
By construction, $I_{\alpha}(\mathcal{T})$ is a $(1-\alpha)$ uniform confidence band for $\beta_{j}(\tau)$. The algorithm for constructing simultaneous and uniform confidence bands is summarized below.

\vskip 0.1in

\begin{algorithm}[Simultaneous and uniform confidence bands]
\label{alg:confidence_bands}
Implement steps (i)--(iii) of Algorithm~\ref{alg:uniform_tests}. Repeat steps (ii) and (iii) for sufficiently many times. Let $c_\alpha$ denote the $(1-\alpha)$ sample quantile of $\mathcal{F}(t^{\infty}(\tau))=\max_{1\leq i \leq m}\left|t_j^\infty(\tau_i)\right|$ or $\sup_{\tau \in \mathcal{T} } \left| t_j^\infty(\tau)\right|$. Construct $I_{\alpha}(\bar{\tau}_{m})$ and $I_{\alpha}(\mathcal{T})$ as described above. 
\end{algorithm}

\begin{remark}
Uniform confidence bands for quantile processes have been proposed by \cite{2002.Koenker.Xiao}, \cite{2015.Qu.Yoon}, and \cite{Belloni.Chernozhukov.2019} for the case of i.i.d. data. The results presented here extend those constructions to quantile regressions using time-series data.
\end{remark}

\begin{remark}
A simultaneous confidence band can also be constructed using Wald statistics as discussed by \cite{montiel2019simultaneous}. We found that the Wald projection approach in our quantile setting tends to produce overly conservative simultaneous bands, confirming a similar finding reported in \cite{montiel2019simultaneous}. We therefore do not pursue this method further.
\end{remark}

The simultaneous confidence bands developed here can be viewed as a quantile regression analogue of the Sup-t method recommended by \cite{montiel2019simultaneous}. Conceptually, Algorithm~\ref{alg:confidence_bands} closely follows their Sup-t procedure. The same idea can also be extended to other inference problems, including tests of  homogeneity and monotonicity; see Section \ref{sec:test_shape} for details. Because these procedures are all based on the supremum of the $t$-statistic over the quantile index, we refer to them collectively as \emph{Sup-t methods}.

The simultaneous confidence bands developed here are constructed using the uniform-in-$\tau$ approach, even when interest is restricted to a finite collection of quantile levels, $\beta(\bar\tau_m)$. An alternative would be to construct simultaneous confidence bands based on the stack-Wald approach. However, unlike the Wald statistics considered in Section~\ref{sec:stacking}, Sup-t statistics do not enjoy nuisance-parameter cancellation under the stacking approach. As a result, both the uniform-in-$\tau$ and stack-Wald approaches require estimation of the long-run covariance structure, eliminating the principal advantage of the stacking method. 




\subsection{Testing Shape Restrictions}
\label{sec:test_shape}

This section demonstrates how the uniform-in-$\tau$ method can be used to test economically meaningful shape restrictions on the quantile effects. We first consider the significance hypothesis, $H_0 : \beta_{j}(\tau) = 0$ for all $\tau \in \mathcal{T}$ against $H_A : \beta_{j}(\tau) \neq 0$ for some $\tau \in \mathcal{T}$. Under the null hypothesis, the covariate $x_{j}$ has no effect at any quantile in $\mathcal{T}$, whereas the alternative allows the effect to be nonzero at some unknown part of the outcome distribution. 

Next, consider the homogeneity hypothesis, $H_{0}: \beta_{j}(\tau_{2}) = \beta_{j}(\tau_{1}) \; \text{for all} \; \tau_{1} \neq \tau_{2} \, \in \mathcal{T}$ against $H_{A}: \beta_{j}(\tau_{2}) \neq \beta_{j}(\tau_{1}) \; \text{for some} \; \tau_{1} \neq \tau_{2} \, \in \mathcal{T}$. The null states that the quantile effect is constant across quantile indices, while the alternative allows the effect to vary across the outcome distribution. Finally, consider the monotonicity hypothesis, $H_{0}: \beta_{j}(\tau_{2}) - \beta_{j}(\tau_{1}) \geq 0 \; \text{for all} \; \tau_{2} > \tau_{1} \, \in \mathcal{T}$ vs. $H_{A}: \beta_{j}(\tau_{2}) - \beta_{j}(\tau_{1}) < 0 \; \text{for some} \; \tau_{2} > \tau_{1} \, \in \mathcal{T}$. Under the null, the quantile effect is weakly increasing in the quantile index, whereas the alternative allows for violations of this monotonicity. 

To implement the tests, consider a fine grid $\{\tau_{1},\ldots,\tau_{n}\} \subset \mathcal{T}$. Let $\widehat{\Sigma}_{j}(\tau_{i},\tau_{l})$ denote the estimated covariance between $\sqrt{T}\widehat{\beta}_{j}(\tau_{i})$ and $\sqrt{T}\widehat{\beta}_{j}(\tau_{l})$. Define the quantities
\begin{equation*}
 \widehat{\sigma}_{j}(\tau_{i}) = \sqrt{\widehat{\Sigma}_{j}(\tau_{i},\tau_{i})}, \hspace{0.5cm} \widehat{\sigma}_{j}(\tau_{i},\tau_{l}) = \sqrt{\widehat{\Sigma}_{j}(\tau_{i},\tau_{i}) + \widehat{\Sigma}_{j}(\tau_{l},\tau_{l}) - 2 \widehat{\Sigma}_{j}(\tau_{i},\tau_{l})}.   
\end{equation*}
The hypotheses of effect significance, monotonicity, and homogeneity can then be tested using the following test statistics: 

\begin{enumerate}
\item[(i)] Significance: $t^{S}(\mathcal{T}) = \max_{1 \leq i \leq n} \sqrt{T} \left| \widehat{\beta}_{j}(\tau_{i})/\widehat{\sigma}_{j}(\tau_{i}) \right|$,
\item[(ii)] Homogeneity: $t^{H}(\mathcal{T}) = \max_{2 \leq i \leq n}  \sqrt{T} \left| \widehat{\beta}_{j}(\tau_{i}) - \widehat{\beta}_{j}(\tau_{i-1}) \right| /\widehat{\sigma}_{j}(\tau_{i},\tau_{i-1})$,
\item[(iii)] Monotonicity: $t^{M}(\mathcal{T}) = \min_{i > l} \sqrt{T} \left\{\widehat{\beta}_{j}(\tau_{i}) - \widehat{\beta}_{j}(\tau_{l})\right\}/\widehat{\sigma}_{j}(\tau_{i},\tau_{l})$.
\end{enumerate}

The limiting distributions of the test statistics are stated below. Let $e_{i}^{(n)}$ be the $n \times 1$ unit vector whose $i$-th element is $1$ and whose remaining elements are $0$. Similarly, let $e_{i}^{(p)}$ be the $p \times 1$ unit vector. Define $a_{i,j} = e_i^{(n)}\otimes e_j^{(p)}$ and $a_{i,l;j} = a_{i,j} - a_{l,j}$. Thus, $a_{i,j}$ selects the $j$-th coefficient associated with quantile ${\tau_i}$ from an $np$-dimensional vector.  Recall that $H(\tau,b) = D(\tau)^{-1} P\left(\widetilde{G}(r,\tau),b\right)D(\tau)^{-1}$.

\begin{proposition}
\label{prop-shape-restriction}
Under conditions in Theorem~\ref{thm_2}, as $T \rightarrow \infty$,
\begin{enumerate}
\item[(1)] If $\beta_{j}(\tau) = 0 \;$ for all $\tau \in \mathcal{T}$, $t^{S}(\mathcal{T}) \Rightarrow \max_{1 \leq i \leq n}  \left| a_{i,j}^{\prime} D(\bar{\tau}_{n})^{-1} G(1,\bar{\tau}_{n}) \right| / \sqrt{a_{i,j}^{\prime} H(\bar{\tau}_{n},b) a_{i,j} }$.
\item[(2)] If $\beta_{j}(\tau) = c \;$ for all $\tau \in \mathcal{T}$ and some constant $c$, 
\begin{equation*}
t^{H}(\mathcal{T}) \Rightarrow \max_{2 \leq i \leq n} \frac{\left|a_{i,i-1;j}^{\prime} D(\bar{\tau}_{n})^{-1}G(1,\bar{\tau}_{n})\right|}{\sqrt{a_{i,i-1;j}^{\prime} H(\bar{\tau}_{n},b) a_{i,i-1;j}}}.
\end{equation*}
\item[(3)] If $\beta_{j}(\tau) = c \;$ for all $\tau \in \mathcal{T}$  and some constant $c$ (this is the least favorable configuration), 
\begin{equation*}
t^{M}(\mathcal{T}) \Rightarrow \min_{i > l} \frac{\left\{a_{i,l;j}^{\prime} D(\bar{\tau}_{n})^{-1}G(1,\bar{\tau}_{n})\right\}}{\sqrt{a_{i,l;j}^{\prime} H(\bar{\tau}_{n},b) a_{i,l;j}}}.
\end{equation*}
\end{enumerate}
\end{proposition}

\begin{algorithm}[Testing Shape Restrictions]
\label{alg:shape_tests}
    \begin{enumerate}
\item[(i)] On a fine grid $\{\tau_{1},\ldots,\tau_{n}\}$, compute the test statistics $t^{S}(\mathcal{T})$, $t^{H}(\mathcal{T})$, and $t^{M}(\mathcal{T})$ corresponding to the hypotheses of effect significance, homogeneity, and monotonicity, respectively.. 
\item[(ii)] Draw a sample path $\widehat{G}(1,\tau)$ as described in Section~\ref{sec:uniform.in.tau}. Using this draw, evaluate the limiting distributions defined in Proposition~\ref{prop-shape-restriction} and record their values as $\mathcal{G}^{S}$, $\mathcal{G}^{H}$ and $\mathcal{G}^{M}$.
\item[(iii)] Repeat step (ii) many times. Let $q_{\alpha}^{S}$ and $q_{\alpha}^{H}$ denote the $(1-\alpha)$-th quantiles of $\mathcal{G}^{S}$ and $\mathcal{G}^{H}$, respectively. Let $q_{1-\alpha}^{M}$ denote the $\alpha$-th quantiles of $\mathcal{G}^{M}$. Reject the null hypothesis if $t^{S}(\mathcal{T}) > q_{\alpha}^{S}$ (significance), $t^{H}(\mathcal{T}) > q_{\alpha}^{H}$ (homogeneity) or $t^{M}(\mathcal{T}) < q_{1-\alpha}^{M}$ (monotonicity).
\end{enumerate}
\end{algorithm}

\subsection{Test-Based Data-Driven Bandwidth Rules}
\label{sec:bandwidth}

We use test-based, data-driven rules to select the bandwidth $M$ and the number of series terms $K$. These rules balance size distortion against power and have been developed for kernel-HAR inference \citep{SPJ2008,S2014b} and OS-HAR inference \citep{sun2011robust,hwang2025har}. We extend them to multiple-quantile inference by using the stacked QR score process and test-specific restriction matrices for the Wald, Sup-t, homogeneity, and monotonicity tests. Because existing optimality results rely on higher-order expansions of pivotal limiting distributions, whereas our uniform-in-$\tau$ limits are generally non-pivotal, we view the resulting choices as practical guidelines rather than formally optimal rules.

The procedure approximates the serial dependence of the stacked QR score process with a VAR(1) model. The estimated VAR parameters provide plug-in estimates of the long-run variance and smoothing bias, which are used to select $M$ and $K$ subject to a prespecified degree of size distortion. Detailed formulas and implementation steps are provided in Section C of the Supplemental Appendix.

\section{Simulation study}
\label{sec:sim}

We carried out a simulation study that focuses on four issues: (i) the performance of the proposed fixed-smoothing approaches relative to the existing HAC (small-b or large-K) approaches, (ii) the properties of the Wald tests for testing restrictions across a given set of quantiles and the properties of the uniform-in-$\tau$ tests for the significance, homogeneity, and monotonicity hypotheses, (iii) the performance of the data-dependent bandwidth selection procedures, and (iv) the performance of the kernel HAR and OS HAR estimators. Results are given for the following location-shift model: $y_{t} = \beta_0 + \beta_1 x_t + e_{t}$ where $e_{t} = \rho e_{t-1} + v_{t}$ and $x_t = \rho x_{t-1} + \epsilon_{t}$, with $v_{t} \sim \text{N}\left(0,1-\rho^2 \right)$ and $\epsilon_{t} \sim \text{N}\left(0,1-\rho^2 \right)$.

The regression coefficients are set to $\beta_0 = 0$, $\beta_1 = 1$, implying that $\beta_{1}(\tau) = 1$ for all $\tau$. The autoregressive coefficients for $x_t$ and $e_t$ take common values $\rho \in \{0.0, 0.5, 0.8\}$. The sample sizes considered are $T \in \{200, 500, 800\}$. For HAR variance matrix estimation, we use the Bartlett kernel with bandwidth $M$, where $M$ is selected using the bandwidth selection procedure in Section~\ref{sec:bandwidth}. For OS variance matrix estimation, we employ the orthonormal DCT-II cosine basis: $\phi_k(r) = \sqrt{2}\cos\left(\pi k\left(r-\frac{1}{2T}\right)\right), \; k=1,2,...,T-1$, where $K$ is determined by the procedure in Section~\ref{sec:bandwidth}. 
To ensure invertibility of the OS variance estimator, $K$ must be at least as large as the number of restrictions under consideration.\footnote{For the Wald test, this number is equal to $m$. For the uniform-in-$\tau$ test, invertibility is not required.} 
All simulation results are based on $5,000$ Monte Carlo replications.

\begin{table}[!htb]
\centering
\caption{Inference for the single quantile}
\label{tab:sim_single_slope}
\footnotesize
\setlength{\tabcolsep}{3.2pt}
\begin{tabular}{@{}rrrrrrrrrrrr@{}}
\toprule
\multicolumn{2}{c}{Setting}
  & \multicolumn{3}{c}{fixed-$b$}
  & \multicolumn{3}{c}{fixed-$K$}
  & \multicolumn{2}{c}{small-$b$}
  & \multicolumn{2}{c}{large-$K$} \\
\cmidrule(lr){1-2}\cmidrule(lr){3-5}\cmidrule(lr){6-8}
\cmidrule(lr){9-10}\cmidrule(lr){11-12}
  & & \multicolumn{3}{c}{kernel-HAR}
  & \multicolumn{3}{c}{OS-HAR}
  & \multicolumn{2}{c}{kernel-HAC}
  & \multicolumn{2}{c}{OS-HAR} \\
\cmidrule(lr){3-5}\cmidrule(lr){6-8}\cmidrule(lr){9-10}\cmidrule(lr){11-12}
$T$ & $\rho$ & $M$ & 10\% & 5\% & $K$ & 10\% & 5\%
  & 10\% & 5\% & 10\% & 5\% \\
\midrule
\multicolumn{12}{@{}l}{\(\tau = 0.5\)} \\
\cmidrule(lr){1-12}
200 & 0.0 & 1 & 0.098 & 0.049 & 199 & 0.098 & 0.049 & 0.101 & 0.052 & 0.100 & 0.051 \\
200 & 0.5 & 7 & 0.113 & 0.064 & 49 & 0.116 & 0.061 & 0.131 & 0.074 & 0.122 & 0.068 \\
200 & 0.8 & 25 & 0.132 & 0.077 & 20 & 0.137 & 0.078 & 0.185 & 0.119 & 0.152 & 0.096 \\
\cmidrule(lr){1-12}
500 & 0.0 & 1 & 0.097 & 0.048 & 499 & 0.098 & 0.049 & 0.099 & 0.049 & 0.099 & 0.049 \\
500 & 0.5 & 7 & 0.111 & 0.059 & 119 & 0.114 & 0.060 & 0.121 & 0.064 & 0.117 & 0.061 \\
500 & 0.8 & 26 & 0.117 & 0.063 & 46 & 0.120 & 0.062 & 0.139 & 0.081 & 0.126 & 0.070 \\
\cmidrule(lr){1-12}
800 & 0.0 & 1 & 0.096 & 0.046 & 799 & 0.099 & 0.046 & 0.099 & 0.048 & 0.099 & 0.047 \\
800 & 0.5 & 15 & 0.097 & 0.051 & 134 & 0.103 & 0.050 & 0.107 & 0.059 & 0.105 & 0.052 \\
800 & 0.8 & 54 & 0.109 & 0.063 & 51 & 0.115 & 0.066 & 0.138 & 0.082 & 0.121 & 0.071 \\
\midrule
\multicolumn{12}{@{}l}{\(\tau = 0.75\)} \\
\cmidrule(lr){1-12}
200 & 0.0 & 1 & 0.112 & 0.062 & 199 & 0.112 & 0.061 & 0.114 & 0.063 & 0.113 & 0.063 \\
200 & 0.5 & 7 & 0.123 & 0.070 & 51 & 0.126 & 0.070 & 0.139 & 0.080 & 0.133 & 0.074 \\
200 & 0.8 & 24 & 0.147 & 0.089 & 20 & 0.150 & 0.090 & 0.203 & 0.130 & 0.168 & 0.105 \\
\cmidrule(lr){1-12}
500 & 0.0 & 1 & 0.101 & 0.050 & 499 & 0.102 & 0.051 & 0.103 & 0.052 & 0.103 & 0.052 \\
500 & 0.5 & 7 & 0.115 & 0.063 & 125 & 0.120 & 0.063 & 0.124 & 0.069 & 0.122 & 0.064 \\
500 & 0.8 & 25 & 0.131 & 0.075 & 48 & 0.137 & 0.074 & 0.154 & 0.093 & 0.143 & 0.085 \\
\cmidrule(lr){1-12}
800 & 0.0 & 1 & 0.100 & 0.049 & 799 & 0.100 & 0.050 & 0.100 & 0.051 & 0.100 & 0.050 \\
800 & 0.5 & 14 & 0.101 & 0.055 & 141 & 0.108 & 0.057 & 0.112 & 0.062 & 0.109 & 0.058 \\
800 & 0.8 & 51 & 0.117 & 0.067 & 54 & 0.121 & 0.068 & 0.143 & 0.086 & 0.126 & 0.074 \\
\midrule
\multicolumn{12}{@{}l}{\(\tau = 0.9\)} \\
\cmidrule(lr){1-12}
200 & 0.0 & 1 & 0.147 & 0.090 & 199 & 0.147 & 0.090 & 0.150 & 0.093 & 0.148 & 0.092 \\
200 & 0.5 & 5 & 0.159 & 0.102 & 57 & 0.163 & 0.103 & 0.175 & 0.114 & 0.168 & 0.110 \\
200 & 0.8 & 21 & 0.211 & 0.141 & 23 & 0.215 & 0.146 & 0.256 & 0.187 & 0.231 & 0.165 \\
\cmidrule(lr){1-12}
500 & 0.0 & 1 & 0.111 & 0.059 & 499 & 0.112 & 0.061 & 0.112 & 0.062 & 0.112 & 0.061 \\
500 & 0.5 & 5 & 0.132 & 0.075 & 143 & 0.136 & 0.077 & 0.141 & 0.080 & 0.139 & 0.079 \\
500 & 0.8 & 22 & 0.161 & 0.108 & 55 & 0.166 & 0.110 & 0.180 & 0.122 & 0.171 & 0.116 \\
\cmidrule(lr){1-12}
800 & 0.0 & 1 & 0.108 & 0.056 & 799 & 0.109 & 0.056 & 0.110 & 0.057 & 0.110 & 0.057 \\
800 & 0.5 & 11 & 0.117 & 0.064 & 161 & 0.121 & 0.064 & 0.126 & 0.070 & 0.124 & 0.065 \\
800 & 0.8 & 44 & 0.149 & 0.089 & 61 & 0.148 & 0.092 & 0.170 & 0.108 & 0.156 & 0.099 \\
\bottomrule
\end{tabular}

\par\medskip
{\footnotesize
 \parbox{0.80\linewidth}{%
      Note: The table reports rejection rates for the test of a single slope coefficient at $\tau \in \{0.5, 0.75, 0.9\}$ based on 5,000 simulation repetitions. The ``Small-b Kernel HAC'' method follows \cite{galvao2024hac}. The test-based $M$ and $K$ are computed in each replication, and the medians are reported. 
 }%
}
\end{table}


\vskip 0.1in

\textbf{Test for a single quantile: } As a benchmark, we first examine the performance of the tests for a single quantile level. As shown by \cite{hwang2025har}, the kernel-based approaches reduce to the conventional fixed-b method, with the Wald and t-tests simplifying to (\ref{wald_limit_pw_ke}) and (\ref{t_limit_pw_ke}) with $q=1$. The resulting limit distributions coincide with those in Kiefer and Vogelsang (2005), so the critical values reported in their Table 1 can be used directly. Tests based on the OS variance estimator similarly reduce to their usual $t$ and $F$ limits. Table~\ref{tab:sim_single_slope} reports results for $\tau \in \{0.5,0.75, 0.9\}$, corresponding to the center, shoulder, and right tail of the outcome distribution, respectively. The results for the left tail are nearly symmetric and are omitted to save space. 

In Table~\ref{tab:sim_single_slope}, we report the medians of the data-driven bandwidth, $M$, and number of series terms, $K$, across replications. When $\rho=0.0$, the data-driven methods tend to choose a small $M$ and a large $K$; whereas in the highly dependent case with $\rho=0.8$, larger $M$ and smaller $K$ are chosen. Overall, the bandwidths reflect the degree of serial dependence well. 

Table~\ref{tab:sim_single_slope} then reports the empirical null rejection rates at nominal levels of 10\% and 5\%. The label ``fixed-b kernel-HAR'' refers to the kernel-based method developed in Section~\ref{sec:stacking}. As a benchmark, we also report the rejection rates for the HAC (small-b) approach of \citet{galvao2024hac}, labeled ``small-b kernel-HAC''. Similarly, ``fixed-K OS-HAR'' denotes the OS variance approach developed in Section~\ref{sec:stacking} and is compared with the ``large-K OS-HAR'' method. We do not include the smooth block bootstrap method of \citet{GregoryLahiriNordman18} in the simulation comparison, as existing evidence suggests that its finite-sample size control is generally close to that of the HAC method of \citet{galvao2024hac}. See \citet{hwang2025har} and \citet{CaiLong2026} for the simulation evidence.

Relative to \cite{galvao2024hac}, the fixed-b methods provide substantial improvements in size control. For example, consider the highly dependent case with $\rho=0.8$, a nominal level of  5\%, and evaluation at $\tau=0.5$. The fixed-b kernel method reduces the empirical rejection rate from $0.119$ to $0.077$ when $T=200$, from $0.081$ to $0.063$ when $T=500$, and from $0.082$ to $0.063$ when $T=800$. Similar improvements are observed across all quantiles considered. The same pattern arises when moving from the large-K to the fixed-K OS-HAR method, although the improvement is somewhat smaller than for the kernel-based method. Overall, these results provide strong evidence that fixed-smoothing inference can deliver practically important gains in size control relative to conventional small-b or large-K inference.

In comparing the fixed-b and small-b kernel methods and the fixed-K and large-K OS methods, we use the same bandwidth values $M$ and $K$, respectively, in both procedures. This ensures that any differences in rejection rates are attributable to the inference method itself rather than differences in bandwidth choice. We also experimented with the MSE-optimal bandwidth proposed in \citet{galvao2024hac} and obtained qualitatively similar results.

\begin{table}[!htb]
\centering
\caption{Inference for the slope coefficient at two quantile levels}
\label{tab:sim_two_taus}
\footnotesize
\setlength{\tabcolsep}{3.5pt}
\begin{tabular}{@{}rrrrrrrrrrrrr@{}}
\toprule
\multicolumn{3}{c}{Setting}
  & \multicolumn{3}{c}{fixed-$b$}
  & \multicolumn{3}{c}{fixed-$K$}
  & \multicolumn{2}{c}{small-$b$}
  & \multicolumn{2}{c}{large-$K$} \\
\cmidrule(lr){1-3}\cmidrule(lr){4-6}\cmidrule(lr){7-9}
\cmidrule(lr){10-11}\cmidrule(lr){12-13}
  & & & \multicolumn{3}{c}{kernel-HAR}
  & \multicolumn{3}{c}{OS-HAR}
  & \multicolumn{2}{c}{kernel-HAC}
  & \multicolumn{2}{c}{OS-HAR} \\
\cmidrule(lr){4-6}\cmidrule(lr){7-9}\cmidrule(lr){10-11}\cmidrule(lr){12-13}
$(\tau_1,\tau_2)$ & $T$ & $\rho$ & $M$ & 10\% & 5\%
  & $K$ & 10\% & 5\% & 10\% & 5\% & 10\% & 5\% \\
\midrule
$(0.5,0.75)$ & 200 & 0.0 & 1 & 0.112 & 0.059 & 199 & 0.111 & 0.059 & 0.116 & 0.064 & 0.115 & 0.063 \\
 & 200 & 0.5 & 7 & 0.114 & 0.066 & 48 & 0.114 & 0.064 & 0.141 & 0.084 & 0.128 & 0.078 \\
 & 200 & 0.8 & 27 & 0.149 & 0.090 & 18 & 0.148 & 0.088 & 0.268 & 0.192 & 0.197 & 0.135 \\
\cmidrule(lr){2-13}
 & 500 & 0.0 & 1 & 0.106 & 0.052 & 499 & 0.106 & 0.056 & 0.107 & 0.058 & 0.107 & 0.058 \\
 & 500 & 0.5 & 7 & 0.110 & 0.057 & 121 & 0.114 & 0.058 & 0.124 & 0.067 & 0.121 & 0.062 \\
 & 500 & 0.8 & 28 & 0.124 & 0.073 & 42 & 0.128 & 0.074 & 0.170 & 0.108 & 0.149 & 0.089 \\
\cmidrule(lr){2-13}
 & 800 & 0.0 & 1 & 0.103 & 0.050 & 799 & 0.104 & 0.053 & 0.106 & 0.054 & 0.105 & 0.054 \\
 & 800 & 0.5 & 14 & 0.105 & 0.053 & 137 & 0.107 & 0.052 & 0.121 & 0.064 & 0.112 & 0.056 \\
 & 800 & 0.8 & 57 & 0.122 & 0.067 & 47 & 0.129 & 0.066 & 0.187 & 0.117 & 0.146 & 0.083 \\
\midrule
$(0.5,0.9)$ & 200 & 0.0 & 1 & 0.130 & 0.077 & 199 & 0.129 & 0.078 & 0.135 & 0.082 & 0.133 & 0.081 \\
 & 200 & 0.5 & 7 & 0.142 & 0.092 & 47 & 0.144 & 0.092 & 0.174 & 0.114 & 0.161 & 0.107 \\
 & 200 & 0.8 & 29 & 0.197 & 0.126 & 18 & 0.196 & 0.126 & 0.322 & 0.247 & 0.250 & 0.181 \\
\cmidrule(lr){2-13}
 & 500 & 0.0 & 1 & 0.112 & 0.055 & 499 & 0.113 & 0.058 & 0.114 & 0.059 & 0.114 & 0.060 \\
 & 500 & 0.5 & 7 & 0.126 & 0.068 & 117 & 0.126 & 0.069 & 0.142 & 0.082 & 0.136 & 0.076 \\
 & 500 & 0.8 & 30 & 0.150 & 0.091 & 42 & 0.153 & 0.092 & 0.207 & 0.138 & 0.173 & 0.112 \\
\cmidrule(lr){2-13}
 & 800 & 0.0 & 1 & 0.102 & 0.053 & 799 & 0.103 & 0.055 & 0.104 & 0.057 & 0.104 & 0.057 \\
 & 800 & 0.5 & 15 & 0.117 & 0.057 & 133 & 0.117 & 0.057 & 0.134 & 0.071 & 0.124 & 0.062 \\
 & 800 & 0.8 & 61 & 0.133 & 0.077 & 48 & 0.140 & 0.083 & 0.200 & 0.131 & 0.156 & 0.101 \\
\midrule
$(0.1,0.9)$ & 200 & 0.0 & 1 & 0.141 & 0.091 & 199 & 0.140 & 0.091 & 0.145 & 0.096 & 0.144 & 0.095 \\
 & 200 & 0.5 & 7 & 0.161 & 0.106 & 50 & 0.162 & 0.106 & 0.194 & 0.130 & 0.179 & 0.121 \\
 & 200 & 0.8 & 27 & 0.222 & 0.146 & 20 & 0.231 & 0.155 & 0.344 & 0.265 & 0.278 & 0.207 \\
\cmidrule(lr){2-13}
 & 500 & 0.0 & 1 & 0.119 & 0.065 & 499 & 0.119 & 0.069 & 0.120 & 0.070 & 0.120 & 0.070 \\
 & 500 & 0.5 & 7 & 0.142 & 0.080 & 124 & 0.145 & 0.081 & 0.155 & 0.094 & 0.150 & 0.088 \\
 & 500 & 0.8 & 28 & 0.176 & 0.111 & 47 & 0.183 & 0.117 & 0.233 & 0.161 & 0.203 & 0.140 \\
\cmidrule(lr){2-13}
 & 800 & 0.0 & 1 & 0.108 & 0.056 & 799 & 0.108 & 0.058 & 0.109 & 0.059 & 0.109 & 0.059 \\
 & 800 & 0.5 & 14 & 0.117 & 0.060 & 141 & 0.117 & 0.064 & 0.132 & 0.074 & 0.122 & 0.069 \\
 & 800 & 0.8 & 58 & 0.149 & 0.090 & 53 & 0.157 & 0.099 & 0.216 & 0.146 & 0.171 & 0.115 \\
\bottomrule
\end{tabular}

\par\medskip
{\footnotesize
 \parbox{0.8\linewidth}{
    Note: Empirical null rejection rates of the stack-Wald test for the null $H_0:\beta(\tau_1)=\beta(\tau_2)=1$ are reported for the values of $(\tau_1,\tau_2)$ shown in the table. Rejection rates are shown at nominal significance levels of 10\% and 5\%. $M$ and $K$ denote the median bandwidth values selected by the test-based bandwidth procedures in Section \ref{sec:bandwidth}.
 }%
}
\end{table}

\vskip 0.1in


\textbf{Stack-Wald tests for joint hypotheses:}  The Wald test in \eqref{eq_wald_m} is implemented using the two smoothing-based inference procedures, kernel-HAR and OS-HAR. We begin with a simple setting involving comparisons across two quantiles. The null hypothesis is $H_{0}: \beta_{1}(\tau_1) = 1$, $\beta_{1}(\tau_2) = 1$ so that the number of restrictions is $m=2$. We consider three pairs of $(\tau_1 , \tau_2)$: $(0.5, 0.75)$, $(0.5, 0.9)$, and $(0.1, 0.9)$, corresponding to joint equality tests for the center and shoulder of the distribution, the center and the tail, and the left and right tails, respectively. The Wald test has the usual fixed-b distribution and the critical values are provided by \citet{KieferVogelsang05}.

Table~\ref{tab:sim_two_taus} reports the empirical null rejection rates with the selected test-based bandwidth values. Overall, the fixed-b kernel and fixed-K OS approaches perform well, particularly when the sample size is relatively large. To illustrate, consider the case with $T=500$, $\rho=0.8$, and a nominal level of 5\%. For $(\tau_1, \tau_2) = (0.5, 0.75)$, the empirical rejection rates are $0.073$, and $0.074$ for the fixed-b kernel and fixed-K OS methods, respectively.  These rates compare favorably with the corresponding rejection rates of $0.108$ and $0.089$ for the small-b kernel and large-K OS methods.
For $(\tau_1 , \tau_2) = (0.5,0.9)$, the rejection rates of the fixed-smoothing methods increase slightly to $0.091$ and $0.092$, respectively. This suggests that joint inference becomes more challenging when one of the quantiles is located farther in the tail of the distribution. Comparing the two fixed-smoothing approaches, the kernel and OS methods exhibit similar performance, with neither consistently dominating the other.

One natural question is how the Wald test performs as the number of restrictions, or equivalently the number of quantile levels, increases. 
To investigate this issue, we test the null hypothesis $H_0 : \beta_1(\tau_1)= \cdots =\beta_1(\tau_m) = 1$ over a grid of quantile levels $\{\tau_1, \ldots, \tau_{m}\} \in [0.1,0.9]$, where the number of restrictions is $m \in \{5,9,17\}$. 
The results are discussed in detail in Section~\ref{sec:appendix_sim} of the Supplemental Appendix. The results in Table~\ref{tab:sim-stacked} show that the empirical null rejections of the Wald test increases as the number of restrictions $m$ grows. For example, consider the case with $T=500$, $\rho=0.8$, and a nominal level of 5\%. The rejection rates are $0.112$, $0.138$, and $0.179$ for $m = 5$, $9$, and $17$, respectively, for the kernel-HAR method. 

\vskip 0.1in

\textbf{Sup-t test and uniform confidence bands: } For the Sup-t test, the null hypothesis is $H_0: \beta_{1}(\tau) = 1$ for all $\tau \in \mathcal{T}$, where the quantile range of interest is $\mathcal{T} = [0.1,0.9]$. The simulation design follows that for the Wald test, with three sample sizes and three values of the autocorrelation coefficient. Results for the Sup-t test are given for a range of bandwidth values. In practice, implementation of the Sup-t test requires a finite grid of quantile levels together with interpolation over the grid. We therefore consider grids $\{\tau_1 , \ldots, \tau_n\} \in \mathcal{T}$ with $n \in \{5,9,17\}$. Unlike the Wald test, however, the number of grid points $n$ does not correspond to the number of restrictions under the null. As shown below, the results are largely insensitive to the choice of $n$. 

The critical values are obtained by simulation as described in Section~\ref{sec:uniform.in.tau}. Specifically, to approximate the Gaussian process via the Karhunen–Lo\`{e}ve expansion, one must choose $L$, the number of leading eigenvalues in the eigenvalue decomposition of the covariance matrix, as described in Algorithm~\ref{alg:uniform_tests}. We set $L$ to be the smallest integer that explains $99\%$ of the sum of all positive eigenvalues. Recall our simulations use 5,000 replications. Within each replication, the critical value is obtained from 5,000 simulated statistics, with the Brownian motions approximated on a grid of 1,000 points. The nuisance parameter required in (11) is estimated with the MSE-optimal bandwidth of \cite{galvao2024hac}.


\begin{figure}[h]
    \centering

    \begin{subfigure}{\linewidth}
        \centering
        \includegraphics[width=0.9\linewidth]{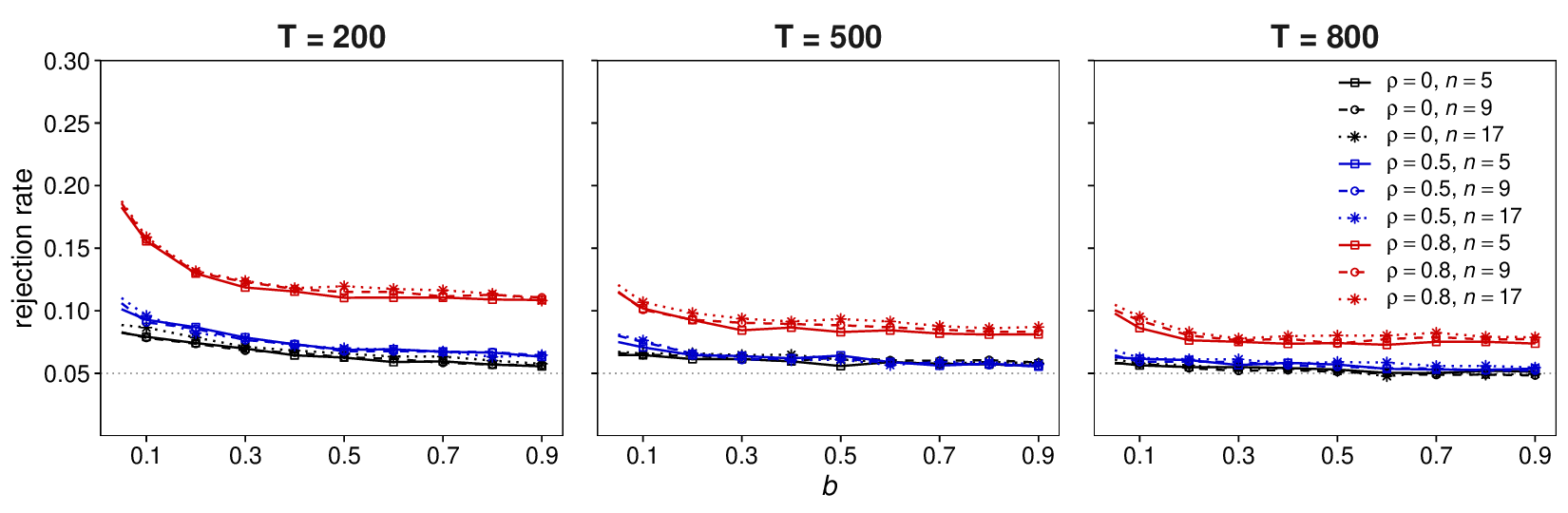}
        \caption{Rejection rates of the Sup-t test; kernel HAR estimator}
        \label{fig:4_supt_kernel}
    \end{subfigure}

    \vspace{0.8em}

    \begin{subfigure}{\linewidth}
        \centering
        \includegraphics[width=0.9\linewidth]{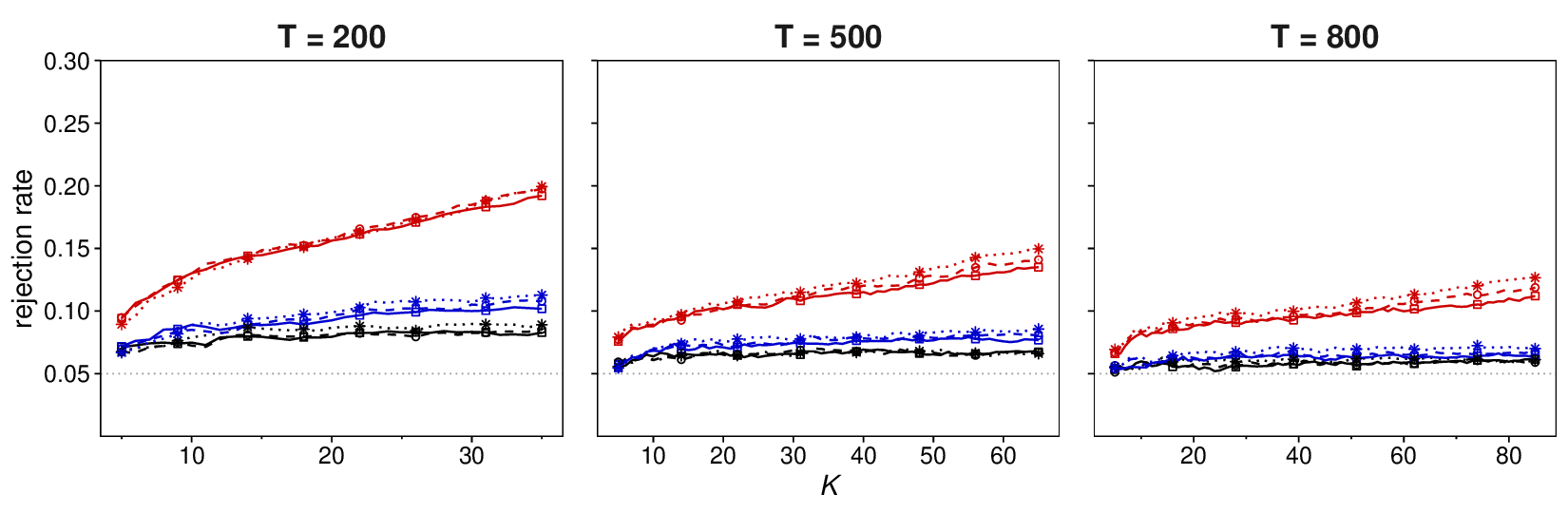}
        \caption{Rejection rates of the Sup-t test; OS HAR estimator}
        \label{fig:5_supt_OS}
    \end{subfigure}
    \caption{Rejection rates of the Sup-t test across alternative specifications.}
    \label{fig:supt_combined}
\end{figure}

Figure \ref{fig:supt_combined} reports the results for the kernel-HAR and OS-HAR cases in subfigures (a) and (b), respectively. The plots display the null empirical rejection rates of the Sup-t test at the nominal 5\% level. Equivalently, these results can be interpreted in terms of the coverage rates of the associated 95\% uniform confidence bands discussed in Section~\ref{sec:band}. For example, a rejection rate of $0.06$ corresponds to a coverage probability of $0.94$.

Overall, the Sup-t test exhibits excellent size control across bandwidth choices when using the kernel-HAR estimator. Rejection rates are slightly above the nominal level when $b$ is very small, but quickly approach the nominal level as $b$ increases to moderate values, even under strong serial dependence with $\rho=0.8$. A similar pattern is observed for the OS-HAR estimator under fixed-K inference: rejection rates remain close to the nominal level except when $K$ is large. The Sup-t test is largely insensitive to the number of grid points $n$ used to discretize $\mathcal{T}$. Some small variations arise for the OS-HAR approach when $K$ is large, but these are not relevant for the cases this paper focuses on, where $b$ is large and $K$ is small.

We conduct additional simulations to investigate further aspects of the proposed tests. We summarize the main findings here and refer readers to Section C of the Supplemental Appendix. First, the proposed tests of monotonicity and homogeneity exhibit excellent size control. Second, a sensitivity analysis of the Wald and Sup-t tests across values of $m$ and $n$ suggests that the Sup-t test may be a useful alternative to the Wald test when the null hypothesis involves many restrictions. Third, both tests exhibit excellent finite-sample power across a wide range of smoothing parameters $b$ and $K$.

\section{Predictive Quantile Regression for Stock Returns}
\label{sec:app}

This section applies our uniform fixed-b inference methods to the problem of predicting stock returns across quantiles. The predictability of stock returns has long been an important issue in the empirical study of financial markets. A substantial body of work has focused on forecasting the mean of returns; see \citet{Welch.Goyal.2008}, and \citet{lettau2010measuring} for comprehensive reviews. More recently, an emerging literature has examined predictability beyond the mean using quantile regression, asking whether different parts of the return distribution are forecastable. This line of research finds that several economic state variables exhibit limited predictive power for the center of the distribution, but exert statistically and economically significant effects in the tails. See \citet{Cenesizoglu.Timmermann.2008}, \citet{Han.Linton.Oka.Whang.2016}, \citet{Gungor.Luger.2021}, \citet{MaynardShimotsuKuriyama24}, and \citet{hoga2025selfnormalized}.

Our empirical analysis uses an updated version of the monthly data from \citet{Welch.Goyal.2008}, covering the period from January 1926 to December 2024 with sample size 1,188. Stock returns are measured using the S\&P 500 index, including dividends, and are converted to excess returns by subtracting the short-term T-bill rate. To keep our results directly comparable with the existing literature, the lagged value of each predictor is included separately in a univariate regression: $y_{t} = \alpha(\tau) + \beta_{1}(\tau) x_{t-1} + e_{t}(\tau)$.

Several commonly used predictors, such as the dividend price ratio and the earnings price ratio, exhibit strong persistence. To ensure that we work with stationary variables, we instead focus on stock variance, inflation rate, long-term returns, and net equity expansion as predictor variables. 
In all cases, the augmented Dickey-Fuller test provides strong evidence against the null hypothesis of a unit root. \citet{Welch.Goyal.2008} provide detailed variable definitions. This section presents the predictive effects of stock returns, while results for the other predictors are reported in Section C.2 of the Supplemental Appendix.

To illustrate the proposed inference methods across quantiles, we revisit the predictive quantile regression results in Figure~\ref{fig-app1} and construct uniform confidence bands using the fixed-b Sup-t approach. In Figure~\ref{fig-app2}, the black solid line represents the estimated stock variance quantile effect on equity returns, the red dotted lines represent the 95\% uniform confidence band proposed in this paper, and the purple dashed lines denote the corresponding pointwise confidence bands, included for comparison. Both bands are constructed using the same bandwidth $M = bT$ with $b=0.2$. As expected, the uniform band is wider than the pointwise band, but remains sufficiently tight to be informative.

We then consider the questions posed in the introduction for two quantile indices, $\tau_1$ and $\tau_2$. Specifically, we examine (i) the significance hypothesis $H_0: \beta_{1}(\tau_1) = \beta_{1}(\tau_2) = 0$, (ii) the equality hypothesis $H_0 : \beta_{1}(\tau_1) = \beta_{1}(\tau_2)$, and (iii) the monotonicity hypothesis $H_0: \beta_{1}(\tau_1) \leq \beta_{1}(\tau_2)$. To compare the center and the right tail of the return distribution, we set $\tau_1 = 0.5$ and $\tau_2 = 0.9$. The corresponding test statistics are $3.784$, $5.236$, and $5.236$, while the associated 95\% critical values are $3.068$, $2.492$, and $-1.845$, respectively. These results indicate rejection of the first two null hypotheses and failure to reject the third hypothesis at the 5\% level. In particular, the stock variance quantile effect is statistically significant at least at one quantile, differs between the center and the tail, and is larger at the tail than at the center. To assess heterogeneity within the right tail, we next compare two nearby quantiles, $\tau_1 = 0.91$ and $\tau_2 = 0.93$. The corresponding test statistics are $3.072$, $0.362$, and $0.362$, while the associated 95\% critical values are $2.797$, $2.769$, and $-2.105$, respectively. The results indicate that the stock variance quantile effect is statistically significant at least at one quantile, but not significantly different across the two quantiles. 

\begin{figure}[h]
\begin{center}
\caption{Quantile Effects of Stock Variance on S\&P 500 Returns}
\label{fig-app2}\vspace{-0.5cm}
\includegraphics[width=0.7\textwidth, height=0.4\textheight]{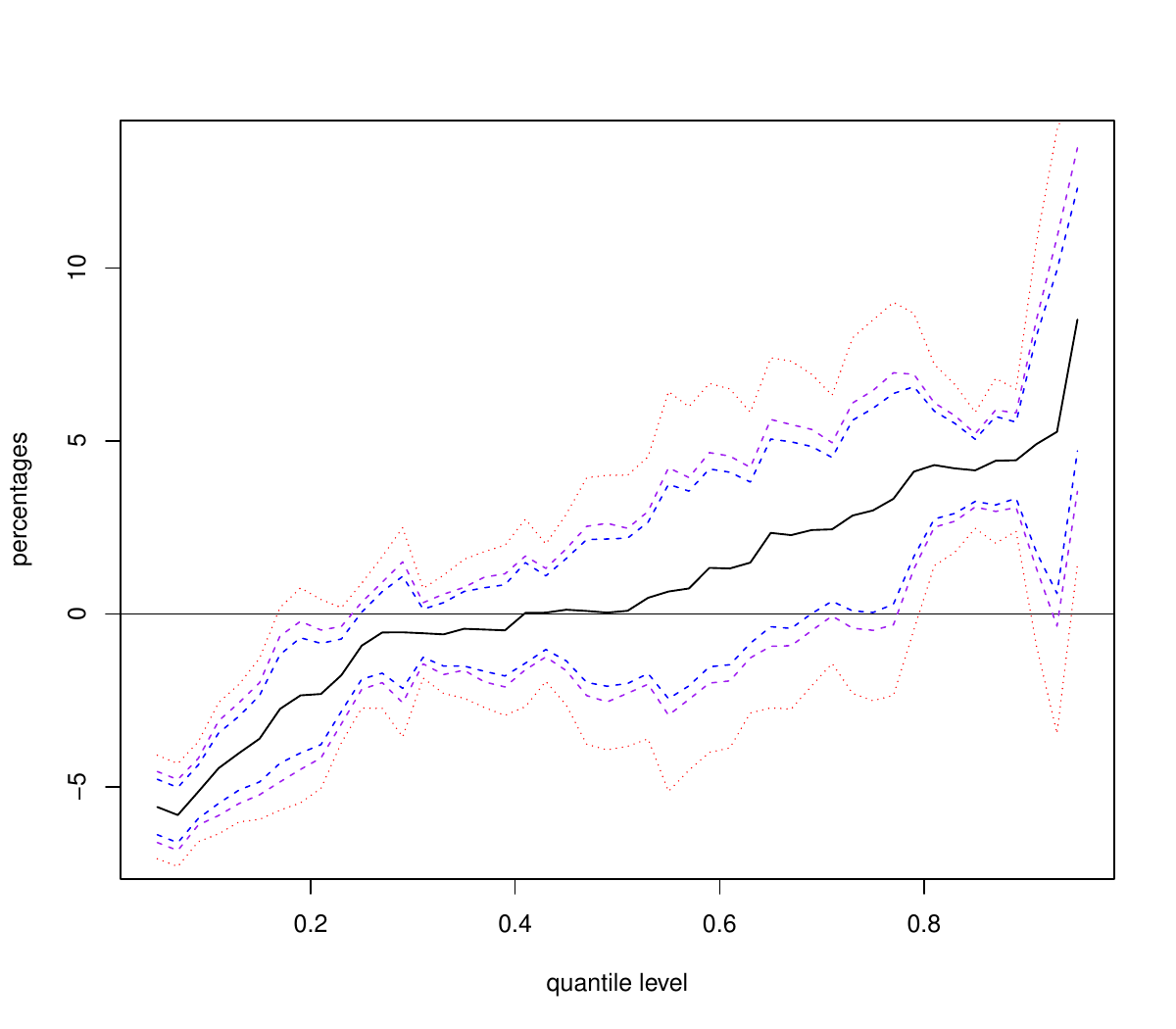}
\end{center}
\par
\begin{center}
\begin{minipage}[l]{16.0cm}
{\footnotesize Estimated quantile regression coefficients of monthly S\&P 500 excess returns on lagged stock variance. The solid black line shows the coefficient estimates. The dashed blue lines and dashed purple lines represent 95\% pointwise confidence bands based on the HAC method of \citet{galvao2024hac} and the fixed-b HAR method, respectively, while the dotted red lines denote 95\% uniform confidence bands based on the fixed-b HAR method. The Bartlett kernel is used with bandwidth $M = 0.2 T = 237$.}
\end{minipage}
\end{center}
\end{figure}

We then turn to hypotheses defined over the entire quantile process. Let $\mathcal{T} = [0.05, 0.95]$. First, consider the (global) significance hypothesis $H_0 : \beta_{1}(\tau) = 0$ for all $\tau \in \mathcal{T}$. The test statistic $t^{S}(\mathcal{T}) = 9.059$ and the 95\% critical value is $3.719$, indicating that the effect is indeed nonzero for some quantiles. Next, we test the (global) homogeneity hypothesis $H_0 : \beta_{1}(\tau_1) = \beta_{1}(\tau_2)$ for all $\tau_1 \neq \tau_2 \in \mathcal{T}$. The test statistic $t^{H}(\mathcal{T}) = 6.476$, which exceeds its 95\% critical value of $4.902$, indicating that the stock variance quantile effect is not constant and varies across quantiles. Finally, we examine the monotonicity hypothesis $H_0 : \beta_{1}(\tau_1) \leq \beta_{1}(\tau_2)$ for any $\tau_1 < \tau_2 \in \mathcal{T}$. The test statistic is $t^{M}(\mathcal{T}) = -1.038$, compared to a 95\% critical value of $-4.245$. Since the null is not rejected, the results do not provide sufficient evidence against the hypothesis that the stock variance quantile effect is monotonically increasing over $\mathcal{T}$. These findings are consistent with the uniform confidence band.

The same inferential methods can be applied to multi-step quantile prediction: $y_{t} = \alpha(\tau) + \beta_{h}(\tau) x_{t-h} + e_{t}(\tau)$. Figure~\ref{fig-app4} repeats the return-prediction exercise using stock variance while varying the forecast horizon, $h$. Panels (a)-(d) present multi-step-ahead quantile regression results for $h=1$, 3, 12, and 24, corresponding to forecasts one month, three months, one year, and two years ahead. The results indicate that the effect of market-wide stock variance is strongest at short horizons, especially in the tails, but becomes smaller and less precisely estimated as the horizon increases. These findings illustrate how quantile predictive regressions, combined with uniform inference procedures, help uncover heterogeneity in predictive effects across both quantiles and forecast horizons.

\begin{figure}[!htb]
\begin{center}
\caption{Effects of Stock Variance on S\&P 500 Returns over Multiple Forecast Horizons}
\label{fig-app4}
\includegraphics[width=0.8\textwidth, height=0.47\textheight]{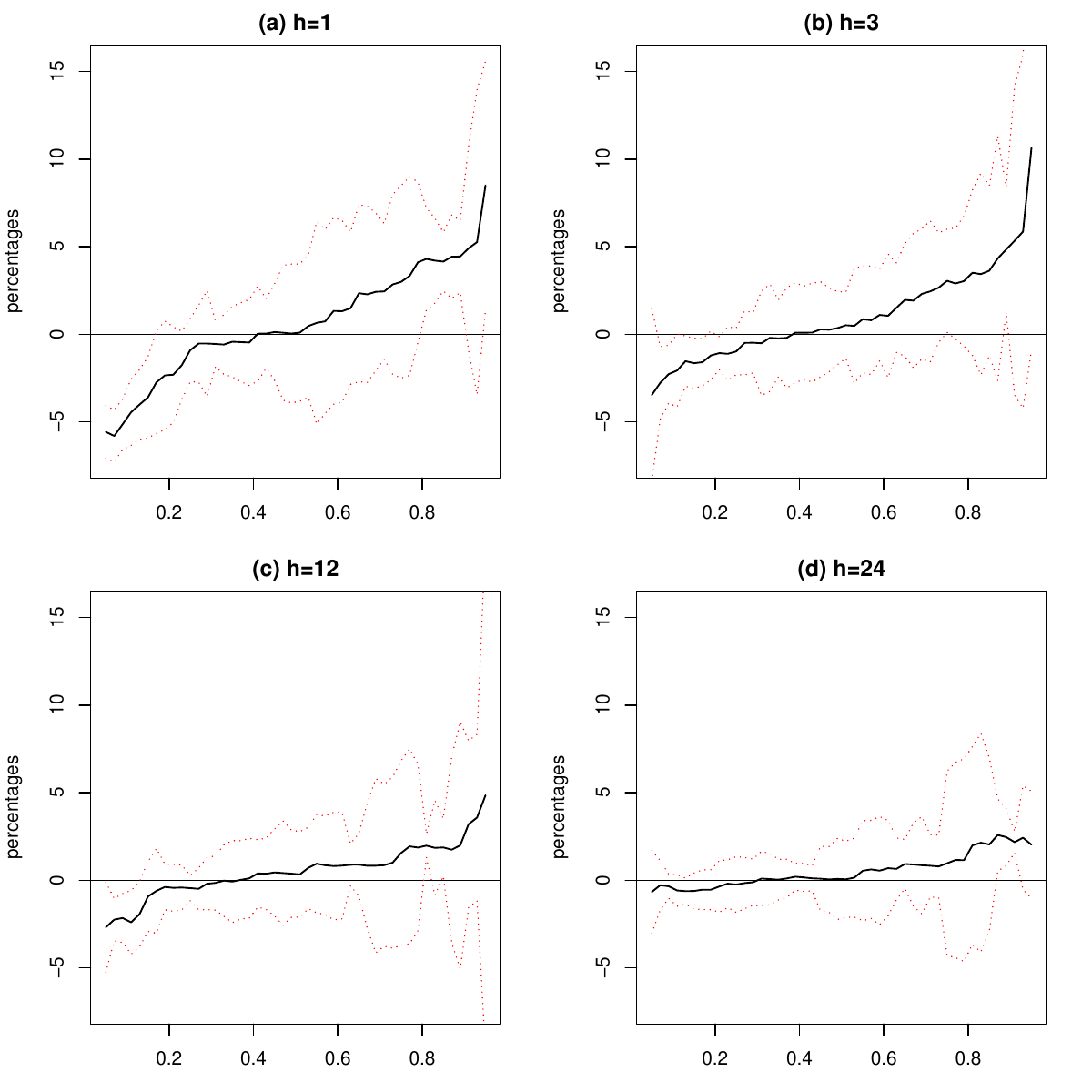}
\end{center}
\par
\begin{center}
\small{
\begin{minipage}[l]{15.0cm}
Each panel presents multi-step-ahead quantile regression results for $h = 1$, $3$, $12$, and $24$, corresponding to forecasts one month, three months, one year, and two years ahead.
\end{minipage}
}
\end{center}        
\end{figure}

\section{Conclusion}\label{sec:conclusion}

This paper develops fixed-smoothing inference methods for time-series quantile regression that are robust to heteroskedasticity and autocorrelation. Our approach enables simultaneous and uniform inference for quantile effects by explicitly accounting for dependence across quantiles. The proposed framework provides a unified set of tools for studying heterogeneity in quantile effects. Theoretical results establish the validity of the proposed procedures under general weak dependence.

Finite sample simulation evidence indicates that the fixed-smoothing methods substantially improve size control relative to existing HAC-based approaches while maintaining good power. The proposed test-based, data-driven bandwidth rules are found to work well in practice. Our empirical application further illustrates the value of uniform inference for uncovering heterogeneous predictive effects across quantiles and forecast horizons. 

A natural extension of this work would be inference for quantile impulse response functions based on quantile vector autoregressions or local projection methods \citep{White.Kim.Manganelli.2015, Han.Jung.Lee.2024}. Uniform inference procedures that jointly account for quantile levels and impulse-response horizons would provide useful tools for assessing the heterogeneous effects of macroeconomic shocks across the outcome distribution and over time. We leave this extension for future research.

\setcounter{section}{0}
\renewcommand{\thesection}{A}
\section{Appendix: Proofs of Theorems and Propositions}
\label{sec_technical_details}

This appendix provides proofs of the theorems and propositions. To simplify the notation, we will drop the dependence on the quantile level $\tau$ in quantities such as $z_{t}(\tau)$ whenever it does not cause confusion. 

\begin{proof}[Proof of Theorem \ref{thm_1}]
Let $\Delta$ be an arbitrary compact subset of $R^{p}$. For $\xi \in \Delta$, define the sequential empirical processes: 

\vskip -0.3in

\begin{align*}
S(r,\tau,\xi) &= T^{-1/2} \sum_{t=1}^{[rT]} x_{t}\left(\tau - I\left(y_{t} \leq x_{t}^{\prime}\beta_{0}(\tau) + x_{t}^{\prime}\xi/\sqrt{T}\right)\right), \\
S(r,\tau,0) &= T^{-1/2} \sum_{t=1}^{[rT]} x_{t}\left(\tau - I(y_{t} \leq x_{t}^{\prime}\beta_{0}(\tau))\right).
\end{align*}

\vskip -0.1in

Using this notation, $S([rT],\tau)$ and $\widehat{S}([rT],\tau)$ in the main text can be written as $S(r,\tau,0)$ and $S(r,\tau, \sqrt{T}(\widehat{\beta}(\tau)-\beta_{0}(\tau)))$, respectively. By adding and subtracting terms, $S(r,\tau,\xi) - S(r,\tau,0)$ is equal to $- R\left(r, \tau, \xi \right) - U\left(r, \tau, \xi \right)$ where

\begin{align}
    R\left(r, \tau, \xi \right) & = T^{-1/2} \sum_{t=1}^{[rT]} x_{t} \left\{I\left(y_{t} \leq x_{t}^{\prime}\beta_{0}(\tau) + x_{t}^{\prime}\xi/\sqrt{T}\right) - I(y_{t} \leq x_{t}^{\prime}\beta_{0}(\tau)) \right. \nonumber \\
 & \textnormal{ } \qquad \textnormal{ } \qquad \textnormal{ } \qquad - \left. F\left(x_{t}^{\prime}\beta_{0}(\tau) + x_{t}^{\prime}\xi/\sqrt{T} \Big| x_{t}\right)+F(x_{t}^{\prime}\beta_{0}(\tau)|x_{t}) \right\} \nonumber \\
U\left(r, \tau, \xi \right)&= T^{-1/2} \sum_{t=1}^{[rT]} x_{t} \left(F\left(x_{t}^{\prime}\beta_{0}(\tau) + x_{t}^{\prime}\xi/\sqrt{T} \Big| x_{t} \right) - F(x_{t}^{\prime}\beta_{0}(\tau) | x_{t}) \right) \label{eq-R-U}
\end{align}

For these two terms, we aim to establish the following results: uniformly in $r \in [0,1]$, $\tau \in \mathcal{T}$, and $\xi \in \Delta$, 
\begin{equation*}
\|R\left(r, \tau, \xi \right)\| = O_{p}(T^{-1/4}\log(T)) \quad \textnormal{and} \quad U\left(r, \tau, \xi \right)  = r D(\tau) \xi + O_{p}(T^{-1/2}).
\end{equation*}

First, note that, without loss of generality, $x_{t}$ can be assumed to be nonnegative. This is because, as in \citet{bai1996}, $x_{tj}$ can always be written as $x_{tj}^{+} - x_{tj}^{-} = \max(x_{tj},0) - \max(-x_{tj},0)$ where both $x_{tj}^{+}$ and $x_{tj}^{-}$ are non-negative and satisfy Assumption 3. The process $S(r,\tau,\xi)$ can be written as a linear combination of at most $2p$ processes with each process having nonnegative weights. 
It is thus enough to regard $x_t$ as a nonnegative scalar variable. An application of the triangle inequality would complete the argument for the original process.  
Second, divide $\mathcal{T} = [\epsilon , 1-\epsilon]$ into an equally spaced partition: $\tau_0 = \epsilon < \tau_1 < \cdots < \tau_{N(\tau)} = 1-\epsilon $. Let $\delta = \tau_{j} - \tau_{j-1}$, then the number of intervals in the partition is $N(\tau) = [(1-2\epsilon)/\delta] + 1$. Let $\delta = T^{-1/2-d}$ for some $d \in (0,1/2)$. 
Third, note that both $I\left(y_{t} \leq x_{t}^{\prime}\beta_{0}(\tau) + x_{t}^{\prime}\xi/\sqrt{T}\right)$ and $F\left(x_{t}^{\prime}\beta_{0}(\tau) + x_{t}^{\prime}\xi/\sqrt{T} \Big|x_{t}\right)$ are non-decreasing functions of $\tau$. For the ease of notation, $F(\cdot|x_{t})$ is shorten to $F_{t}(\cdot)$ and $f(\cdot|x_{t})$ to $f_{t}(\cdot)$. 

Then for any $\tau \in  [\tau_{j-1} , \tau_{j}]$, $R\left(r, \tau, \xi \right)$ is less than or equal to 

\vskip -0.1in

\begin{align*}
& T^{-1/2} \sum_{t=1}^{[rT]} x_{t} \left\{I\left(y_{t} \leq x_{t}^{\prime}\beta_{0}(\tau_{j}) + x_{t}^{\prime}\xi/\sqrt{T}\right) - I(y_{t} \leq x_{t}^{\prime}\beta_{0}(\tau_{j-1})) \right. \\
& \textnormal{ } \qquad - \left. F_{t}\left(x_{t}^{\prime}\beta_{0}(\tau_{j-1}) + x_{t}^{\prime}\xi/\sqrt{T}\right)+F_{t}(x_{t}^{\prime}\beta_{0}(\tau_{j})) \right\} \\
& = R\left(r, \tau_{j}, \xi \right) + T^{-1/2} \sum_{t=1}^{[rT]} \{I(y_{t} \leq x_{t}^{\prime}\beta_{0}(\tau_{j})) - I(y_{t} \leq x_{t}^{\prime}\beta_{0}(\tau_{j-1})) - (\tau_{j}-\tau_{j-1})\} \\
& \textnormal{ } + T^{-1/2} \sum_{t=1}^{[rT]} \left\{ F_{t}\left(x_{t}^{\prime}\beta_{0}(\tau_{j}) + x_{t}^{\prime}\xi/\sqrt{T}\right) - F_{t}\left(x_{t}^{\prime}\beta_{0}(\tau_{j-1}) + x_{t}^{\prime}\xi/\sqrt{T}\right) - F_{t}\left(x_{t}^{\prime}\beta_{0}(\tau_{j}) \right) + F_{t}\left(x_{t}^{\prime}\beta_{0}(\tau_{j-1}) \right) \right\} \\
&= R\left(r, \tau_{j}, \xi \right) + B_{1}(r,\tau_{j}) + B_{2}(r,\tau_{j}).
\end{align*}

Three quantities, $R\left(r, \tau_{j}, \xi \right)$, $B_{1}(r,\tau_{j})$, and $B_{2}(r,\tau_{j})$, depend only on the end points of the interval, $\tau_{j}$ and $\tau_{j-1}$, not on $\tau$. A reverse inequality holds: $R\left(r, \tau, \xi \right)$ is bigger than or equal to 

\vskip -0.1in

\begin{small}
\begin{align*}
& T^{-1/2} \sum_{t=1}^{[rT]} x_{t} \left\{I\left(y_{t} \leq x_{t}^{\prime}\beta_{0}(\tau_{j-1}) + x_{t}^{\prime}\xi/\sqrt{T}\right) - I(y_{t} \leq x_{t}^{\prime}\beta_{0}(\tau_{j})) \right. \\
& \textnormal{ } \qquad - \left. F_{t}\left(x_{t}^{\prime}\beta_{0}(\tau_{j}) + x_{t}^{\prime}\xi/\sqrt{T}\right)+F_{t}(x_{t}^{\prime}\beta_{0}(\tau_{j-1})) \right\} \\
& = R\left(r, \tau_{j-1}, \xi \right) - T^{-1/2} \sum_{t=1}^{[rT]} \{I(y_{t} \leq x_{t}^{\prime}\beta_{0}(\tau_{j})) - I(y_{t} \leq x_{t}^{\prime}\beta_{0}(\tau_{j-1})) - (\tau_{j}-\tau_{j-1})\} \\
& \textnormal{ } - T^{-1/2} \sum_{t=1}^{[rT]} \left\{ F_{t}\left(x_{t}^{\prime}\beta_{0}(\tau_{j}) + x_{t}^{\prime}\xi/\sqrt{T}\right) - F_{t}\left(x_{t}^{\prime}\beta_{0}(\tau_{j-1}) + x_{t}^{\prime}\xi/\sqrt{T}\right) - F_{t}\left(x_{t}^{\prime}\beta_{0}(\tau_{j}) \right) + F_{t}\left(x_{t}^{\prime}\beta_{0}(\tau_{j-1}) \right) \right\} \\
& = R\left(r, \tau_{j-1}, \xi \right) - B_{1} (r,\tau_{j})- B_{2}(r,\tau_{j}).
\end{align*}
\end{small}

\vskip -0.1in

Therefore, by the triangle inequality,
\begin{align}
& \sup_{r \in [0,1]} \sup_{\tau \in \mathcal{T}} \sup_{\xi \in \Delta} \|R\left(r, \tau, \xi \right)\| \leq \label{eq-the-goal} \\
& \textnormal{ } \quad \sup_{r \in [0,1]} \max_{1 \leq j \leq N(\tau)} \sup_{\xi \in \Delta} 2 \|R\left(r, \tau_{j}, \xi \right)\| + \sup_{r \in [0,1]} \max_{1 \leq j \leq N(\tau)} \|B_{1}(r,\tau_{j})\| + \sup_{r \in [0,1]} \max_{1 \leq j \leq N(\tau)} \sup_{\xi \in \Delta} \|B_2(r,\tau_{j})\|. \nonumber
\end{align}

Consider the first term on the right hand side of (\ref{eq-the-goal}). For the supremum over the partial sum index $r \in [0,1]$, the parameter $r$ may be restricted to the points $k/T$ with $k = 1,2,\ldots,T$ because the partial sum process changes values only at those points. Let $N(r) = T$, then it suffices to study:
\begin{equation}
\label{eq-R1}
\quad \max_{1 \leq k \leq N(r)} \max_{1 \leq j \leq N(\tau)} \sup_{\xi \in \Delta} \|R\left(k/T, \tau_{j}, \xi \right)\|.
\end{equation}

The following lemma establishes the stochastic order of (\ref{eq-R1}). 

\begin{lemma}
\label{lemma-R}
Assume Assumption 3, then for $T$ large enough,
\begin{equation*}
\quad \max_{1 \leq k \leq N(r)} \max_{1 \leq j \leq N(\tau)} \sup_{\xi \in \Delta} \|R\left(k/T, \tau_{j}, \xi \right)\| = O_{p}(T^{-1/4}\log(T)).
\end{equation*}
\end{lemma}

\begin{proof}
Because the proof is lengthy, we outline its main steps here and defer the details to Section B of the Supplemental Appendix. The proof proceeds in two steps. First, we show that $\|R\left(k/T, \tau_{j}, \xi \right)\| = O_{p}(T^{-1/4})$ for each point ($k/T , \tau_j$) and each $\xi \in \Delta$. Second, we extend this result to hold uniformly over $[0,1] \times \mathcal{T} \times \Delta$ using a chaining argument and a Bernstein inequality for mixing processes. This establishes the claim. 
\end{proof}

The next lemma studies the remaining terms on the right hand side of (\ref{eq-the-goal}). Recall that $\delta = \tau_{j} - \tau_{j-1} = T^{-1/2-d}$ for some $d \in (0,1/2)$. 

\begin{lemma}
\label{lemma-B-1}
Assume Assumption 3. Then, for $T$ large enough, 

\vskip -0.3in

\begin{align*}
\text{(i)} & \quad \sup_{r \in [0,1]} \max_{1 \leq j \leq N(\tau)} \|B_{1}(r,\tau_{j})\| = O_{p}\left(T^{-1/4-d/2} \log(T) \right), \\
\text{(ii)} & \quad  \sup_{r \in [0,1]} \max_{1 \leq j \leq N(\tau)} \sup_{\xi \in \Delta} \|B_2(r,\tau_{j})\| = O_{p}\left(T^{-1/2-d} \right).
\end{align*}
\end{lemma}

\begin{proof}
See Section B of the Supplemental Appendix. 
\end{proof}

\noindent Combining Lemmas~\ref{lemma-R} and \ref{lemma-B-1}, one can obtain $\sup_{r \in [0,1]} \sup_{\tau \in \mathcal{T}} \sup_{\xi \in \Delta} \|R\left(r, \tau, \xi \right)\|  = O_{p}\left(T^{-1/4} \log(T) \right)$.  
For the second term in (\ref{eq-R-U}), we have the following lemma. 

\begin{lemma}
\label{lemma-U}
Assume Assumption 3. Then for $T$ large, $U\left(r, \tau, \xi \right)  = r D(\tau) \xi + O_{p}(T^{-1/2})$, uniformly in $r \in [0,1]$, $\tau \in \mathcal{T}$, and $\xi \in \Delta$.
\end{lemma}

\begin{proof}
See Section B of the Supplemental Appendix.
\end{proof}

We are now ready to complete the proof of Theorem \ref{thm_1}. Let $w(\tau) = \sqrt{T} (\widehat{\beta}(\tau) - \beta_{0}(\tau))$. Under Assumption 2, $w(\tau) = O_{p}(1)$ uniformly in $\tau \in \mathcal{T}$. From the decomposition in (\ref{eq-R-U}),

\vskip -0.2in

\begin{align*}
S(r,\tau,w(\tau)) = S(r,\tau,0) - R(r,\tau,w(\tau)) - U(r,\tau,w(\tau)).
\end{align*}%
From Lemmas~\ref{lemma-R} to \ref{lemma-U}, there exist a compact set $\Delta \in R^{p}$ such that $w(\tau) \in \Delta$ and $R(r,\tau,w(\tau)) = O_{p}(T^{-1/4}\log(T))$ and $U(r,\tau,w(\tau)) = r D(\tau) w(\tau) + O_{p}(T^{-1/2})$, uniformly in $(r,\tau) \in [0,1] \times \mathcal{T}$. With the Bahadur representation in Assumption 2, $w(\tau) = D(\tau)^{-1} \, T^{-1/2} \sum_{t=1}^{T} x_{t}(\tau - I(y_{t} \leq x_{t}^{\prime}\beta_{0}(\tau))) + r_{T} =  D(\tau)^{-1} S(1,\tau,0) + r_{T}$. This in turn implies that the first order term of $U(r,\tau,w(\tau))$ is $r D(\tau) D(\tau)^{-1} S(1,\tau,0)  + r_{T}= r S(1,\tau,0) + r_{T}$. Therefore, $
S(r,\tau,w(\tau)) = S(r,\tau,0) - r S(1,\tau,0) + u_{T}$, where $u_{T} = O_{p}(T^{-1/4}\log(T)) + r_{T} = o_{p}(1)$ uniformly in $(r,\tau) \in [0,1] \times \mathcal{T}$. This proves part (i) of the theorem. If $r_{T} = o_{p}(T^{-1/4}\log(T))$, then $u_{T} = O_{p}(T^{-1/4}\log(T))$, as claimed in part (ii). The part (iii) then follows from Assumption 1. 
\end{proof}

\begin{proof}[Proof of Theorem \ref{thm_2}]: Let $k_{ij} = K\left(\frac{i-j}{b_T}\right)$, and recall that $\widehat{S}(i,\tau)$ denotes $S(i/T,\tau,w(\tau))$. By summation by parts, $\widehat{\Lambda}(\tau) = \frac{1}{T}\sum_{i=1}^{T}\sum_{j=1}^{T} k_{ij}\widehat{z}_{i}(\tau) \widehat{z}_{j}(\tau)^{\prime}$ is equal to 
\begin{align}
    & \sum_{i=1}^{T-1}\sum_{j=1}^{T-1} \widehat{S}(i,\tau)(k_{ij} - k_{i,j+1}-k_{i+1,j}+k_{i+1,j+1}) \widehat{S}(j,\tau)^{\prime} + \sum_{i=1}^{T-1} \widehat S(i,\tau)(k_{iT}-k_{i+1,T}) \widehat S(T,\tau)^{\prime}  \nonumber \\
   & + \sum_{j=1}^{T-1} \widehat S(T,\tau)(k_{Tj}-k_{T,j+1}) \widehat S(j,\tau)^{\prime} +  \widehat S(T,\tau)k_{TT} \widehat S(T,\tau)^{\prime} \label{eq:sum_byparts}
\end{align}

We first show the last three terms in \eqref{eq:sum_byparts} are asymptotically negligible uniformly in $\tau\in\mathcal{T}$. By the computational property of QR, $\sup_{\tau\in\mathcal{T}}\|\widehat S(T,\tau)\|=o_p(1)$, and by Theorem~\ref{thm_1}, $\sup_{1\leq i\leq T}\sup_{\tau\in\mathcal{T}}
\|\widehat S(i,\tau)\|=O_p(1)$. Moreover, under the kernel conditions, $\sum_{i=1}^{T-1}|k_{iT}-k_{i+1,T}|=O(1)$ and $\sum_{j=1}^{T-1}|k_{Tj}-k_{T,j+1}|=O(1)$, from the mean value theorem and boundedness of the derivative for the kernel. It holds automatically for the Bartlett kernel. Thus,
$$
    \left\| \sum_{i=1}^{T-1} \widehat S(i,\tau)(k_{iT}-k_{i+1,T})\widehat S(T,\tau)^{\prime} \right\|_{\infty} \leq \|\widehat S(i,\tau)\|_{\infty} \|\widehat S(T,\tau)\|_{\infty}\sum_{i=1}^{T-1} |k_{iT}-k_{i+1,T}| = o_p(1)
$$
The same argument, together with $k_{TT} = O(1)$, shows that the last two terms are also $o_p(1)$ uniformly in $\tau\in \mathcal T$. We then analyze the first term in \eqref{eq:sum_byparts} for each of the three cases, focusing on the first case here. The analyses of the remaining two cases are analogous and are provided in Section B of the Supplemental Appendix.

(i) When $\mathcal{K}(\cdot)$ is twice continuously differentiable, $\mathcal{K}''$ is continuous and hence uniformly continuous on
the relevant compact interval. Thus,  
   $ T^{2} \left(k_{ij} - k_{i,j+1}-k_{i+1,j}+k_{i+1,j+1}\right) = - \frac{1}{b^{2}} \mathcal{K}''\left(\frac{r_i-s_j}{b}\right) + \rho_{ij,T}$ where $s_j = j/T$, $r_i = i/T$ and $\max_{i,j} |\rho_{ij,T}| \to 0$. Therefore, we can write the first term as $V_T(\tau)+R_T(\tau)$
   where  
    \begin{equation*}
       V_{T}(\tau) = \frac{-1}{b^2 T^2} \sum_{i=1}^{T-1}\sum_{j=1}^{T-1} \widehat S(i,\tau) \mathcal{K}''\left(\frac{r_{i}-s_{j}}{b}\right) \widehat S(j,\tau)^{\prime}, \qquad R_T(\tau)=\frac{1}{T^2}\sum_{i=1}^{T-1}\sum_{j=1}^{T-1}
\widehat S(i,\tau)\rho_{ij,T}\widehat S(j,\tau)^{\prime}.
     \end{equation*}%
Define the limit functional as
    \begin{equation*}
       V(\tau) =  \frac{-1}{b^{2}} \int_{0}^{1} \int_{0}^{1} \tilde{G}(r,\tau) \mathcal{K}''\left(\frac{r-s}{b}\right) \tilde{G}(s,\tau)^{\prime} dr ds.
    \end{equation*}

Since $\sup_{\tau\in\mathcal T}\|R_T(\tau)\|
\leq \max_{i,j}|\rho_{ij,T}|\,
\|\widehat S(i,\tau)\|_\infty^2=o_p(1)$, 
it suffices to show that $V_T(\tau)\Rightarrow V(\tau)$ in $\ell^\infty(\mathcal T)$. For $f\in\ell^\infty([0,1] \times \mathcal T; \mathbb R^p)$, define the discrete and limiting quadratic functionals:

\vskip -0.1in

    \begin{align*}
       & \Phi_{T}(f)(\tau) = \frac{-1}{b^2 T^2} \sum_{i=1}^{T-1}\sum_{j=1}^{T-1} \mathcal{K}''\left(\frac{r_{i}-s_{j}}{b}\right) f(r_{i},\tau)  f(s_{j},\tau)^{\prime}, \\
       & \Phi(f)(\tau) = \frac{-1}{b^{2}} \int_{0}^{1} \int_{0}^{1} \mathcal{K}''\left(\frac{r-s}{b}\right) f(r,\tau) f(s,\tau)^{\prime} dr ds.
     \end{align*}%
    Then, $V_{T}(\tau) =\Phi_{T}(\widehat S)(\tau)$ and $V(\tau) =\Phi(\tilde{G})(\tau)$. Since $\Phi_T(\widehat S)=\Phi_T(\widehat S)-\Phi(\widehat S)+\Phi(\widehat S)$, we will first show that $\|\Phi_T(\widehat S)-\Phi(\widehat S)\|_\infty=o_p(1)$ and $\Phi(\widehat S)\Rightarrow \Phi(\widetilde G)$ by establishing the continuity of $\Phi$. 
    
    For the first term, we note that since $\widehat S([rT],\tau)=\widehat S(i/T,\tau)$ for $r\in[i/T,(i+1)/T)$,
\begin{align*}
\left\|\Phi_T(\widehat S)(\tau)-\Phi(\widehat S)(\tau)\right\|
&\leq
\frac{1}{b^2T^2}
\sum_{i=1}^{T-1}\sum_{j=1}^{T-1}
\left\|\widehat S(i,\tau)\right\|
\left\|\widehat S(j,\tau)\right\|\times
\sup_{\substack{r\in[\frac{i}{T},\frac{i+1}{T})\\
s\in[\frac{j}{T},\frac{j+1}{T})}}
\left|
\mathcal{K}''\left(\frac{r_i-s_j}{b}\right)
-
\mathcal{K}''\left(\frac{r-s}{b}\right)
\right|.
\end{align*}
Since $\mathcal{K}''$ is uniformly continuous on the relevant compact
interval and $\left|\frac{r_i-s_j}{b}-\frac{r-s}{b}
\right|\leq \frac{2}{bT}$, the last supremum is $o(1)$. Hence, $\left\|
\Phi_T(\widehat S)-\Phi(\widehat S)\right\|_{\infty} = o_p(1)$. 

Next, for bounded $f$ and $g$,
\begin{align*}
\|\Phi(f)-\Phi(g)\|_\infty =& \left\|\frac{1}{b^2} \int_{0}^1\int_{0}^1 K^{\prime \prime}\left(\frac{r-s}{b}\right)\left(f(r,\tau)  f(s,\tau)^{\prime} -  g(r,\tau)  g(s,\tau)^{\prime}\right)\right\|_{\infty}\\
\leq&\frac{C_K}{b^2}\left(\|f\|_\infty+\|g\|_\infty\right)\|f-g\|_\infty,
\end{align*}
where $C_K=\sup_x|\mathcal{K}''(x)|$. Hence, $\Phi$ is continuous
from $\ell^\infty([0,1]\times\mathcal T)$ to $\ell^\infty(\mathcal T)$.
Since $\widehat S\Rightarrow \tilde{G}$ in $\ell^{\infty}([0,1] \times \mathcal{T})$ by Theorem \ref{thm_1} (iii), the continuous mapping theorem
implies $\Phi(\widehat S)\Rightarrow\Phi(\widetilde G)$ in $\ell^\infty(\mathcal T)$. Then, Slutsky theorem establishes the claim.

\end{proof}

\begin{proof}[Proof of Theorem~\ref{thm:fixedK}]: Consider $\Phi_T^k(\tau)$. By summation by parts, it is equal to 
    \begin{align*}
\phi_k(1)\widehat S(T,\tau)-
\sum_{i=1}^{T-1}
\widehat S\left(i,\tau\right)
\left[
\phi_k\left(\frac{i+1}{T}\right)
-
\phi_k\left(\frac{i}{T}\right)
\right] 
=
\phi_k(1)\widehat S(T,\tau)
-
\int_0^1 \widehat S([rT],\tau)\,d\phi_k(r).
\end{align*}
The equality follows by expressing the sum as a Riemann–Stieltjes integral with respect to $\phi_k$, which is continuous and of bounded variation. By Theorem~\ref{thm_1},
$\widehat S\Rightarrow\widetilde G$ in
$\ell^\infty([0,1]\times\mathcal T)$ and
$\sup_{\tau\in\mathcal T}\|\widehat S(T,\tau)\|=o_p(1)$. Since $\phi_k(1)$ is finite,
$\sup_{\tau\in\mathcal T}
\|\phi_k(1)\widehat S(T,\tau)\|=o_p(1)$.

For each $k$, define the linear functional
$\mathcal G_k(f)(\tau)=-\int_0^1 f(r,\tau)\,d\phi_k(r)$. Note that the continuity and finite-difference condition imply that ${\rm TV}(\phi_k)<\infty$, we have
$\|\mathcal G_k(f)-\mathcal G_k(g)\|_\infty
\leq {\rm TV}(\phi_k)\|f-g\|_\infty$.
Thus, $\mathcal G_k$ is continuous from
$\ell^\infty([0,1]\times\mathcal T)$ to
$\ell^\infty(\mathcal T)$. Also note that the map
$f\mapsto(\mathcal G_1(f),\ldots,\mathcal G_K(f))$ is continuous. By continuous mapping theorem,
\[
\left(\Phi_T^1,\ldots,\Phi_T^K\right)
\Rightarrow
\left(H_1,\ldots,H_K\right)
\quad\text{in }
\{\ell^\infty(\mathcal T)\}^K,
\quad {\rm where \ }
H_k(\tau)
=
-\int_0^1\widetilde G(r,\tau)\,d\phi_k(r).
\]
By integration by parts and the zero-mean condition
$\int_0^1\phi_k(r)\,dr=0$,
$H_k(\tau)=\int_0^1\phi_k(r)\,dG(r,\tau)$, 
\begin{align*}
\operatorname{Cov}\left(
H_k(\tau_1),H_l(\tau_2)
\right)
&=
\left(
\int_0^1\phi_k(r)\phi_l(r)\,dr
\right)
\Lambda(\tau_1,\tau_2) =
\mathbf 1\{k=l\}\Lambda(\tau_1,\tau_2),
\end{align*}
where the last equality follows from the orthonormality of
$\{\phi_k\}_{k\geq1}$.
Hence, each $H_k$ has the same distribution as $G(1,\cdot)$, and
$H_1,\ldots,H_K$ are mutually independent. We can therefore write
$H_k(\tau)\overset{d}{=}G_k(1,\tau)$, where
$\{G_k(1,\tau)\}_{k=1}^K$ are independent copies of $G(1,\tau)$.

Finally, note that the map
$(h_1,\ldots,h_K)\mapsto
K^{-1}\sum_{k=1}^Kh_k(\tau)h_k(\tau)^{\prime}$
is continuous from $\{\ell^\infty(\mathcal T)\}^K$ to
$\ell^\infty(\mathcal T)$. Thus, an application of the continuous
mapping theorem gives the desired result:
\[
\widehat{\Lambda}^{\rm OS}(\tau)
=
\frac{1}{K}\sum_{k=1}^K
\Phi_T^k(\tau)\Phi_T^k(\tau)^{\prime}
\Rightarrow
\frac{1}{K}\sum_{k=1}^K
G_k(1,\tau)G_k(1,\tau)^{\prime}
\quad\text{in }\ell^\infty(\mathcal T). 
\]
\end{proof}

\begin{proof}[Proof of Proposition \ref{prop:fixedb}]. Under Assumptions~\ref{assum_gaussian_limit} and \ref{assum_asymp_linear}, 
$\sqrt{T}\bigl(\widehat\beta(\tau)-\beta_0(\tau)\bigr)
=D(\tau)^{-1}S(T,\tau)+o_p(1)$
uniformly in $\tau\in\mathcal T$. Since $r(\cdot)$ is continuously differentiable and $\beta_0(\tau)$ lies in a compact set, the uniform delta method yields, under the null,
\[
\sqrt{T}r(\widehat\beta(\tau))
=R_0D(\tau)^{-1}S(T,\tau)+o_p(1)
\quad \text{in } \ell^\infty(\mathcal T).
\]

By Theorem~\ref{thm_1}, $\widehat S([rT],\tau)
=S([rT],\tau)-rS(T,\tau)+o_p(1)$ 
uniformly in $(r,\tau)$. Thus, under Assumption~\ref{assum_gaussian_limit}, $(S(T,\cdot),\widehat S)$ converges weakly to $(G(1,\cdot),\widetilde G)$ jointly.

The proofs of Theorems~\ref{thm_2} and \ref{thm:fixedK} show that $\widehat\Lambda^\iota(\tau)$ is asymptotically a continuous functional of $\widehat S$. The continuous mapping theorem and Slutsky's theorem therefore yield
\[
\left(
\sqrt{T}r(\widehat\beta(\tau)),
\widehat\Lambda^\iota(\tau)
\right)
\Rightarrow
\left(
R_0D(\tau)^{-1}G(1,\tau),
\Lambda_\infty^\iota(\tau)
\right),
\]
where
$\Lambda_\infty^{\rm KE}(\tau)
=P(\widetilde G(r,\tau),b)$ and 
$\Lambda_\infty^{\rm OS}(\tau)
=\frac{1}{K}\sum_{k=1}^K
G_k(1,\tau)G_k(1,\tau)'.$

Condition (2) of the statement and the uniform non-singularity of $D(\tau)$ imply
$\sup_{\tau\in\mathcal T}
\|\widehat D(\tau)^{-1}-D(\tau)^{-1}\|=o_p(1)$. 
Similarly, the uniform consistency of $\widehat\beta(\tau)$ and the uniform continuity of $R(\cdot)$ on a compact neighborhood of $\{\beta_0(\tau):\tau\in\mathcal T\}$ yield
$\sup_{\tau\in\mathcal T}
\|R(\widehat\beta(\tau))-R_0\|=o_p(1)$.

For $\iota \in \{\mathrm{KE},\mathrm{OS}\}$, define the estimated covariance matrix used in the Wald and $t$ statistics as
\[
\widehat H^\iota(\tau)
=
R(\widehat\beta(\tau))
\widehat D(\tau)^{-1}
\widehat\Lambda^\iota(\tau)
\widehat D(\tau)^{-1}
R(\widehat\beta(\tau))'.
\]
The continuous mapping theorem then gives
\[
\left(
\sqrt{T}r(\widehat\beta(\tau)),
\widehat H^\iota(\tau)
\right)
\Rightarrow
\left(
R_0D(\tau)^{-1}G(1,\tau),
H^\iota(\tau)
\right).
\]

Condition (iii) ensures that $H^\iota(\tau)$ is almost surely uniformly positive definite on $\mathcal{T}$. Hence, matrix inversion is continuous at the limiting covariance process under the uniform norm. The continuous mapping theorem then yields \eqref{wald_limit} and \eqref{t_limit} for the Wald and $t$ statistics, respectively. 
\end{proof}

\begin{proof}[Proof of Proposition \ref{prop:supt_direct}]:
By Proposition~\ref{prop:fixedb}, 
$$t_j(\tau):=
\frac{\widehat\beta_j(\tau)-\beta_j(\tau)}
{\widehat\sigma_j(\tau)}
\Rightarrow
\frac{[D(\tau)^{-1}G(1,\tau)]_j}
{\sqrt{[D(\tau)^{-1}
\Lambda_\infty^\iota(\tau)
D(\tau)^{-1}]_{jj}}}
=:t_j^\infty(\tau)$$ in $\ell^\infty(\mathcal T)$, where
$\Lambda_\infty^{\rm KE}(\tau)
=P(\widetilde G(r,\tau),b)$ and
$\Lambda_\infty^{\rm OS}(\tau)
=K^{-1}\sum_{k=1}^K
G_k(1,\tau)G_k(1,\tau)'$.

The map
$f\mapsto\max_{1\leq i\leq m}|f(\tau_i)|$
is continuous under the uniform norm. Hence,
by the continuous mapping theorem, $
\max_{1\leq i\leq m}|t_j(\tau_i)|
\Rightarrow
\max_{1\leq i\leq m}|t_j^\infty(\tau_i)|$. 

The numerator of $t_j^\infty(\tau_i)$ is Gaussian and independent of its limiting variance. Thus, each
$t_j^\infty(\tau_i)$ has a continuous
distribution, so
$\Pr(\max_{1\leq i\leq m}
|t_j^\infty(\tau_i)|=c_\alpha)=0$
for any $c_\alpha>0$. The Portmanteau theorem therefore yields
\begin{align*}
\Pr\left(
\beta_j(\bar\tau_m)\in I_\alpha(\bar\tau_m)
\right)=
\Pr\left(
\max_{1\leq i\leq m}|t_j(\tau_i)|
\leq c_\alpha
\right)\to
\Pr\left(
\max_{1\leq i\leq m}|t_j^\infty(\tau_i)|
\leq c_\alpha
\right).
\end{align*}

\end{proof}

\begin{proof}[Proof of Proposition \ref{prop-shape-restriction}]: Let $\bar\tau_n$ denote the fixed grid. By \eqref{eq_thm_stack1}--\eqref{eq_thm_stack3} and the
consistency of $\widehat D(\bar\tau_n)$, we have the joint weak convergence
\begin{align*}
\left(\sqrt{T}\big(\widehat\beta(\bar\tau_n)-\beta(\bar\tau_n)\big) , \, \widehat\Sigma(\bar\tau_n) \right)
\Rightarrow 
\left( D(\bar\tau_n)^{-1}G(1,\bar\tau_n), \, H(\bar\tau_n,b) \right).
\end{align*}

Hence, for any nonzero vector $a$ with strictly positive limiting variance, the continuous mapping theorem yields
\[
\frac{\sqrt{T}a'
(\widehat\beta(\bar\tau_n)-\beta(\bar\tau_n))}
{\sqrt{a'\widehat\Sigma(\bar\tau_n)a}}
\Rightarrow
\frac{a'D(\bar\tau_n)^{-1}G(1,\bar\tau_n)}
{\sqrt{a'H(\bar\tau_n,b)a}}.
\]

Under the null hypotheses in parts (1) and (2), the
corresponding restrictions satisfy $a_{i,j}'\beta(\bar\tau_n)=0$
and $a_{i,i-1;j}'\beta(\bar\tau_n)=0$, respectively.
Under the least favorable configuration in part (3),
$a_{i,l;j}'\beta(\bar\tau_n)=0$ for all $i>l$.
The stated results then follow from the continuous mapping theorem applied to the absolute values and finite maxima or minima of the jointly convergent statistics.
\end{proof}

\singlespacing


\bibliographystyle{apalike2}
\bibliography{reference_fixedbQR.bib}

\bigskip \bigskip

\newpage

\centerline{ \Large \bf{Supplemental Appendix for }} \vspace{5mm} 

\centerline{ \Large \bf{``Fixed-smoothing Uniform Inference for Quantile Regression''}} \vspace{5mm}

\centerline{ \bf{Kaicheng Chen, Antonio Galvao, Seunghwa Rho, Tim Vogelsang, and Jungmo Yoon}} \vspace{1cm} 
\centerline{ \bf{(For online publication only)}} \vspace{1.0cm} 


\setcounter{section}{0} \setcounter{equation}{0}%
\setcounter{definition}{0}\setcounter{assumption}{0} \setcounter{lemma}{0} \setcounter{table}{0} \setcounter{figure}{0}%
\setcounter{page}{1} \renewcommand{\thepage}{B-\arabic{page}} %
\renewcommand{\theequation}{B.\arabic{equation}}
\setcounter{lemma}{0}
\renewcommand{\thelemma}{B.\arabic{lemma}}
\setcounter{proposition}{0}
\renewcommand{\theproposition}{B.\arabic{proposition}}
\setcounter{corollary}{0}
\renewcommand{\thecorollary}{B.\arabic{corollary}}
\renewcommand{\theassumption}{B.\arabic{assumption}}%
\renewcommand{\theremark}{B.\arabic{remark}}
\renewcommand{\thesection}{B.}
\renewcommand{\thesubsection}{B.\arabic{section}.\arabic{subsection}}
\setcounter{subsection}{0}
\renewcommand{%
\thedefinition}{B.\arabic{definition}}
\renewcommand{\thetable}{B.\arabic{table}}
\renewcommand{\thefigure}{B.\arabic{figure}}

\openup -0.2em

\setcounter{footnote}{0}

\onehalfspacing

\section{Auxiliary Lemmas}

The next two lemmas are used to prove Lemma~\ref{lemma-R} in the main text. To study (\ref{eq-R1}), we first establish its pointwise convergence. 

\begin{lemma}
\label{lemma-beta-pt}
Assume Assumption 3. Fix $r \in [0,1]$, $\tau \in \mathcal{T}$, and $\xi \in \Delta$. Then for $T$ large enough, $\| R(r,\tau,\xi) \| = O_{p}\left(T^{-1/4}\right)$.
\end{lemma}

\begin{proof} 
When $r=0$, the claim holds trivially. So fix $0 < r \leq 1$ and let 
\begin{equation*}
\eta_{t}(\xi)  = I(y_{t} \leq x_{t}^{\prime}\beta_{0}(\tau) + x_{t}^{\prime}\xi/\sqrt{T}) - I(y_{t} \leq x_{t}^{\prime}\beta_{0}(\tau)) - F_{t}(x_{t}^{\prime}\beta_{0}(\tau) + x_{t}^{\prime}\xi/\sqrt{T}) + F_{t}(x_{t}^{\prime}\beta_{0}(\tau)).
\end{equation*}
Note that $R(r,\tau,\xi) = T^{-1/2}\sum_{t=1}^{[rT]} \eta_{t}(\xi) x_{t}$. To evaluate its stochastic order, calculate its variance. By stationarity,

\begin{align*}
\textnormal{Var}\left(T^{-1/2}\sum_{t=1}^{[rT]} \eta_{t}(\xi) x_{t} \right) &= r \textnormal{Var}\left(\eta_{t}(\xi) x_{t} \right) + 2 \sum_{j=2}^{[rT]} \left(\frac{[rT]}{T} - \frac{j}{T}\right) \textnormal{Cov}\left(\eta_{1}(\xi) x_{1} ,\eta_{j}(\xi) x_{j} \right) \\
&< r \textnormal{Var}\left(\eta_{t}(\xi) x_{t} \right) + 2 r \sum_{j=2}^{[rT]} \left| \textnormal{Cov}\left(\eta_{1}(\xi) x_{1} ,\eta_{j}(\xi) x_{j} \right) \right|.
\end{align*}

To simplify the notation, assume without loss of generality that $x_{t}$ is a scalar random variable. Also, let $\eta_{t}$ denote $\eta_{t}(\xi)$. To calculate $\textnormal{Var}\left(\eta_{t} x_{t} \right)$, note that $\eta_{t}$ is a centered Bernoulli random variable with the success probability $p_{t} = \textnormal{P}(x_{t}\beta_{0}(\tau) \leq y_{t} \leq x_{t}\beta_{0}(\tau) + x_{t}\xi/\sqrt{T}) = \textnormal{P}(0 \leq e_{t}(\tau) \leq x_{t}\xi/\sqrt{T})$ where the quantile error $e_{t}(\tau) = y_{t} - x_{t}\beta_{0}(\tau)$. So, $\textnormal{E}[\eta_{t} x_{t}]=0$ by the law of iterated expectation and $\textnormal{Var}\left(\eta_{t} x_{t}\right) = \textnormal{E}[\left(\eta_{t} x_{t}\right)^{2}] = \textnormal{E}[ \textnormal{E}[\eta_{t}^{2}|x_{t}] x_{t}^{2}] = \textnormal{E}[p_{t}(1 - p_{t}) x_{t}^{2}]$. 
To evaluate this quantity, note that under Assumption~3(i), (iii),
$$p_{t} = \int_{0}^{x_{t}\xi/\sqrt{T}} f(e_{t}|x_{t})de_{t} \leq \int_{0}^{x_{t}\xi/\sqrt{T}} \bar{f} de_{t} \leq \bar{f} |x_{t}\xi| /\sqrt{T}.$$ Therefore, there exists a constant $c_1$ such that $\textnormal{E}[\left(\eta_{t}(\xi) x_{t}\right)^{2}] \leq c_1 T^{-1/2}$. Next, we evaluate $\textnormal{E}\left[(\eta_{1} x_{1})^{2} \, (\eta_{j} x_{j})^{2}\right]$. It is equal to $\textnormal{E}\left[\textnormal{E}\left[(\eta_{1} \, \eta_{j})^{2}|x_1,x_j\right] (x_{1} x_{j})^{2}\right]$ and the inner expectation can be decomposed as $\textnormal{Var}\left[(\eta_{1} \, \eta_{j})|x_1,x_j\right] + \textnormal{E}\left[(\eta_{1} \, \eta_{j})|x_1,x_j\right]^{2}$. The order of these quantities depend on $\textnormal{E}\left[I(0 \leq e_{1}(\tau) \leq x_{1}\xi/\sqrt{T}) I(0 \leq e_{j}(\tau) \leq x_{j}\xi/\sqrt{T})|x_{1},x_{j}\right]$ which can be written as,
\begin{align*}
\int_{0}^{x_{1}\xi/\sqrt{T}} \int_{0}^{x_{j}\xi/\sqrt{T}} f(e_{1},e_{j}|x_{1},x_{j}) de_1 de_{j} \leq \int_{0}^{x_{1}\xi/\sqrt{T}} \int_{0}^{x_{j}\xi/\sqrt{T}} \check{f} de_1 de_{j} = \check{f} \xi^{2} T^{-1} |x_{1} x_{j}|
\end{align*}%
by Assumption~3(iv). Proceeding as before, $\textnormal{E}\left[\textnormal{Var}\left[(\eta_{1} \, \eta_{j})|x_1,x_j\right](x_{1} x_{j})^{2}\right] < \check{f} \xi^{2} \textnormal{E}\left[|x_{1} x_{j}|^{3}\right] T^{-1}$ by  Assumption~3(i). A similar argument applies to the second term, so by combining them, $\textnormal{E}\left[(\eta_{1} \eta_{j})^{2} (x_{1} x_{j})^{2}\right]  < c_{2} T^{-1}$ for some constant $c_2$. Then by Lemma C.2 in Galvao and Kato (2016), 
$\left|\textnormal{Cov}\left(\eta_{1} x_{1},\eta_{j} x_{j}\right)\right| \leq C T^{-1/2} \beta(j)^{1/2}$ for some constant $C$. Conclude that 
\begin{equation*}
r \sum_{j=2}^{[rT]} \left| \textnormal{Cov}\left(\eta_{1} x_{1},\eta_{j} x_{j} \right) \right| \leq r c_{0} T^{-1/2} \sum_{j=2}^{T} \beta(j)^{1/2}  \leq r c_{0} T^{-1/2} \sum_{j=1}^{\infty} \beta(j)^{1/2} \leq r C T^{-1/2}
\end{equation*}
because $\sum_{j=1}^{\infty} \beta(j)^{1/2} < \infty$ from Assumption~3(ii). Combining the above results, this establishes that $\textnormal{Var}\left(T^{-1/2} \sum_{t=1}^{[rT]} \eta_{t}(\xi) x_{t}\right) \leq r C T^{-1/2}$. 
\end{proof}

The next lemma extends the pointwise convergence in Lemma~\ref{lemma-beta-pt} to the uniform convergence. 

\begin{lemma}
\label{lemma-beta-uf}
Assume Assumption 3, then for $T$ large enough,
$$\quad \max_{1 \leq k \leq N(r)} \max_{1 \leq j \leq N(\tau)} \sup_{\xi \in \Delta} \|R\left(k/T, \tau_{j}, \xi \right)\| = O_{p}(T^{-1/4}\log(T)).$$
\end{lemma}

\begin{proof}
To extend a pointwise result to a uniform result, we will use the chaining argument. Because we will use the Bernstein inequality, it will be convenient to define 
$$\widetilde{R}(r,\tau,\xi) = T^{-1/2} R(r,\tau,\xi) = \frac{1}{T} \sum_{t=1}^{[rT]} \eta_{t}(\xi) x_{t}$$%
and study its asymptotic order. For this purpose, take $\gamma_{T} = T^{-3/4}\log(T)$. 

Partition the parameter space $\Delta$ into a union of cubes $E_{i}$ which have vertices on the set $(k_{i1}b_{T},\ldots,k_{ip}b_{T})$ where $k_{il} = \{0,\pm 1,\ldots,\pm b_{T}^{-1}\}$. Let $b_{T} = T^{-1/4}$, then $E_i$ has the side length $b_{T} = c_0 T^{-1/4}$ for some finite constant $c_{0}$ and the number of cubes $N(\xi) = (2 b_{T}^{-1}+1)^{p} = (2T^{1/4}+1)^{p}$. Write $\xi \in E_{i}$ if $\xi$ lies into the $i$-th cube $E_{i}$ where $i=1,\ldots,N$. Let $\xi_i$ be the smallest value in $E_{i}$. By the triangle inequality
\begin{align*}
\sup_{\xi \in \Delta} | \widetilde{R}(r,\tau,\xi) | \leq \max_{1 \leq i \leq N(\xi)} |\widetilde{R}(r,\tau,\xi_{i})| + \max_{1 \leq i \leq N(\xi)} \sup_{\xi \in \Delta \cap E_{i}} | \widetilde{R}(r,\tau,\xi) - \widetilde{R}(r,\tau,\xi_{i})|.
\end{align*}%

To analyze the first term, let 
$$P(\xi) = \max_{1 \leq k \leq T} \max_{1 \leq j \leq N(\tau)}\max_{1 \leq i \leq N(\xi)} | \widetilde{R}(k/T,\tau_{j},\xi_{i}) |.$$ 
By the union bound, for any $M>0$,
\begin{equation}
\label{eq-prob-P}
\textnormal{Pr}\left(|P(\xi)| \geq M \, \gamma_{T}\right) \leq \sum_{k=1}^{N(r)} \sum_{j=1}^{N(\tau)} \sum_{i=1}^{N(\xi)} \textnormal{Pr}\left(| \widetilde{R}(k/T,\tau_{j},\xi_{i})| \geq M \, \gamma_{T}\right).
\end{equation}%

Note that $|\eta_{t}(\xi)| \leq 1$ and because of Assumption~3(i), we have $\max_{1 \leq t \leq T}\|x_{t}\| = o_{p}(T^{1/4})$ (see Lemma 2 in Ota, Kato, and Hara, 2019). So there exists $\kappa \in (0, 1/8)$ and $B>0$ such that $\max_{t} |\eta_{t}(\xi) x_{t}| \leq \max_{t}\|x_{t}\| \leq B T^{1/4-\kappa}$. 

For each $T$, by the Bernstein inequality under strong mixing conditions with geometrically-decaying mixing coefficients (Merlev\`ede, Peligrad, and Rio, 2009; Hang and Steinwart, 2017), there exists a constant $c_1 > 0$ not depending on $r$, $\tau$, and $\xi$ such that
\begin{align}
& \textnormal{Pr}\left(T| \widetilde{R}(k/T,\tau_{j},\xi_{i})| \geq T M \gamma_{T}\right) \label{eq-bernstein-1} \\
& \leq \exp\left(-\frac{c_{1} T^2 M^2 \gamma_{T}^{2}}{\textnormal{Var}\left(\sum_{t=1}^{[rT]}\eta_{t}(\xi_i)\right)+B^2 T^{1/2-2\kappa} +B T^{1/4-\kappa} T M \gamma_{T}\log(T)^{2}}\right) \nonumber  \\
 &\leq \exp\left(-\frac{c_{1} T^2 M^2 T^{-3/2} \log(T)^{2}}{C r T^{1/2} + B^2 T^{1/2-2\kappa} + B M T^{1/2-\kappa} \log(T)^{3}}\right) \nonumber \\
&= \exp\left(-\frac{c_{1} M^2 \log(T)^{2}}{C r + B^2 T^{-2\kappa}+ B M T^{-\kappa} \log(T)^{3}}\right) 
\label{eq-bernstein-2}
\end{align}%
where the second inequality utilizes the variance calculation in Lemma~\ref{lemma-beta-pt}. Let $T$ grow to infinity, then the second and third terms in the denominator of (\ref{eq-bernstein-2}) converge to $0$ for any finite $M$ and $B$. Therefore, (\ref{eq-bernstein-1}) is less than $\exp\left(-\frac{c_{1} M^2}{C} \log(T)^{2}\right)$ in large samples. With (\ref{eq-prob-P}), conclude that, 
 \begin{align*}
\textnormal{Pr}\left(|P(\xi)| \geq M \gamma_{T}\right) &\leq N(r) N(\tau) N(\xi) \cdot  \exp\left(-\frac{c_{1} M^2}{C} \log(T)^{2}\right) \\
&= \exp\left(\left( \frac{3}{2}+\frac{p}{4}+d \right) \log(T)-\frac{c_{1} M^2}{C} \log(T)^{2}\right) = o(T^{-1})
\end{align*}%
for any finite $M>0$. Next, consider 
$$Q(\xi) =  \max_{1 \leq k \leq T} \max_{1 \leq j \leq N(\tau)} \max_{1 \leq i \leq N(\xi)} \sup_{\xi \in \Delta \cup E_{i}} | \widetilde{R}(k/T,\tau_{j},\xi) - \widetilde{R}(k/T,\tau_{j},\xi_{i})|.$$ 
Given that $\xi \in E_{i}$, $Q(\xi)$ is bounded from above by $$\max_{1 \leq k \leq T} \max_{1 \leq j \leq N(\tau)} \max_{1 \leq i \leq N(\xi)} \frac{1}{T}\sum_{t=1}^{k} \bar{\eta}_{t}(\xi_{i}) x_{t}$$%
where 
$\bar{\eta}_{t}(\xi_{i}) = \big\{ I\left(e_{t}(\tau) \leq x_{t}^{\prime}\xi_{i}/\sqrt{T} + b_{T}s_{t}/\sqrt{T}\right) - I\left(e_{t}(\tau) \leq x_{t}^{\prime}\xi_{i}/\sqrt{T} - b_{T}s_{t}/\sqrt{T}\right) - F_{t}\left(x_{t}^{\prime}\xi_{i}/\sqrt{T} + b_{T}s_{t}/\sqrt{T} \right) +F_{t}\left(x_{t}^{\prime}\xi_{i}/\sqrt{T} - b_{T}s_{t}/\sqrt{T} \right) \big\}$ and $s_{t} = \sum_{j=1}^{p}|x_{tj}|$. Apply the Bernstein inequality to this upper bound. Similarly for $\eta_{t}(\xi)$, one can determine that $\textnormal{Var}\left(\frac{1}{\sqrt{T}} \sum_{t=1}^{[rT]} \bar{\eta}_{t}(\xi_{i}) x_{t}\right)$ is of order $O\left(b_{T}/\sqrt{T}\right)$. Given the choice of $b_T$, $\textnormal{Var}\left(\sum_{t=1}^{[rT]}\bar{\eta}_{t}(\xi_{i}) x_{t}\right) \leq C r T^{1/4}$ where $C$ is a constant. Set $\bar{\gamma}_{T} = T^{-3/4-\kappa} \log(T)^{3}$. By the union bound, for any $M > 0$,
\begin{equation*}
\textnormal{Pr}\left(|Q(\xi)| \geq M \, \bar{\gamma}_{T}\right) \leq \sum_{k=1}^{N(r)} \sum_{j=1}^{N(\tau)} \sum_{i=1}^{N(\xi)} \textnormal{Pr}\left( \big| \sum_{t=1}^{[rT]}\bar{\eta}_{t}(\xi_{i}) x_{t} \big| \geq T M \, \bar{\gamma}_{T} \right).
\end{equation*}%

The summand in the right hand side is bounded by 
\begin{align*}
&\textnormal{Pr}\left(\big|\sum_{t=1}^{[rT]}\bar{\eta}_{t}(\xi_{i}) x_{t} \big| \geq T M \, \bar{\gamma}_{T}\right) \\
&\leq \exp\left(-\frac{c_{1} T^2 M^2 \bar{\gamma}_{T}^{2}}{\textnormal{Var}\left(\sum_{t=1}^{[rT]}\bar{\eta}_{t}(\xi_i)x_{t}\right)+ B^{2} T^{1/2-2\kappa} + B T^{1/4-\kappa} T M \bar{\gamma}_{T}\log(T)^{2}}\right) \\
 &\leq \exp\left(-\frac{c_{1} M^2 T^{1/2-2\kappa} \log(T)^{6}}{C r T^{1/4} + B^{2} T^{1/2-2\kappa} + B M T^{1/2-2\kappa} \log(T)^{5}}\right) \\
&\leq \exp\left(-\frac{c_{1} M^2 \log(T)}{C r T^{-1/4+2\kappa} \log(T)^{-5} + B^2 \log(T)^{-5} + B M }\right).
\end{align*}%
For $T$ sufficiently large, the first and second terms in denominator converge to $0$ for any finite $c_1$ and $M$. So in large samples the last expression is less than 
$\exp\left(- \frac{c_1 M}{B} \log(T)\right)$. Therefore, by choosing $M > \frac{B}{c_1}(5/2+p/4+d)$,
\begin{align*}
\textnormal{Pr}\left(|Q(\beta)| \geq M \eta_{T}\right) & \leq N(r) N(\tau) N(\xi) \exp\left(-\frac{c_1 M}{B} \log(T)\right) \\
& = \exp\left(-\left(\frac{c_1 M}{B} - \left(\frac{3}{2}+\frac{p}{4}+d\right)\right) \log(T)\right) = o(T^{-1}).
\end{align*}%
Combining $P(\beta)$ and $Q(\beta)$, one can choose $M>0$ such that 

\begin{align*}
\textnormal{Pr}\left(\max_{1 \leq k \leq N(r)} \max_{1 \leq j \leq N(\tau)} \sup_{\xi \in \Delta} \|\widetilde{R}(r,\beta,\tau)\| > M \gamma_{T} \right) = o(T^{-1}).
\end{align*}

So, $\max_{1 \leq k \leq N(r)} \max_{1 \leq j \leq N(\tau)} \sup_{\xi \in \Delta} \|\widetilde{R}(r,\beta,\tau)\| = O_{p}\left(T^{-3/4} \log(T) \right)$, therefore, 

$\max_{1 \leq k \leq N(r)} \max_{1 \leq j \leq N(\tau)} \sup_{\xi \in \Delta} \|R(r,\beta,\tau)\| = O_{p}\left(T^{-1/4} \log(T) \right)$ as desired. 
\end{proof}

We now consider the next two terms on the right-hand side of (\ref{eq-the-goal}). 

\begin{proof}[Proof of Lemma~\ref{lemma-B-1}]: To prove the part (i) of Lemma~\ref{lemma-B-1}, let $B_{1}(r,\tau_{j}) = T^{-1/2}\sum_{t=1}^{[rT]} \zeta_{t} x_{t}$ where
$\zeta_{t} = I(y_{t} \leq x_{t}^{\prime}\beta_{0}(\tau_{j})) - I(y_{t} \leq x_{t}^{\prime}\beta_{0}(\tau_{j-1})) - (\tau_{j} - \tau_{j-1}). $
The proof is analogous to the proofs of Lemmas~\ref{lemma-beta-pt} and \ref{lemma-beta-uf}. Note that $\zeta_{t}$ is a centered Bernoulli variable with success probability $\delta = \tau_{j} - \tau_{j-1}$. Proceed similarly to the proof of Lemma~\ref{lemma-beta-pt}, one can show that there exists a constant $C>0$ such that $\textnormal{Var}\left(T^{-1/2} \sum_{t=1}^{[rT]} \zeta_{t} x_{t} \right) \leq C r \delta$. Because  $\delta = T^{-1/2-d}$, this establishes the pointwise rate; $\| B_{1}(r,\tau_{j}) \| = O_{p}\left(T^{-1/4-d/2} \right)$. One can then extend it to the uniform result by applying the Bernstein inequality and obtain $\sup_{r \in [0,1]} \max_{1 \leq j \leq N(\tau)} \|B_{1}(r,\tau_{j})\| = O_{p}\left(T^{-1/4-d/2} \log(T) \right)$. The proof goes similarly as those of Lemma~\ref{lemma-beta-uf}, so we will omit the details.

To prove the part (ii) of Lemma~\ref{lemma-B-1}, let $B_2(r,\tau_{j}) = T^{-1/2} \sum_{t=1}^{[rT]} \upsilon_{t} x_{t}$ where
$$\upsilon_{t} = 
F_{t}\left(x_{t}^{\prime}\beta_{0}(\tau_{j}) + x_{t}^{\prime}\xi/\sqrt{T}\right) - F_{t}\left(x_{t}^{\prime}\beta_{0}(\tau_{j-1}) + x_{t}^{\prime}\xi/\sqrt{T}\right) - F_{t}\left(x_{t}^{\prime}\beta_{0}(\tau_{j}) \right) + F_{t}\left(x_{t}^{\prime}\beta_{0}(\tau_{j-1}) \right).$$ 

By the Taylor series expansion, 
$\upsilon_{t} = \left\{ f_{t}\left(x_{t}^{\prime}\beta_{0}(\tau_{j}) \right) - f_{t}\left(x_{t}^{\prime}\beta_{0}(\tau_{j-1}) \right) \right\} x_{t} x_{t}^{\prime} \, \xi/\sqrt{T} + O_{p}(T^{-1})$ under Assumption~3 (v), (vii). The first order term in $B_2(r,\tau_{j})$ is 

\noindent $T^{-1} \sum_{t=1}^{[rT]} \left\{ f_{t}\left(x_{t}^{\prime}\beta_{0}(\tau_{j}) \right) - f_{t}\left(x_{t}^{\prime}\beta_{0}(\tau_{j-1}) \right) \right\} x_{t} x_{t}^{\prime} \, \xi$. To study it, observe that by the mean value theorem, $x_{t}^{\prime}\beta_{0}(\tau_{j}) - x_{t}^{\prime}\beta_{0}(\tau_{j-1}) = Q(\tau_{j}|x_{t}) - Q(\tau_{j-1}|x_{t}) = \frac{\tau_{j} - \tau_{j-1}}{f_{t}(\tilde{q}_{t})}$ where $x_{t}^{\prime}\beta_{0}(\tau_{j-1}) \leq \tilde{q}_{t}\leq x_{t}^{\prime}\beta_{0}(\tau_{j})$. Because $f_{t}(\cdot)$ is bounded away from zero uniformly, $x_{t}^{\prime}\beta_{0}(\tau_{j}) - x_{t}^{\prime}\beta_{0}(\tau_{j-1}) = O(\delta) = O(T^{-1/2-d})$ for any $j$. By the mean value theorem, $f_{t}\left(x_{t}^{\prime}\beta_{0}(\tau_{j})\right) - f_{t}\left(x_{t}^{\prime}\beta_{0}(\tau_{j-1}) \right) = f_{t}^{\prime}\left(x_{t}^{\prime}\beta_{0}(\tilde{\tau}_{j}) \right) \cdot \left( x_{t}^{\prime}\beta_{0}(\tau_{j}) - x_{t}^{\prime}\beta_{0}(\tau_{j-1}) \right)$ where $\tilde{\tau}_{j} \in [\tau_{j-1}, \tau_{j}]$. By Assumption~3(v), $| f_{t}\left(x_{t}^{\prime}\beta_{0}(\tau_{j})\right) - f_{t}\left(x_{t}^{\prime}\beta_{0}(\tau_{j-1}) \right) | \leq c |x_{t}^{\prime}\beta_{0}(\tau_{j}) - x_{t}^{\prime}\beta_{0}(\tau_{j-1})|$ for some $c > 0$ for any $j$. Because $T^{-1}\sum_{t=1}^{[rT]} x_{t} x_{t}^{\prime} \, \xi \rightarrow r Q \xi$, the first order term in $B_2(r,\tau_{j})$ is less than $r Q \xi \times O_{p}(T^{-1/2-d})$. Because $\Delta$ is a compact set in $R^{p}$, conclude that $\sup_{r \in [0,1]} \max_{1 \leq j \leq N(\tau)} \sup_{\xi \in \Delta} \|B_2(r,\tau_{j})\| = O_{p}(T^{-1/2-d})$. 
\end{proof}

\noindent We now prove Lemma~\ref{lemma-U} which concerns the second term in (\ref{eq-R-U}).  

\begin{proof}[Proof of Lemma~\ref{lemma-U}]: By the Taylor expansion,

\vskip -0.2in

\begin{align*}
U\left(r, \tau, \xi \right) &= T^{-1/2} \sum_{t=1}^{[rT]} x_{t} \left(F_{t}\left(x_{t}^{\prime}\beta_{0}(\tau) + x_{t}^{\prime}\xi/\sqrt{T} \right) - F_{t}(x_{t}^{\prime}\beta_{0}(\tau) \right) \\
&= \frac{1}{T} \sum_{t=1}^{[rT]} f_{t}(x_{t}^{\prime}\beta_{0}(\tau))  x_{t} x_{t}^{\prime} \, \xi + \frac{1}{2 T} \sum_{t=1}^{[rT]} f_{t}^{\prime}(\bar{q}_{t})  x_{t} x_{t}^{\prime} \, \xi \xi^{\prime} x_{t} \frac{1}{\sqrt{T}},
\end{align*}%
where $x_{t}^{\prime}\beta_{0}(\tau) \leq \bar{q}_{t} \leq x_{t}^{\prime}\beta_{0}(\tau) + x_{t}^{\prime}\xi/\sqrt{T}$, uniformly in $\tau$ and $\xi$. The first order term converges to $r D(\tau) \xi$, for $T$ large enough, under Assumption 3(vi). The second order term can be shown to be the order of $O_{p}(T^{-1/2})$ uniformly in $r$ under Assumption~3(v), (vii). The details are omitted for brevity. Conclude that for $T$ large, $U\left(r, \tau, \xi \right)  = r D(\tau) \xi + O_{p}(T^{-1/2})$, uniformly in $r \in [0,1]$, $\tau \in \mathcal{T}$, and $\xi \in \Delta$. 

\end{proof}

\begin{proof}[Proof of Theorem \ref{thm_2}, for cases (ii) and (iii)]: To complete the proof of Theorem~\ref{thm_2}, we analyze the first term in \eqref{eq:sum_byparts} for the second and third cases.

\noindent (ii) When $\mathcal{K}(.)$ is continuous, $\mathcal{K}(x) = 0$ for $|x|\geq 1$, and $\mathcal{K}(x)$ is twice continuously differentiable except for $|x|=1$, we have, uniformly in $\tau\in\mathcal T$,
    \begin{align}
         & \sum_{i=1}^{T-1}\sum_{j=1}^{T-1} \widehat S(i,\tau)(k_{ij} - k_{i,j+1}-k_{i+1,j}+k_{i+1,j+1}) \widehat S(j,\tau)^{\prime} \nonumber \\
         =& -\frac{1}{b^2} \frac{1}{T^2} \sum_{i=1}^{T-1}\sum_{j=1}^{T-1} 1\{|i-j|<M\}\mathcal{K}''\left(\frac{r_i-s_j}{b}\right)\widehat S(i,\tau)\widehat S(j,\tau)^{\prime} \nonumber \\
         &+\frac{1}{bT} \sum_{i=1}^{T-M-1} K'_{-}\left(1\right)\widehat S(i,\tau)\widehat S(i+M,\tau)^{\prime}+ \frac{1}{bT} \sum_{j=1}^{T-M-1} K'_{-}\left(1\right)\widehat S(j+M,\tau)\widehat S(j,\tau)^{\prime} + o_p(1)
        \label{eq:fixedb_case2}
    \end{align}
    where $K'_{-}\left(1\right)$ is the derivative of $\mathcal{K}(x)$ from the left at $x=1$.  For the first term in \eqref{eq:fixedb_case2}, the same argument as step (i) gives
    \begin{align*}
      &  -\frac{1}{b^2} \frac{1}{T^2} \sum_{i=1}^{T-1}\sum_{j=1}^{T-1} 1\{|i-j|<M\}\mathcal{K}''\left(\frac{r_i-s_j}{b}\right)\widehat S(i,\tau)\widehat S(j,\tau)^{\prime}\\
        &\rightarrow -\frac{1}{b^2}\int\int_{|r-s|<b}\mathcal{K}''\left(\frac{|r-s|}{b}\right) \widetilde{G}(r,\tau)\widetilde{G}(s,\tau)^{\prime}drds \qquad \textnormal{in} \quad  \ell^{\infty}(\mathcal{T}).
    \end{align*}

    The second and the third terms in \eqref{eq:fixedb_case2} are symmetric, so it suffices to show the result for the second term. A similar argument follows as step (i) by noting that (1) the uniform continuity and and convergence of the corresponding functional do not depend on the kernel except for the boundedness of $K'_{-}(1)$ and (2) the double sum is replaced by the single sum. Therefore,
    \begin{equation*}
        \frac{1}{bT} \sum_{i=1}^{T-M-1} K'_{-}\left(1\right)\widehat S(i,\tau)\widehat S(i+M,\tau)' \to \frac{K'_{-}(1)}{b} \int_0^{1-b} \widetilde{G}(r,\tau) \widetilde{G}(r+b,\tau)^{\prime}  dr \qquad \textnormal{in} \quad  \ell^{\infty}(\mathcal{T}).
    \end{equation*}
    Due to the symmetry of the second and the third terms, we establish the claim for the second case.

\noindent (iii) When $\mathcal{K}(.)$ is the Bartlett kernel, we have
    \begin{align}
          \sum_{i=1}^{T-1}\sum_{j=1}^{T-1}\widehat S(i,\tau) (k_{ij} - k_{i,j+1}-k_{i+1,j}+k_{i+1,j+1})\widehat S(j,\tau)^{\prime} =\frac{T-1}{T} \frac{2}{M} \sum_{i=1}^{T-1}\widehat S(i,\tau)\widehat S(i,\tau)^{\prime} \nonumber \\
         - \frac{T-1}{T} \frac{1}{M} \sum_{i=1}^{T-M-1}\widehat S(i+M)\widehat S(i)^{\prime}- \frac{T-1}{T} \frac{1}{M} \sum_{i=1}^{T-M-1}\widehat S(i,\tau)\widehat S(i+M,\tau)^{\prime}.\label{eq:fixedb_case3}
    \end{align}
    Following step (i), a similar argument for the uniform continuity and the uniform convergence of the corresponding functionals in \eqref{eq:fixedb_case3} and $\widehat S \Rightarrow \tilde{G}$ in $\ell^{\infty}([0,1] \times \mathcal{T})$ together implies the claim for the third case. 
\end{proof}

The following proposition provides a sufficient condition for Assumption 1.

\vskip 0.1in

\begin{proposition}
    \label{prop:donsker}
    Assumption~3 implies that Assumption~\ref{assum_gaussian_limit} is satisfied.
\end{proposition}

\begin{proof}
Let $f_j(\tau) = x_{tj} (\tau - 1(y\leq x_{t}^\prime\beta_0(\tau))) $ and denote the function class $F_j = \{f_j(\tau): \tau \in \mathcal T \}$ for each $j$, and let $F = F_1\times \cdots \times F_p$. 

The proof proceeds by verifying the conditions of Lemma~\ref{lemma_fn_donsker} below:
\begin{enumerate}
    \item For some $q>2$, $\textnormal{E}\left[\sup_{(f_1, \ldots ,f_p)\in F}\left\| (f_1, \ldots ,f_p) \right\|\right]^{q} \leq \textnormal{E}\|x_t\|^{q} < \infty$. The last inequality follows from Assumption~3(i). 
    \item The beta-mixing coefficients satisfy $l^{\frac{q}{q-2}}(\log l )^{\frac{2(q-1)}{q-2}} \beta(l)\to 0$. This condition follows directly from Assumption~3(ii).
    \item $F$ is a VC-subgraph class: Let $\tilde F_j$ be a larger class $\tilde F_j = \{(y,x)\mapsto x_{tj}(\tau - 1\{y\leq x_{t}^\prime\beta: \tau \in \mathcal T, \beta\in \mathbb R^p\})\}$. Given $(\tau,\beta)$ and some arbitrary real number $m$, its subgraph is given by
    \begin{align*}
        &\left\{(x,y,m) \in \mathbb{R}^{p+2}: x_{j}(\tau - 1\{y \le x^\prime\beta\})>m   \right\}\\
        =& \left\{ \{ x_{j}>m/\tau\} \cap  \left\{ y_t> x^\prime\beta \right\}\right\} \cup \left\{ \{ x_{j}<m/(\tau-1)\} \cap  \left\{ y_t\le  x^\prime \beta \right\}\right\}.
    \end{align*}
    Note that the sets on the right-hand side are finite set operations of finite-dimensional half-spaces (Lemma 2.6.17 of van der Vaart and Wellner (1996)). Hence, $\tilde F_j$ and, therefore, $F_j$ is VC-subgraph. Thus, the finite union $F$ is also a VC-subgraph class.
\end{enumerate}
Thus, Lemma~\ref{lemma_fn_donsker} applies and $S([rT],\tau)\Rightarrow G(r,\tau)$ in $\ell^\infty([0,1]\times\mathcal T)$, as desired. 
\end{proof}


The following lemma, which is used to prove Proposition~\ref{prop:donsker},
provides sufficient conditions for weak convergence of the sequential
empirical process under weak dependence. Specifically, it establishes
a functional central limit theorem for a VC-subgraph class of functions of a strictly stationary, geometrically $\beta$-mixing sequence.


\begin{lemma}
\label{lemma_fn_donsker}
Let $\{X_t\}_{t\geq1}$ be a strictly
stationary, geometrically beta-mixing
sequence, and let $\mathcal F$ be a
pointwise measurable VC-subgraph class
with envelope $F$ satisfying
$PF^q<\infty$ for some $q>2$.
Then the sequential empirical process
$\mathbb Z_n(r,f)=
\frac{1}{\sqrt n}
\sum_{t=1}^{[rn]}
\left(f(X_t)-Pf\right)
$ converges weakly in
$\ell^\infty([0,1]\times\mathcal F)$
to a tight, centered Gaussian process
$\mathbb Z$ with covariance $\operatorname{Cov}
\left(\mathbb Z(r,f),
\mathbb Z(s,g)\right)
= (r\wedge s)\Gamma(f,g)$ where $ \Gamma(f,g)
= \sum_{h=-\infty}^{\infty} \operatorname{Cov} \left(f(X_0),g(X_h)\right)$.
\end{lemma}

\begin{proof}
Let $\mathbb G_n(f) =\frac{1}{\sqrt n}
\sum_{t=1}^n(f(X_t)-Pf)$ denote the ordinary empirical process.
Since $\mathcal F$ is VC-subgraph,
$PF^q<\infty$ for some $q>2$,
and the beta-mixing coefficients decay
geometrically, the conditions of
Theorem~2.1 of \cite{arcones1994central}
are satisfied. Hence,
$\mathbb G_n\Rightarrow\mathbb G$
in $\ell^\infty(\mathcal F)$,
where $\mathbb G$ is a tight,
centered Gaussian process with covariance
$\Gamma(f,g)$.

Let $d_{\mathbb G}$ denote the canonical
semimetric of $\mathbb G$.
Since $\mathcal F$ is Donsker,
$\mathbb G_n$ is asymptotically
equicontinuous under $d_{\mathbb G}$.
Thus, for every $\varepsilon>0$,
\begin{equation}
\lim_{\delta\downarrow0}
\limsup_{n\to\infty}
P\left(
\|\mathbb G_n\|_{\mathcal F_\delta}
>\varepsilon
\right)=0,
\label{eq:donsker_equicont}
\end{equation}
where $\mathcal F_\delta = \{f-g:f,g\in\mathcal F,\,
d_{\mathbb G}(f,g)<\delta\}$ and $\|\mathbb G_n\|_{\mathcal F_\delta}
:=
\sup_{\substack{r\in[0,1],\,f,g\in\mathcal F\\
\|f-g\|_{P,2}\leq\delta}}
|\mathbb G_n(r,f)-\mathbb G_n(r,g)|$.

We now extend this result to the sequential
process using the Ottaviani-type inequality
of B\"ucher (2015).
Since geometric beta-mixing implies
geometric alpha-mixing with the same rate, we can choose
$\ell_n\asymp\log n$ such that
$(n/\ell_n)\alpha(\ell_n)=o(1)$.
Moreover, the envelope moment condition
implies
\begin{align*}
\max_{\substack{0\le j<k\le n\\
k-j\le2\ell_n}}
\left\|
\frac{1}{\sqrt n}
\sum_{t=j+1}^k(f(X_t)-Pf)
\right\|_{\mathcal F} \le
\frac{C\ell_n}{\sqrt n}
\left(\max_{1\le t\le n}F(X_t)+PF\right)=
O_p(\ell_n n^{1/q-1/2})
=o_p(1),
\end{align*}
where the last two equalities follow from
$PF^q<\infty$ and $q>2$.
Since $\mathcal F_\delta$ consists of
differences of functions in $\mathcal F$,
the same bound holds for
$\mathcal F_\delta$. Applying the
Ottaviani-type inequality to
$\mathcal F_\delta$, the short-block
bound and geometric mixing give
\begin{equation}
\lim_{\delta\downarrow0}
\limsup_{n\to\infty}
P\left(
\sup_{r\in[0,1]}
\|\mathbb Z_n(r,\cdot)\|_{\mathcal F_\delta}
>\varepsilon
\right)=0.
\label{eq:donsker_seq_equicont}
\end{equation}

Next, consider equicontinuity in the
time index. Partition $[0,1]$ into
intervals of length at most $\delta$.
By stationarity, a union bound and
the same Ottaviani inequality, together
with the short-block and mixing bounds,
give
\begin{align*}
&\limsup_{n\to\infty}
P\left(
\sup_{|r-s|<\delta}
\|\mathbb Z_n(r,\cdot)
-\mathbb Z_n(s,\cdot)\|_{\mathcal F}
>6\varepsilon
\right)\le
\frac{C}{\delta}
P\left(
\|\mathbb G\|_{\mathcal F}
\geq\frac{\varepsilon}{\sqrt{\delta}}
\right),
\end{align*}
for some constant $C>0$.
Since $\mathbb G$ is a tight Gaussian
process, its supremum norm has Gaussian
tails by Fernique's theorem. Hence,
the right-hand side converges to zero
as $\delta\downarrow0$.

Finally, for any finite collection
$f_1,\ldots,f_m\in\mathcal F$,
the FCLT applies
because $PF^q<\infty$ and
$\sum_{k\geq1}\alpha(k)^{1-2/q}<\infty$.
Thus,
\[
\left(
\mathbb Z_n(\cdot,f_1),\ldots,
\mathbb Z_n(\cdot,f_m)
\right)
\Rightarrow
\left(
\mathbb Z(\cdot,f_1),\ldots,
\mathbb Z(\cdot,f_m)
\right),
\]
where $\operatorname{Cov}
\left(
\mathbb Z(r,f),\mathbb Z(s,g)
\right)
=(r\wedge s)\Gamma(f,g)$. 
Combining finite-dimensional convergence
with asymptotic equicontinuity in both
indices establishes $\mathbb Z_n\Rightarrow\mathbb Z$ in $
\ell^\infty([0,1]\times\mathcal F).$
\end{proof}

\vskip 0.5in


\setcounter{section}{0} \setcounter{equation}{0}%
\setcounter{definition}{0}\setcounter{assumption}{0} \setcounter{lemma}{0} \setcounter{table}{0} \setcounter{figure}{0}%
\setcounter{page}{1} \renewcommand{\thepage}{C-\arabic{page}} %
\renewcommand{\theequation}{C.\arabic{equation}}\renewcommand{\thelemma}{C.%
\arabic{lemma}}\setcounter{corollary}{0}\renewcommand{\thecorollary}{C.%
\arabic{corollary}} \renewcommand{\theassumption}{C.\arabic{assumption}}%
\renewcommand{\theremark}{C.\arabic{remark}}
\setcounter{section}{0}
\renewcommand{\thesection}{C}
\setcounter{subsection}{0}
\renewcommand{\thesubsection}{C.\arabic{subsection}}
\renewcommand{%
\thedefinition}{C.\arabic{definition}}
\renewcommand{\thetable}{C.\arabic{table}}
\renewcommand{\thefigure}{C.\arabic{figure}}

\section{Additional Materials}

\subsection{Test-Based Data-Driven Bandwidth Rules}
\label{sec:bandwidth_app}

Test-based bandwidth selection rules were introduced in \cite{SPJ2008} for the fixed-smoothing kernel-HAR inference method in a simple Gaussian location model, and \cite{S2014b} extends it to a GMM framework. For the OS-HAR inference method, the test-optimal bandwidth rule has been established in \cite{sun2011robust} for a generic joint test in linear regressions, and \cite{hwang2025har} extends this approach to quantile regressions. The existing results of test-optimal rules are based on higher-order expansions of the pivotal limiting distributions. In our setting, however, the fixed-smoothing limiting distribution is not pivotal. We nevertheless use the same logic as a practical guideline for choosing the smoothing parameter. To emphasize this distinction, we refer to the resulting choice as a test-based data-driven bandwidth rule rather than an optimal bandwidth rule.

Fix a testing problem indexed by $a \in \{W,S,H,M\}$, corresponding respectively to the Wald, Sup-t, homogeneity, and monotonicity procedures, and we denote the corresponding restriction matrix as $R_a$ and $d_a = \operatorname{rank}(R_a)$. For example, to test a generic null hypothesis $H_0: r(\beta(\bar \tau_m)) =0$ using the Wald statistic, we set $R_W = \bar R_0$ as defined in Section \ref{sec:stacking}. To test the null hypothesis $H_0: \beta_j( \tau_1)=\cdots =\beta_j( \tau_m)=0$ using the Wald statistic or the Sup-t test, we set $R_W = R_S $ as a $m\times mp$ matrix with $(i, (i-1)p+j)$-th entry being 1 and all else 0, for $i=1,\ldots, m$. For the homogeneity test, we set $R_H$ as a $(m-1)\times mp$ matrix with $(i, (i-1)p+j)$-th entry being $-1$, $(i, ip+j )$ being $1$ and all else 0, for $i=1,\ldots, m$. For the monotonicity test, set $R_M$ the same way as $R_H$.

For the kernel-HAR estimator, we adapt the test-based bandwidth rule of \cite{S2014b} for a GMM framework to a quantile process.\footnote{The test-based bandwidth rule of \cite{S2014b} is derived for an asymptotic F test using a bias-corrected Wald statistic with F distribution as the reference distribution. \cite{S2014b} also shows that both the asymptotic F test and a test using the Wald statistic with fixed-b reference distribution are second-order correct under a small-b approximation. In a sense, they are asymptotically equivalent tests, and we therefore adapt the bandwidth rule for the asymptotic F test to our Wald test.} Define the VAR(1)-implied bias estimator of the Bartlett-kernel-HAR estimator and the long-run variance estimator:\footnote{To simplify the exposition, we focus on the Bartlett kernel, which is also our main focus in the simulation. For the Parzen and QS kernels, the test-bandwidth rules can also be adapted from the corresponding literature mentioned above.}
\begin{align*}
    \widehat{B}(\bar\tau_m)=& - \left(\widehat{A}(I_{mp} - \widehat{A})^{-2}\widehat \Gamma(0) +\widehat \Gamma(0)\widehat{A}'(I_{mp} - \widehat{A}')^{-2}\right),\\
 \widehat{\Omega}(\bar\tau_m) =& (I_{mp} - \widehat{A})^{-1} \widehat{\mathcal{E}}(I_{mp} -\widehat{A}')^{-1},
\end{align*}
where $\widehat A$ is the estimated coefficient matrix from the VAR(1) model of the stacked QR score process, i.e. $\widehat{s}_t(\bar\tau_m) = A \widehat{s}_{t-1}(\bar\tau_m) + \epsilon_t$ for $A\in \mathbb{R}^{mp\times mp}$ and $\textnormal{E}[\epsilon_{t}\epsilon_{t}'] = \mathcal{E}$, where $\widehat{s}_t(\bar \tau_m)$ corresponds to $s_t(\bar \tau_m)$ defined in Section \ref{sec:stacking}; $\widehat \Gamma(0)$ is the solution to $\widehat \Gamma(0) = \widehat A\widehat \Gamma(0)\widehat A' +\widehat{\mathcal{E}}  $ and $\widehat{\mathcal{E}} =\frac{1}{T} \sum_{t=1}^T \widehat{\epsilon}_t  \widehat{\epsilon}_t'$. Furthermore, define the bias term of the auxiliary score:
\begin{align*}
    \bar{B}_{a}(\bar\tau_m) = \frac{1}{d_a} \operatorname{tr} \left( \left[R_a \widehat D(\bar \tau_m) ^{-1}\widehat{B}(\bar\tau_m) \widehat D(\bar \tau_m)^{-1}R_a'\right] \left[ R_a \widehat D(\bar \tau_m)^{-1}\widehat{\Omega}(\bar\tau_m) \widehat D(\bar \tau_m)^{-1}R_a' \right]^{-1} \right).
\end{align*}
Let $\kappa>1$ denote the tolerated size inflation. Following \cite{hwang2025har}, $\kappa$ can be set as $1.05$ for a larger sample size $T$, and $1.1$ otherwise.    Then, the implied bandwidth choice is
\[
\widetilde M_{a}=
\begin{cases}
\left[
\dfrac{3 g_{d_a,\delta^2}\left(Q_{\chi^2_{d_a}}^{(1-\alpha)}\right)\bar{B}_a(\bar\tau_m) }
{\delta^2 g_{d_a+2,\delta^2}\left(Q_{\chi^2_{d_a}}^{(1-\alpha)}\right) }
\right]^{1/2} T^{1/2},
& \bar{B}_a(\bar\tau_m)>0, \\[1.0em]
\dfrac{g_{d_a}\left(Q_{\chi^2_{d_a}}^{(1-\alpha)}\right)Q_{\chi^2_{d_a}}^{(1-\alpha)}|\bar{B}_a(\bar\tau_m)| }
{(\kappa-1)\alpha},
& \bar{B}_a(\bar\tau_m)\le 0,
\end{cases}
\]
where $g_{d,\delta^2}(\cdot)$ denotes the density of a non-central $\chi^{2}$ random variable with $d$ degrees of freedom and non-centrality parameter $\delta^{2}$.\footnote{Following the aforementioned literature, when no economically meaningful local alternative is available, we choose $\delta^{2}$ so that $P\!\left(\chi^{2}_{d_a}(\delta^{2})>Q_{\chi^2_{d_a}}^{(1-\alpha)}\right)=0.75$,
where $\chi^2_{d}(\delta^2)$ denotes a $\chi^2$-distributed random variable with $d$-degree of freedom the non-centrality parameter $\delta^2$ and $Q_{\chi^2_d}^{(1-\alpha)}$ denotes $(1-\alpha)$-th quantile of the $\chi^2$ distribution with $d$-degree of freedom.} Then, we set $\widehat M_{a}=\min\{T-1,\max\{1,\lfloor \widetilde M_{a}\rfloor\}\}$.

For the OS-HAR estimator, we adapt the test-based rule for a single quantile from \cite{hwang2025har} to multiple quantiles. The VAR(1)-implied long-run variance of the QR score remains the same, but the bias term of the OS-HAR estimator is given by
\begin{align*}
    \bar{\bar{B}}_{a}(\bar\tau_m) = &\frac{1}{d_a} \operatorname{tr} \left( \left[R_a \widehat D(\bar \tau_m)^{-1} \widetilde{B}(\bar\tau_m) \widehat D(\bar \tau_m)^{-1} R_a'\right] \left[ R_a \widehat D(\bar \tau_m)^{-1}\widehat{\Omega}(\bar\tau_m) \widehat D(\bar \tau_m)^{-1}R_a' \right]^{-1} \right), \\
\widetilde{B}(\bar\tau_m) =& -\frac{\pi^2}{6} (I_{mp} - \widehat{A})^{-3} \Bigl( \widehat{A} \widehat{\mathcal{E}}+ \widehat{A}^2 \widehat{\mathcal{E}}\widehat{A}' + \widehat{A}^2 \widehat{\mathcal{E}}- 6\widehat{A} \widehat{\mathcal{E}}\widehat{A}'+ \widehat{\mathcal{E}}(\widehat{A}')^2 + \widehat{A} \widehat{\mathcal{E}}(\widehat{A}')^2 + \widehat{\mathcal{E}}\widehat{A}' \Bigr) (I_{mp} - \widehat{A}')^{-3}.
\end{align*}
Then, we choose
\[
\widetilde K_{a}=
\begin{cases}
\left[
\dfrac{\delta^{2} g_{d_a+2,\delta^{2}}\!\left(Q_{\chi^2_{d_a}}^{(1-\alpha)}\right)}
{4\, g_{d_a,\delta^{2}}\!\left(Q_{\chi^2_{d_a}}^{(1-\alpha)}\right)\, |\bar{\bar{B}}_a(\bar\tau_m)|}
\right]^{1/3} T^{2/3},
& \bar{\bar{B}}_a(\bar\tau_m)>0, \\[1.0em]
\left[
\dfrac{(\kappa-1)\alpha}
{g_{d_a}\!\left(Q_{\chi^2_{d_a}}^{(1-\alpha)}\right)Q_{\chi^2_{d_a}}^{(1-\alpha)}\, |\bar{\bar{B}}_a(\bar\tau_m)|}
\right]^{1/2} T,
& \bar{\bar{B}}_a(\bar\tau_m)< 0,
\end{cases}
\]
and choose $\widehat K_{a}=\min\{\max\{d_a+4,\lceil \widetilde K_{a}\rceil\}, T\}$ for tests that requires variance matrix inversion, and $\widehat K_{a}=\min\{\max\{1,\lceil \widetilde K_{a}\rceil\}, T\}$, otherwise.

To implement the formulas above, it remains to obtain the long-run variance that appears in the bandwidth formulas above. We can obtain $\widehat{A}$ as the least squares estimator for the VAR(1) model directly. Alternatively, to avoid substantial finite-sample bias in the estimation of $A_s$ due to the non-smooth quantile score and the computational burden when $m$ is large, we adapt the idea of \citet{galvao2024hac} to multiple quantiles. Let $\widehat{A}_x$ be the estimated VAR(1) coefficient matrix for $x_t$. 
Define $r_t(\bar\tau_m) = (\tau_1-I(e_t(\tau_1)\leq 0 ),\ldots, \tau_m-I(e_t(\tau_m)\leq 0))'$. To approximate the VAR(1) coefficients $A_r$ for $r_t(\bar\tau_m)$, we assume each pair $(e_t(\tau_i), e_{t-1}(\tau_j))$ are jointly normal for $i,j=1,\ldots,m$. We define the coefficient matrix as a linear projection $A_r =  \textnormal{E}[r_{t}r_{t-1} '] \textnormal{E}[r_{t-1} r_{t-1}']^{-1}$ where we suppress its dependence on $\bar\tau_m$ for convenience. Note that $\textnormal{E}[(\tau_i-I(e_t(\tau_i)\leq 0 )) (\tau_j-I(e_t(\tau_j) \leq 0 ))] =\tau_i \wedge \tau_j -\tau_i\tau_j$. If each $e_t(\tau_i)$ has variance one and by pairwise joint normality, we can approximate the $(i,j)$-th entry of $\textnormal{E}[r_{t}r_{t-1} ']$ by $ P(e_{t}(\tau_i)\le 0, e_{t-1}(\tau_j)\le 0 ) - \tau_i\tau_j = \Phi_2\left(\Phi^{-1}(\tau_i),\Phi^{-1}(\tau_j); \operatorname{cov}(e_{t}(\tau_i),e_{t-1}(\tau_j))\right)- \tau_i\tau_j$,
where $\Phi_2(.;\rho)$ is a bivariate normal CDF with mean 0,  variance 1, and covariance $\rho$.
Therefore, we proceed as follows:
\begin{enumerate}
    \item Standardize the residuals $\widehat{e}_t(\tau_j)$ for each $j=1,\ldots, m$ to obtain $\tilde{e}_t(\tau_j)$. Then, obtain the sample covariance $\widehat{\operatorname{cov}}(\tilde{e}_{t}(\tau_i),\tilde{e}_{t-1}(\tau_j))$ for each $(i,j)$ pair.

    \item Let $\widehat{\textnormal{E}}[r_{t-1}r_{t-1} ']$ be such that each $(i,j)$-th entry is $\tau_i\wedge \tau_j -\tau_i\tau_j$. Let $\widehat{\textnormal{E}}[r_{t}r_{t-1} ']$ be such that each $(i,j)$-th entry is estimated by $\Phi_2\left(\Phi^{-1}(\tau_i),\Phi^{-1}(\tau_j); \widehat{\operatorname{cov}}(\tilde{e}_{t}(\tau_i),\tilde{e}_{t-1}(\tau_j))\right)- \tau_i\tau_j$.     
    Then, set $\widehat{A}_r =\widehat{\textnormal{E}}[r_{t}r_{t-1} ']\widehat{\textnormal{E}}[r_{t-1}r_{t-1} ']^{-1}$.
     
\end{enumerate}
With $\widehat{A}_{x}$ and $\widehat{A}_{r}$, we obtain the alternative estimator for $A$ as $\widehat{A} = \widehat{A}_{r}  \otimes \widehat{A}_{x}$.

\subsection{Additional Results in Empirical Analysis}

The main text presents the quantile predictive effects of stock variance on returns. This section reports the corresponding results for three additional predictors: the inflation rate, long-term returns, and net equity expansion.

\begin{figure}
     \centering
        \caption{Quantile Effects on S\&P 500 Returns Across Predictors}
        \label{fig-app3}     
     \begin{subfigure}[b]{0.5\textwidth}
         \centering
         \includegraphics[width=1.1\linewidth, height=5.5cm]{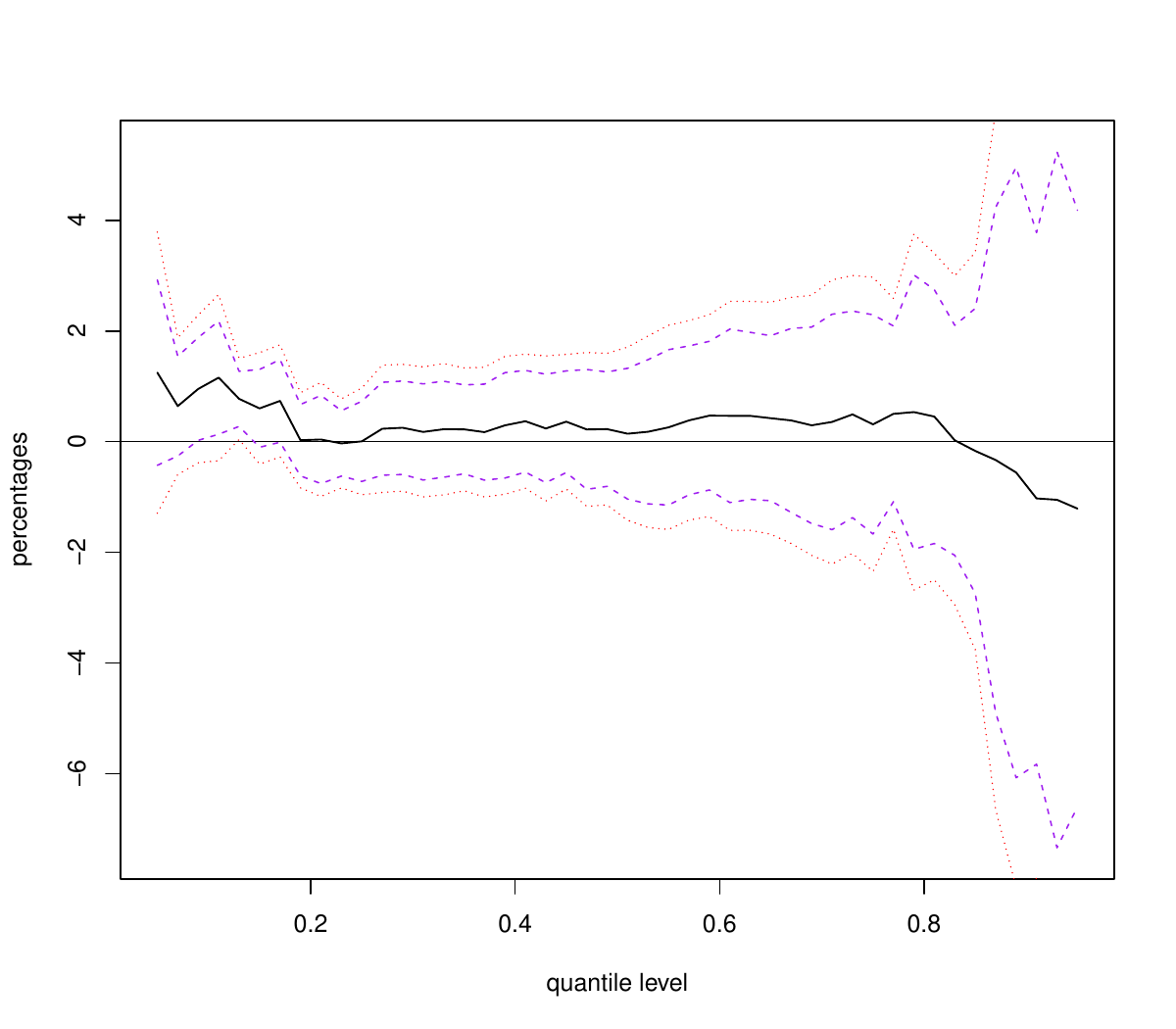}
         \caption{Inflation Rate}
         \label{fig:y equals x}
     \end{subfigure}
     \hfill
     
     \vskip -0.2in
     \begin{subfigure}[b]{0.5\textwidth}
         \centering
         \includegraphics[width=1.1\linewidth, height=5.5cm]{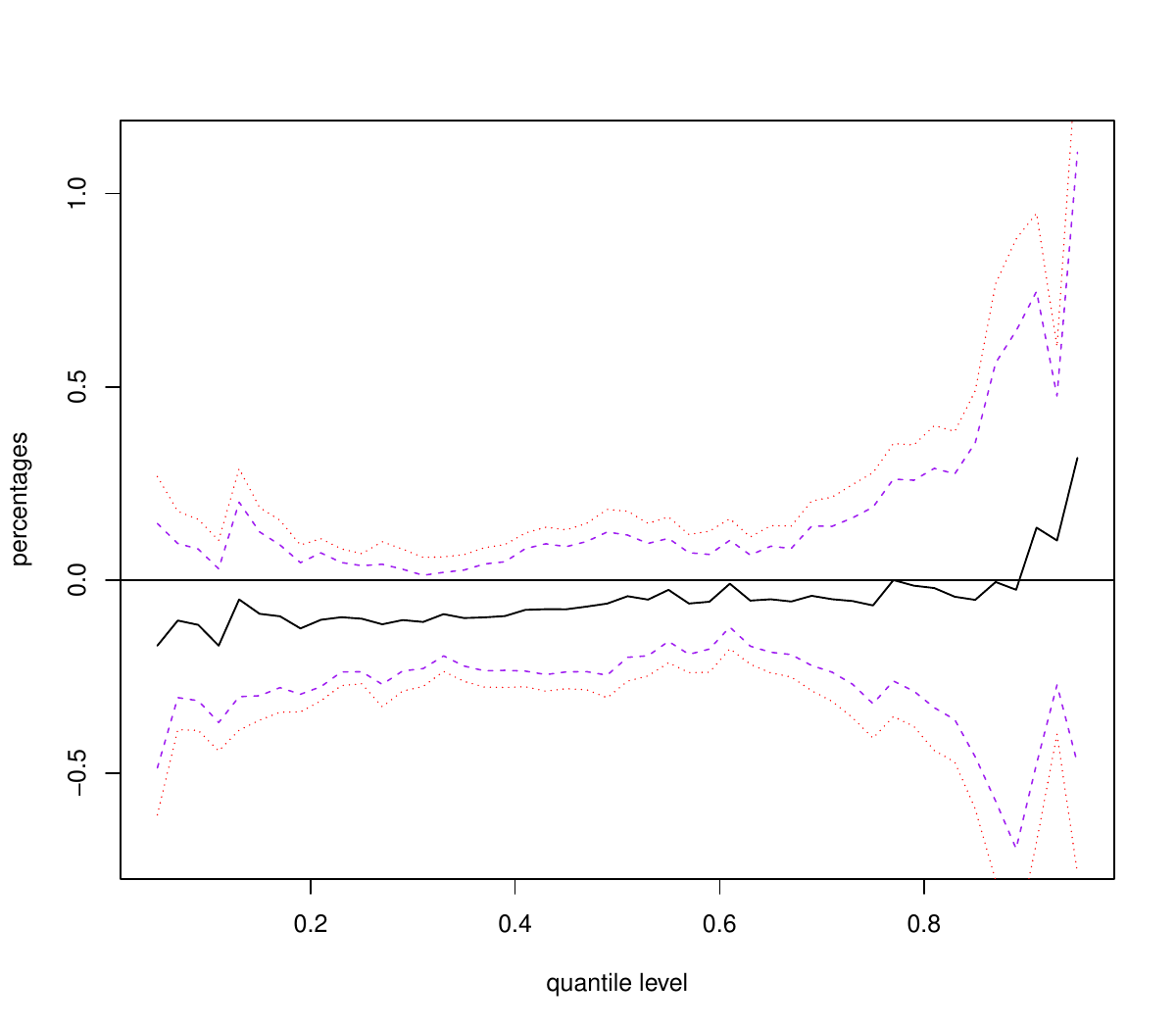}
         \caption{Long-Term Return}
         \label{fig:three sin x}
     \end{subfigure}
     \hfill
     \vskip -0.2in

     \begin{subfigure}[b]{0.5\textwidth}
         \centering
         \includegraphics[width=1.1\linewidth, height=5.5cm]{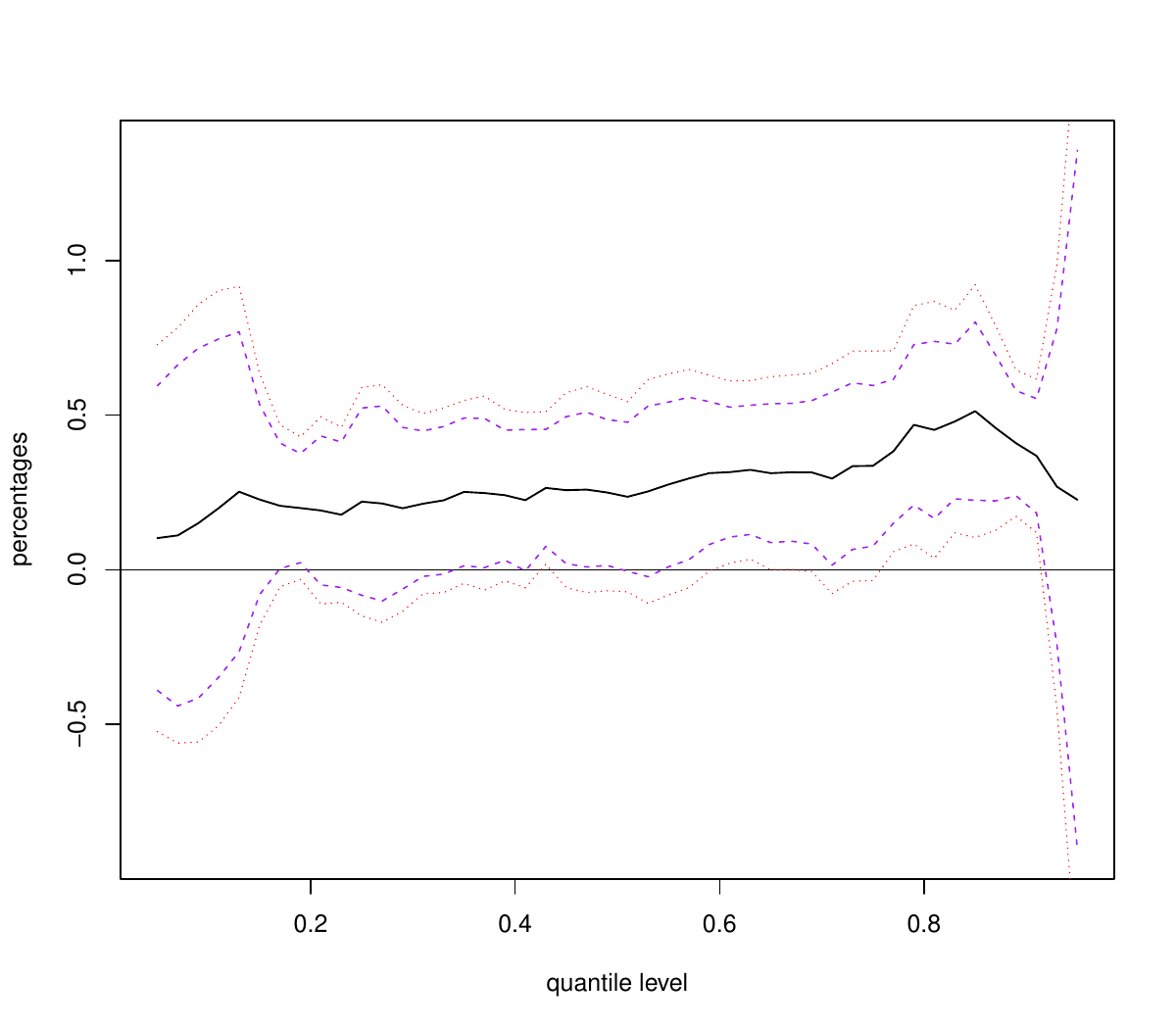}
         \caption{Net Equity Expansion}
         \label{fig:five over x}
     \end{subfigure}
\par
\begin{center}
\small{
\begin{minipage}[l]{14.0cm}
Each panel reports estimated quantile regression coefficients of monthly S\&P 500 excess returns on a lagged predictor. The solid line shows the estimates, while the red and purple dotted lines denote the 95\% uniform and pointwise confidence bands, respectively.
\end{minipage}
}
\end{center}        
\end{figure}

Panel (a) of Figure~\ref{fig-app3} shows the inflation quantile effect on stock returns. The red dotted lines represent the 95\% uniform confidence band, while the purple dashed lines denote the corresponding pointwise confidence bands. The two bands suggest somewhat different interpretations. From the pointwise perspective, the estimated inflation effect is generally close to zero, but appears positive and statistically significant in the left tail, which may be interpreted as evidence that higher inflation is associated with reduced downside risk in stock returns. However, the uniform band indicates that the effect is largely indistinguishable from zero across most quantiles. The only visible exception occurs at $\tau = 0.13$, where the estimated inflation effect is $0.773$ with a 95\% uniform confidence interval of $(0.113, 1.434)$.

Turning to formal inference, the test statistic for the global significance hypothesis is $3.800$, with corresponding critical values of $3.241$ and $4.338$ at the 5\% and 1\% levels, respectively. Thus, the null hypothesis is rejected at the 5\% level but not at the 1\% level, indicating at most weak evidence that inflation has a nonzero effect at some quantiles.

The uniform inference methods are then used to examine the effects of long-term returns and net equity expansion. The long-term return is interpreted as a proxy for equity risk, while net equity expansion is a measure of corporate issuing activity.\footnote{Net equity expansion is measured as the ratio of $12$-month net equity issuance by NYSE-listed firms to their end-of-year market capitalization.}
Panels (b) and (c) of Figure~\ref{fig-app3} report the estimated quantile effects of these predictors along with their 95\% confidence bands. The two variables display sharply contrasting patterns.

Panel (b) of Figure~\ref{fig-app3} reveals little evidence that long-term returns have any meaningful predictive content across the distribution. The estimated quantile effects remain close to zero at all quantiles, and the uniform confidence band consistently covers zero. This lack of effect is confirmed by formal testing: the test statistic for the global significance hypothesis is $1.007$, well below the 5\% critical value of $3.423$. Taken together, both graphical and statistical evidence suggest that long-term returns contribute little to explaining variation in stock returns at any part of the distribution.

In contrast, panel (c) of Figure~\ref{fig-app3} presents a very different picture for net equity expansion. The estimated effects are uniformly positive and exhibit a degree of stability across quantiles, with the confidence band indicating statistical significance over a broad range of the distribution. This visual impression is supported by the formal test: the test statistic for the global significance hypothesis, $5.989$, exceeds the 95\% critical value of $3.461$. These results provide strong evidence that net equity expansion is positively associated with future stock returns across quantiles.

\subsection{Additional Simulations: Tests of Shape Restrictions}

This subsection considers performance of the proposed tests for shape restrictions. 
In practice, the homogeneity test is implemented using the procedure described in Section \ref{sec:test_shape}. For the monotonicity test, we consider adjacent quantile comparisons and test $H_0: \beta_1(\tau_i) -\beta_1(\tau_{i-1}) \geq 0$ for $i=2,\ldots,m$ against $H_1: \beta_1(\tau_i) -\beta_1(\tau_{i-1}) < 0 $ for at least one $i$. Restricting attention to adjacent pairs reduces the computational burden. The monotonicity test can be viewed as a one-sided version of the homogeneity test.

\begin{figure}[htbp]

    \caption{Rejection Rates of the Homogeneity Test}
    \label{fig:homo_combined}

    \centering

    \begin{subfigure}{\linewidth}
        \centering
        \includegraphics[width=0.9\linewidth]{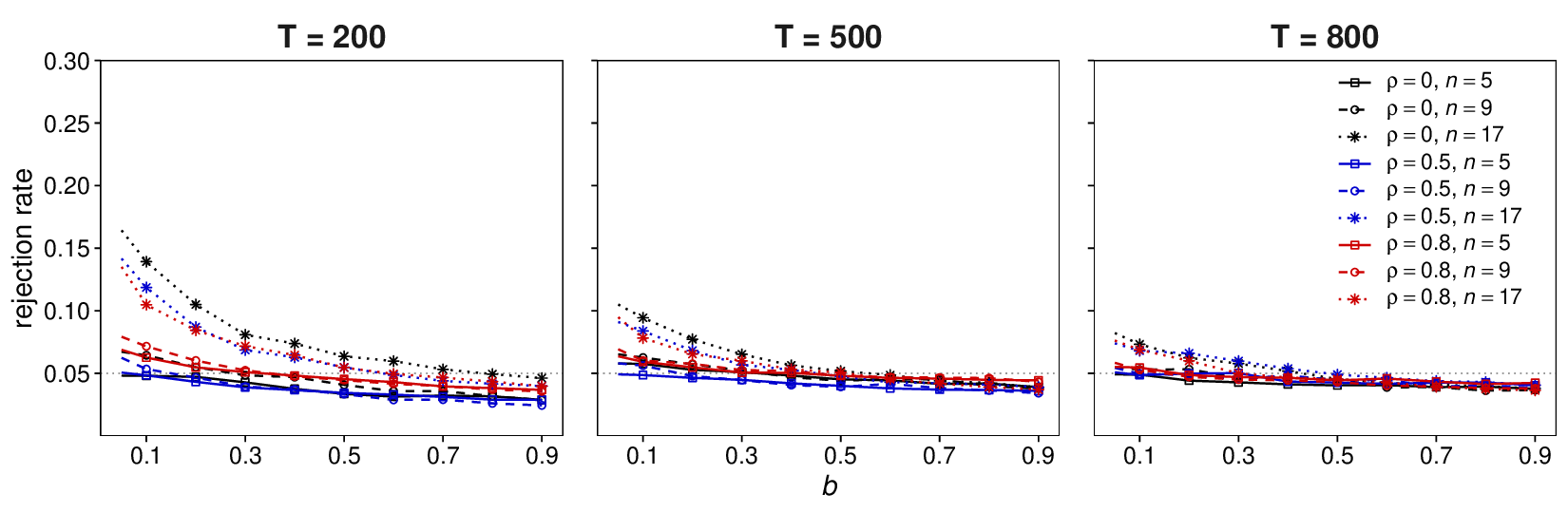}
        \caption{Rejection rates of the homogeneity test; Bartlett variance estimator}
        \label{fig:6_homo_kernel}
    \end{subfigure}

    \vspace{0.8em}

     \begin{subfigure}{\linewidth}
        \centering
        \includegraphics[width=0.9\linewidth]{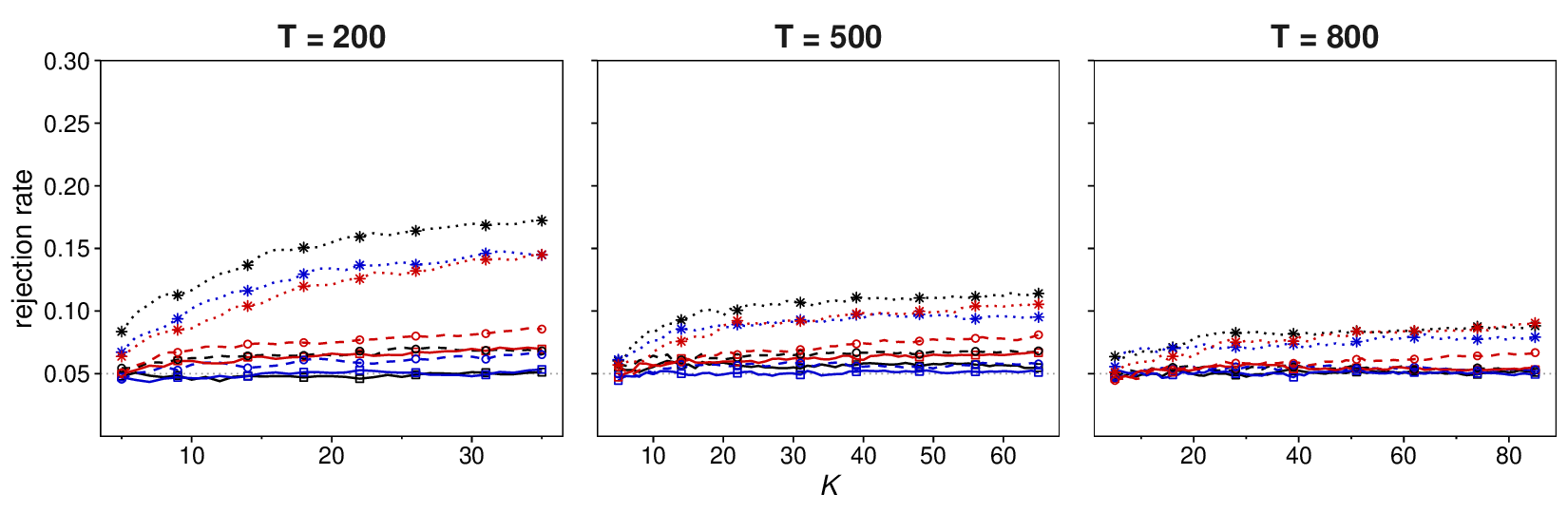}
        \caption{Rejection rates of the homogeneity test; OS variance estimator}
        \label{fig:7_homo_OS}
    \end{subfigure}
\end{figure}

\begin{figure}[htbp]

    \caption{Rejection Rates of the Monotonicity Test}
    \label{fig:mono_combined}
    
    \centering
    \begin{subfigure}{\linewidth}
        \centering
        \includegraphics[width=0.9\linewidth]{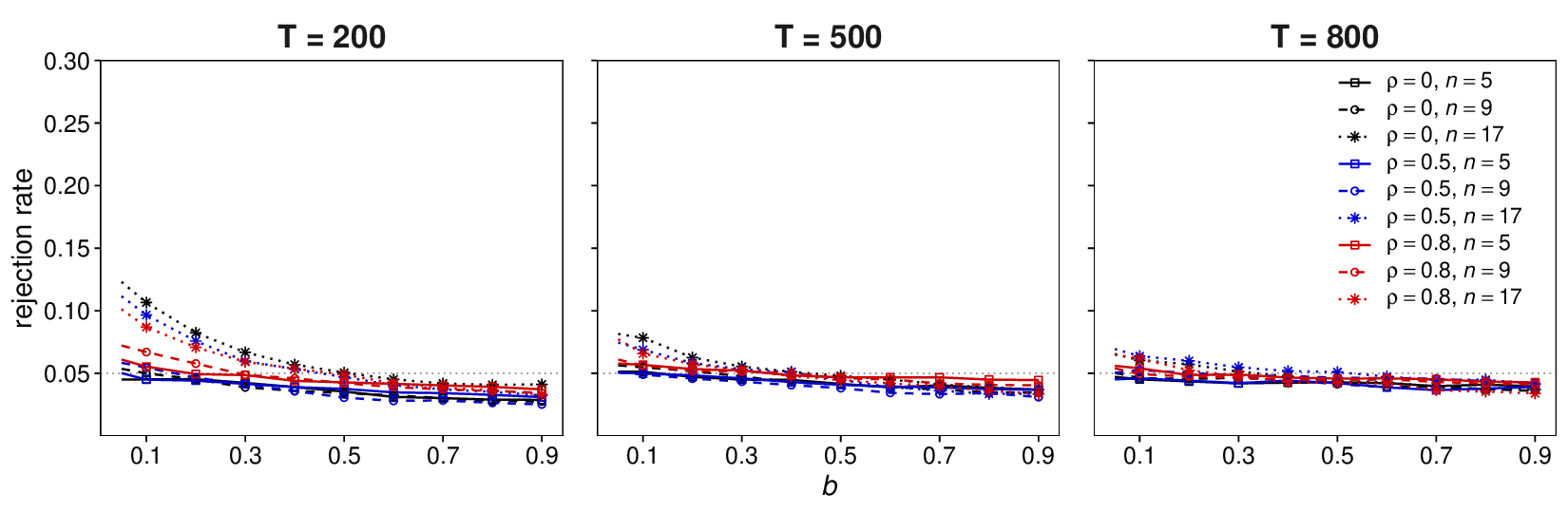}
        \caption{Rejection rates of the monotonicity test; Bartlett variance estimator}
        \label{fig8_mono_kernel}
    \end{subfigure}

    \vspace{0.8em}

    \begin{subfigure}{\linewidth}
        \centering
        \includegraphics[width=0.9\linewidth]{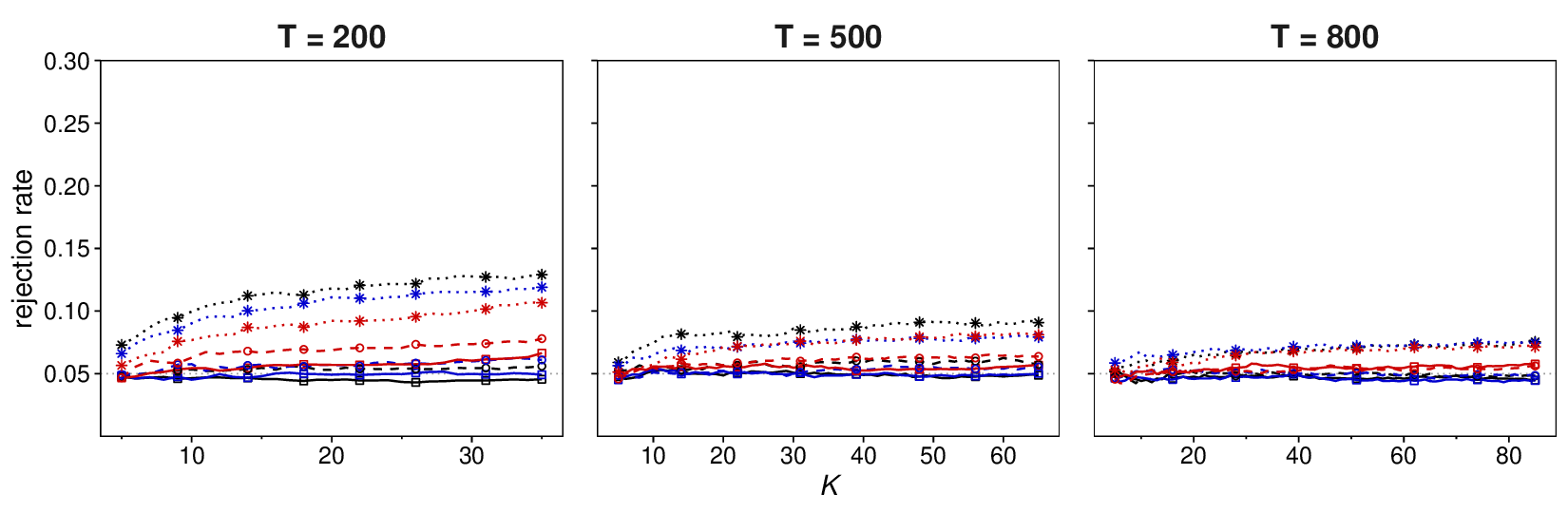}
        \caption{Rejection rates of the monotonicity test; OS variance estimator}
        \label{fig:9_mono_OS}
    \end{subfigure}
\end{figure}

The simulation design is the same as that used for the Sup-t test in Section~\ref{sec:sim}. Figures~\ref{fig:homo_combined} and \ref{fig:mono_combined} report the simulation results for the homogeneity and monotonicity tests, respectively. To examine the effect of bandwidth choices on these tests, we evaluate the tests over a grid of bandwidth values for both the kernel and OS HAR estimators. For the kernel method, we use the grid $b \in \{0.1,0.2,\ldots, 1\}$. For these tests, the OS HAR estimator does not need to be invertible. Therefore, $K$ can be smaller than the number of restrictions, and we consider 10 evenly spaced values between $5$ and $T-1$ in the simulation. 

The two tests exhibit similar overall patterns. In particular, the kernel-HAR approach exhibits excellent finite sample performance, although it becomes slightly conservative as the bandwidth ratio $b$ approaches one. The OS-HAR approach also provides excellent size control, except when K is large.

\newpage

\subsection{Additional Simulations: Sensitivity Analysis}
\label{sec:appendix_sim}

This section first investigates the performance of the Wald test as the number of restrictions, or equivalently the number of quantile levels, increases. The null hypothesis is $H_0 : \beta_1(\tau_1)= \cdots =\beta_1(\tau_m) = 1$, which imposes $m$ restrictions. For $m = 5$, we consider the quantile levels $\tau \in \{0.1, 0.3, \cdots , 0.9\}$. Similarly, $m = 9$ and $m = 17$ correspond to $\tau \in \{0.1, 0.2, \ldots, 0.8, 0.9\}$ and $\tau \in \{0.1, 0.15, \ldots, 0.85, 0.9\}$, respectively. 

Table~\ref{tab:sim-stacked} reports the results. As the number of restrictions $m$ increases, the Wald test tends to exhibit higher null rejection rates. For example, consider the case with $T=500$, $\rho=0.8$, and a nominal significance level of 5\%. For the kernel-HAR method, the rejection rates are $0.112$, $0.138$, and $0.179$ for $m=5$, $9$, and $17$, respectively. The corresponding rejection rates for the OS-HAR method are $0.103$, $0.124$, and $0.155$. We see a similar pattern for other values of $T$ and $\rho$. Overall, the results indicate that over-rejection of the Wald test becomes more pronounced as the number of null restrictions increases.

\begin{table}[!htbp]
\centering
\caption{Wald test for $H_0 : \beta_1(\tau_1)= \cdots =\beta_1(\tau_m) = 1$.}
\label{tab:sim-stacked}
\footnotesize
\setlength{\tabcolsep}{2.5pt}

\begin{tabular}{@{}rrrrrrrrr@{}}
\toprule
\multicolumn{3}{c}{Setting}
  & \multicolumn{3}{c}{kernel-HAR}
  & \multicolumn{3}{c}{OS-HAR} \\
\cmidrule(lr){1-3}\cmidrule(lr){4-6}\cmidrule(lr){7-9}
$T$ & $\rho$ & $m$ & $M$ & 10\% & 5\%
  & $K$ & 10\% & 5\% \\
\midrule
200 & 0.0 & 5 & 1 & 0.157 & 0.097 & 199 & 0.146 & 0.090 \\
 &  & 9 & 1 & 0.170 & 0.109 & 199 & 0.157 & 0.100 \\
 &  & 17 & 1 & 0.269 & 0.187 & 199 & 0.246 & 0.160 \\
\cmidrule(lr){1-9}
200 & 0.5 & 5 & 7 & 0.168 & 0.102 & 43 & 0.155 & 0.093 \\
 &  & 9 & 7 & 0.183 & 0.116 & 43 & 0.175 & 0.107 \\
 &  & 17 & 6 & 0.262 & 0.178 & 45 & 0.248 & 0.162 \\
\cmidrule(lr){1-9}
200 & 0.8 & 5 & 36 & 0.228 & 0.150 & 14 & 0.204 & 0.128 \\
 &  & 9 & 33 & 0.244 & 0.155 & 13 & 0.215 & 0.118 \\
 &  & 17 & 26 & 0.296 & 0.209 & 21 & 0.244 & 0.142 \\
\midrule
500 & 0.0 & 5 & 1 & 0.135 & 0.076 & 499 & 0.129 & 0.070 \\
 &  & 9 & 1 & 0.135 & 0.079 & 499 & 0.132 & 0.075 \\
 &  & 17 & 1 & 0.172 & 0.106 & 499 & 0.167 & 0.101 \\
\cmidrule(lr){1-9}
500 & 0.5 & 5 & 7 & 0.143 & 0.092 & 112 & 0.133 & 0.082 \\
 &  & 9 & 6 & 0.147 & 0.090 & 113 & 0.144 & 0.083 \\
 &  & 17 & 5 & 0.192 & 0.121 & 119 & 0.192 & 0.120 \\
\cmidrule(lr){1-9}
500 & 0.8 & 5 & 35 & 0.176 & 0.112 & 34 & 0.164 & 0.103 \\
 &  & 9 & 33 & 0.213 & 0.138 & 32 & 0.197 & 0.124 \\
 &  & 17 & 27 & 0.266 & 0.179 & 34 & 0.258 & 0.155 \\
\midrule
800 & 0.0 & 5 & 1 & 0.121 & 0.064 & 799 & 0.112 & 0.060 \\
 &  & 9 & 1 & 0.119 & 0.068 & 799 & 0.116 & 0.065 \\
 &  & 17 & 1 & 0.149 & 0.087 & 799 & 0.149 & 0.084 \\
\cmidrule(lr){1-9}
800 & 0.5 & 5 & 14 & 0.136 & 0.081 & 129 & 0.128 & 0.069 \\
 &  & 9 & 12 & 0.139 & 0.081 & 132 & 0.136 & 0.078 \\
 &  & 17 & 10 & 0.166 & 0.096 & 138 & 0.167 & 0.097 \\
\cmidrule(lr){1-9}
800 & 0.8 & 5 & 69 & 0.152 & 0.095 & 39 & 0.144 & 0.085 \\
 &  & 9 & 65 & 0.174 & 0.106 & 37 & 0.176 & 0.103 \\
 &  & 17 & 53 & 0.223 & 0.139 & 39 & 0.232 & 0.139 \\
\bottomrule
\end{tabular}
\par\medskip
{\footnotesize
 \parbox{0.6\linewidth}{
   Note: $M$ and $K$ are the medians of the bandwidths computed by the test-based bandwidth rules described in Section \ref{sec:bandwidth}. The table reports rejection rates for the stack-Wald test at nominal levels of 10\% and 5\%.
 }%
}
\end{table}

We then repeat the same exercise for the Sup-t test. The null hypothesis is $H_0 : \beta_1(\tau) = 1$ for all $\tau \in [0.1,0.9]$. To implement the test, we discretize the quantile index set using $n$ grid points. When $n=5$, we consider the grid $\{0.1, 0.3, \cdots , 0.9\}$. Similarly, $n = 9$ and $n = 17$ correspond to the grids $\{0.1, 0.2, \ldots, 0.8, 0.9\}$ and $\{0.1, 0.15, \ldots, 0.85, 0.9\}$, respectively.

\begin{table}[!htbp]
\centering
\caption{Sup-t test for $H_0 : \beta_1(\tau) = 1$ for all $\tau \in [0.1,0.9]$.}
\label{tab:sim-sup}
\footnotesize
\setlength{\tabcolsep}{2.5pt}
\begin{tabular}{@{}rrrrrrrrr@{}}
\toprule
\multicolumn{3}{c}{Setting}
  & \multicolumn{3}{c}{kernel-HAR}
  & \multicolumn{3}{c}{OS-HAR} \\
\cmidrule(lr){1-3}\cmidrule(lr){4-6}\cmidrule(lr){7-9}
$T$ & $\rho$ & $n$ & $M$ & 10\% & 5\%
  & $K$ & 10\% & 5\% \\
\midrule
200 & 0.0 & 5 & 1 & 0.151 & 0.090 & 199 & 0.150 & 0.089 \\
 &  & 9 & 1 & 0.155 & 0.093 & 199 & 0.153 & 0.094 \\
 &  & 17 & 1 & 0.157 & 0.098 & 199 & 0.154 & 0.097 \\
\cmidrule(lr){1-9}
200 & 0.5 & 5 & 7 & 0.163 & 0.102 & 43 & 0.164 & 0.101 \\
 &  & 9 & 7 & 0.171 & 0.101 & 43 & 0.168 & 0.102 \\
 &  & 17 & 6 & 0.179 & 0.113 & 45 & 0.172 & 0.109 \\
\cmidrule(lr){1-9}
200 & 0.8 & 5 & 36 & 0.197 & 0.127 & 14 & 0.210 & 0.131 \\
 &  & 9 & 33 & 0.201 & 0.127 & 13 & 0.211 & 0.131 \\
 &  & 17 & 26 & 0.214 & 0.138 & 21 & 0.230 & 0.149 \\
\midrule
500 & 0.0 & 5 & 1 & 0.127 & 0.075 & 499 & 0.127 & 0.073 \\
 &  & 9 & 1 & 0.130 & 0.071 & 499 & 0.131 & 0.071 \\
 &  & 17 & 1 & 0.134 & 0.076 & 499 & 0.132 & 0.076 \\
\cmidrule(lr){1-9}
500 & 0.5 & 5 & 7 & 0.149 & 0.087 & 112 & 0.146 & 0.085 \\
 &  & 9 & 6 & 0.156 & 0.093 & 113 & 0.150 & 0.089 \\
 &  & 17 & 5 & 0.163 & 0.101 & 119 & 0.154 & 0.092 \\
\cmidrule(lr){1-9}
500 & 0.8 & 5 & 35 & 0.167 & 0.104 & 34 & 0.171 & 0.104 \\
 &  & 9 & 33 & 0.168 & 0.103 & 32 & 0.171 & 0.106 \\
 &  & 17 & 27 & 0.178 & 0.117 & 34 & 0.180 & 0.115 \\
\midrule
800 & 0.0 & 5 & 1 & 0.117 & 0.060 & 799 & 0.117 & 0.059 \\
 &  & 9 & 1 & 0.116 & 0.064 & 799 & 0.115 & 0.065 \\
 &  & 17 & 1 & 0.119 & 0.069 & 799 & 0.118 & 0.068 \\
\cmidrule(lr){1-9}
800 & 0.5 & 5 & 14 & 0.130 & 0.077 & 129 & 0.130 & 0.074 \\
 &  & 9 & 12 & 0.138 & 0.077 & 132 & 0.136 & 0.077 \\
 &  & 17 & 10 & 0.144 & 0.080 & 138 & 0.138 & 0.079 \\
\cmidrule(lr){1-9}
800 & 0.8 & 5 & 69 & 0.142 & 0.087 & 39 & 0.153 & 0.090 \\
 &  & 9 & 65 & 0.146 & 0.089 & 37 & 0.154 & 0.093 \\
 &  & 17 & 53 & 0.155 & 0.094 & 39 & 0.162 & 0.097 \\
\bottomrule
\end{tabular}

\par\medskip
{\footnotesize
 \parbox{0.6\linewidth}{
   Note: The table reports rejection rates for the Sup-t test at nominal levels of 10\% and 5\%. $M$ and $K$ are the medians of the bandwidths selected by the test-based bandwidth rules described in Section \ref{sec:bandwidth}. 
   }%
}
\end{table}

Table~\ref{tab:sim-sup} reports the rejection rates of the Sup-t test. The results show that the test is largely insensitive to the choice of $n$. For example, consider the case with $T=800$, $\rho=0.8$, and a 5\% significance level. The rejection rates for the kernel-HAR method are $0.087$, $0.089$, and $0.094$ for $n = 5$ ,$9$, and $17$, respectively, while the corresponding rates for the OS-HAR method are $0.074$, $0.077$, and $0.079$, respectively. Thus, when the Sup-t test is viewed as a test of the null hypothesis over the continuous quantile range $\mathcal{T} = [0.1,0.9]$, its performance is largely unaffected by the number of grid points used to discretize $\mathcal{T}$.

These results can also be viewed from a different perspective. If the test is instead interpreted as a test of the multiple hypothesis $H_0 : \beta_1(\tau_1)= \cdots =\beta_1(\tau_n) = 1$, it imposes $n$ restrictions. From this perspective, the results indicate that, unlike the Wald test, the Sup-t test is relatively insensitive to the number of restrictions under the null. This suggests that the Sup-t test serves as a useful alternative to the Wald test when the null hypothesis involves a large number of restrictions.

Comparing the kernel-HAR and OS-HAR approaches, we find that the two methods provide comparable size control, with neither clearly dominating the other. This pattern holds for both the Sup-t and Wald tests.

\subsection{Additional Simulations: Empirical Power of the Tests} 

This section examines the finite sample power of the Wald and Sup-t tests for testing the null $\beta(\tau)=1$ for either a grid of $\tau$ (Wald) or all $\tau$ (Sup-t). The simulation design remains the same as in the previous sections, except that we set $\beta(\tau) = 1+\delta$, where $\delta\in (0,1]$ indexes deviations from the null. A $0.01$ grid is used for $\delta$. Power is size-adjusted so that observed power differences do not reflect null rejection distortions.

\begin{figure}[!htb]
\centering
\includegraphics[width=0.44\linewidth]{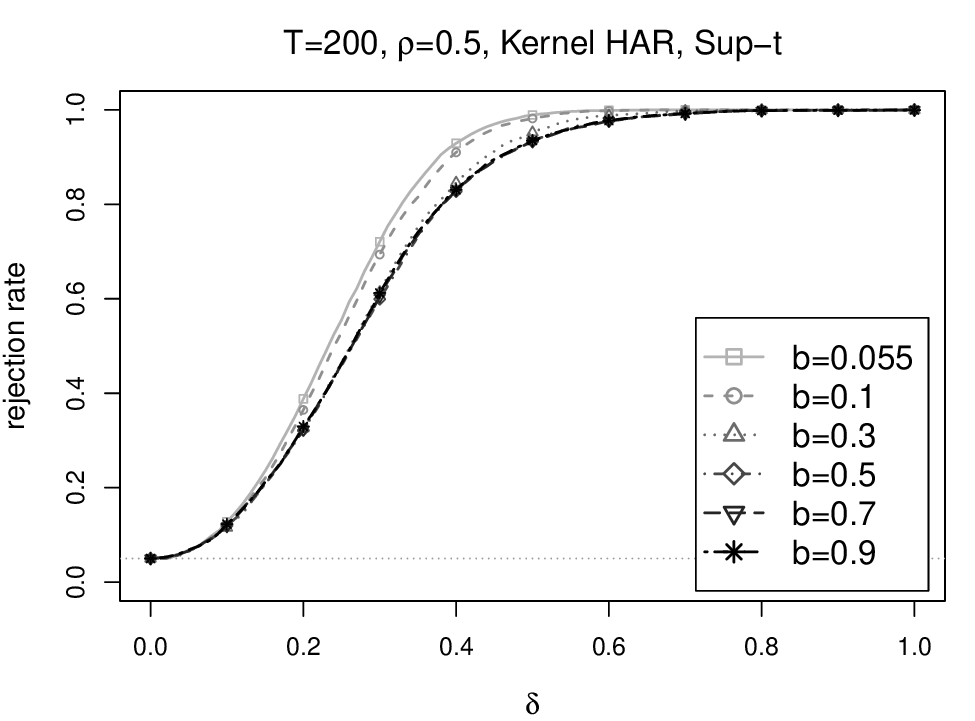}\hfill\includegraphics[width=0.44\linewidth]{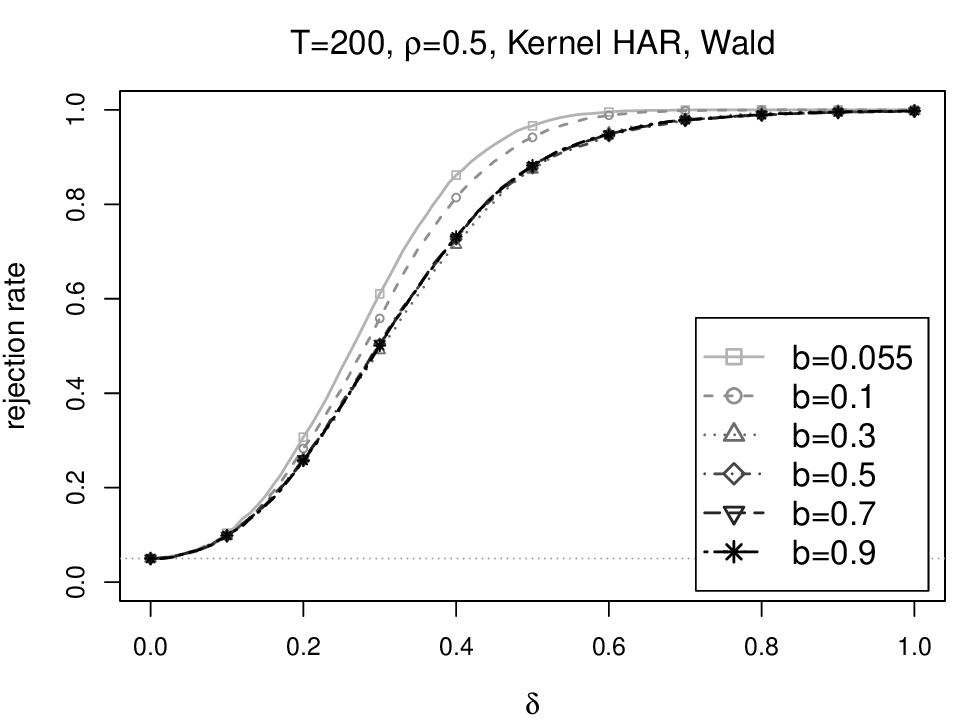}\\
\includegraphics[width=0.44\linewidth]{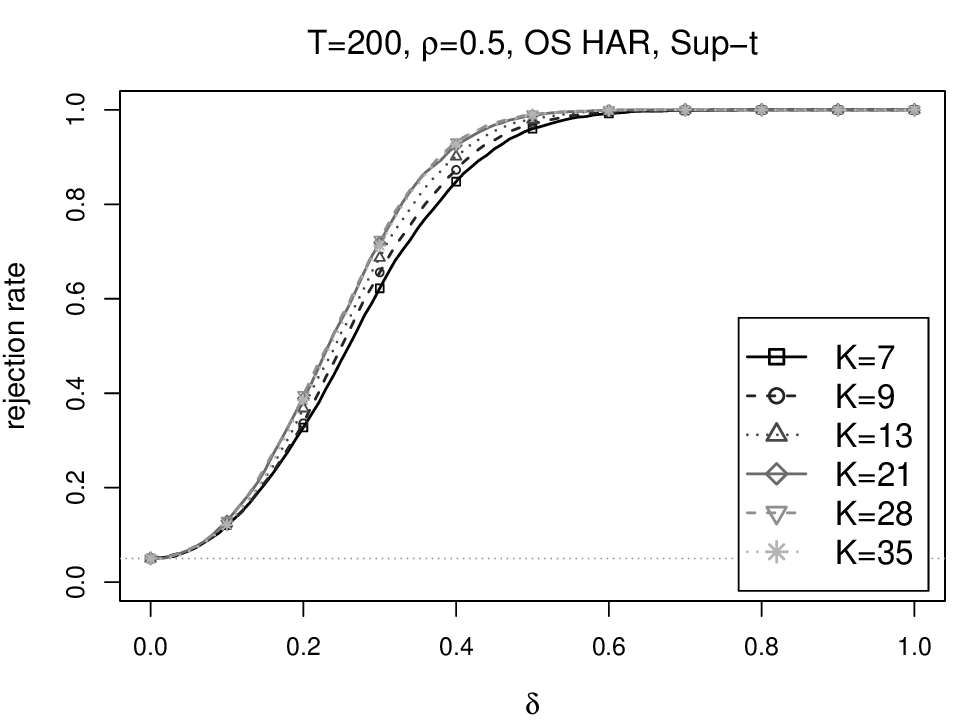}\hfill\includegraphics[width=0.44\linewidth]{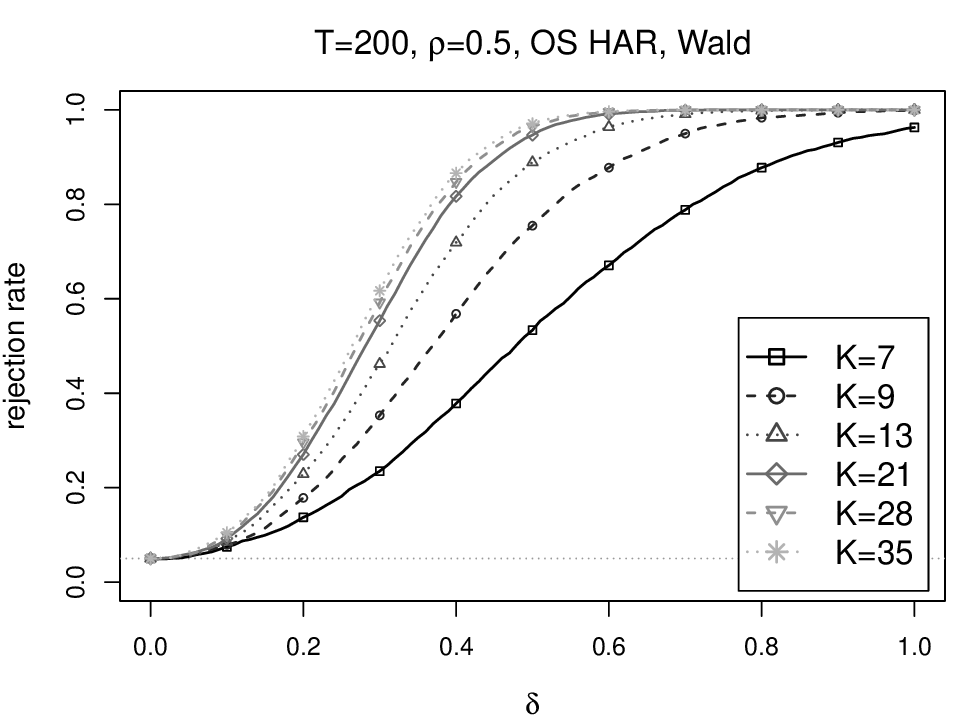}\\
\caption{Size-adjusted power across different smoothing parameters}
\label{fig_power_diffsmoothing}
\end{figure}

Figure~\ref{fig_power_diffsmoothing} illustrates the effect of the smoothing parameters on size-adjusted power for $T=200$ and $\rho=0.5$. We set $m=5$ for the Wald test and $n=5$ for the Sup-t test.\footnote{Power simulation results for other configurations are available upon request.} Overall, the proposed fixed-smoothing methods have good power properties. As expected, power increases as $b$ decreases or $K$ increases. The Sup-t test is less sensitive to the choice of smoothing parameter and, more importantly, has higher power than the Wald test. The higher power of Sup-t relative to Wald in our quantile context is similar to the finding by \cite{montiel2019simultaneous} of more narrow confidence intervals (higher power) of Sup-t relative to other methods for inference on impulse response functions.

Figures~\ref{fig_power_diffm_kernel} and \ref{fig_power_diffm_os} report size-adjusted power for different choices of $m$ and $n$, focusing on the case with $T=200$, $\rho=0.5$, and smoothing parameters $b \in \{0.1,0.3\}$, and $K \in \{13,21\}$. Overall, the simulations reveal the well-known trade-off between size distortion and power associated with the degree of smoothing. This trade-off is particularly pronounced for the Wald test in which $m$ corresponds to the number of restrictions under the null. When the number of restrictions $m$ is small, the Wald test exhibits strong power, but its power declines substantially as $m$ increases. 

In contrast, for the Sup-t test, $n$ determines how finely the grid discretizes $\mathcal{T}$, and while increasing $n$ tends to increase power, increases beyond $n=5$ yields negligible changes in power. 

We also investigate how the placement of the quantile levels $\tau$ (e.g., near the center or in the tails of the distribution) affect empirical power. For the case $m=2$, we consider two alternative choices: $\tau \in \{0.1,0.9\}$ and $\tau \in \{0.3,0.7\}$. The first grid places the quantile levels in the tails of the distribution, while the second focuses on central quantiles. The results show that the tail-based grid yields noticeably lower power than the center-based grid, confirming that inference becomes more challenging at tail quantiles.

\begin{figure}[H]
\caption{Size-adjusted power with different $m$ and $n$, Kernel-HAR method}
\centering
\includegraphics[width=0.44\linewidth]{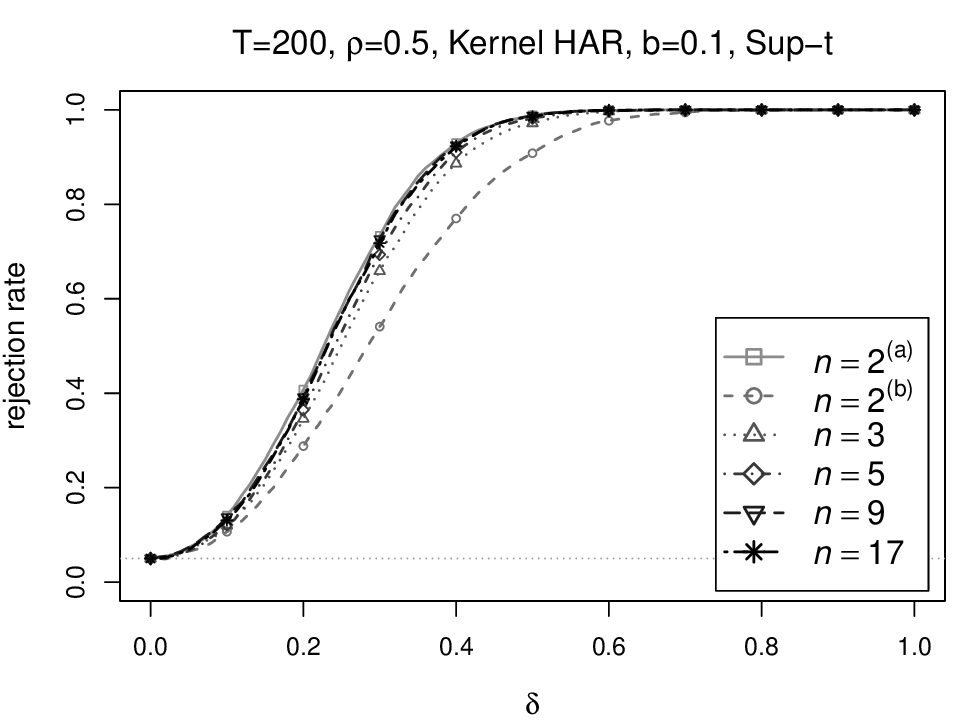}\hfill\includegraphics[width=0.44\linewidth]{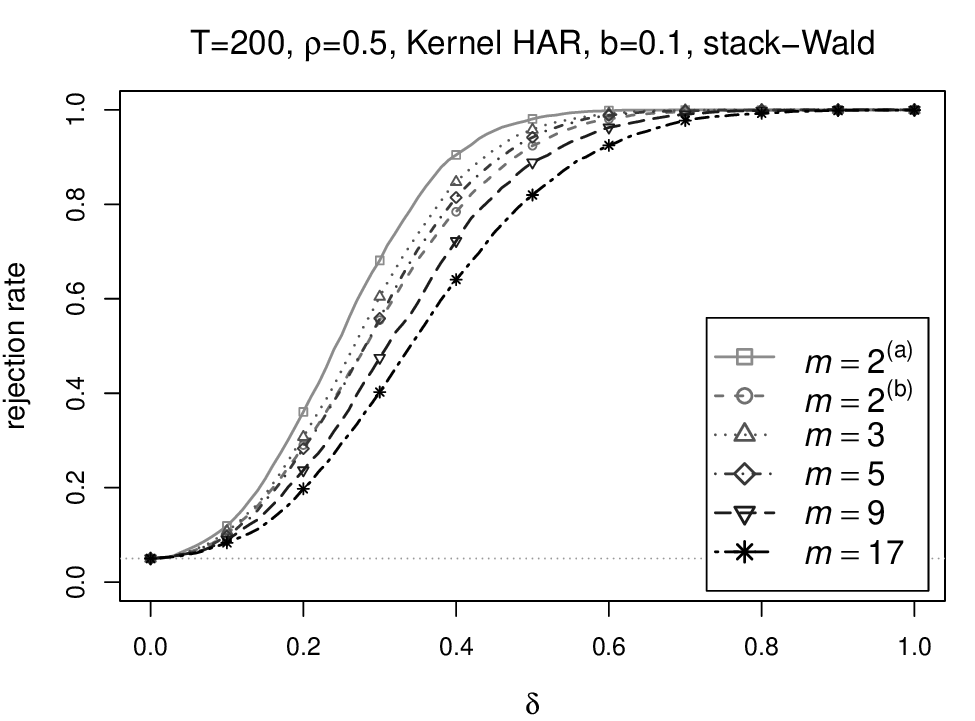}\\
\includegraphics[width=0.44\linewidth]{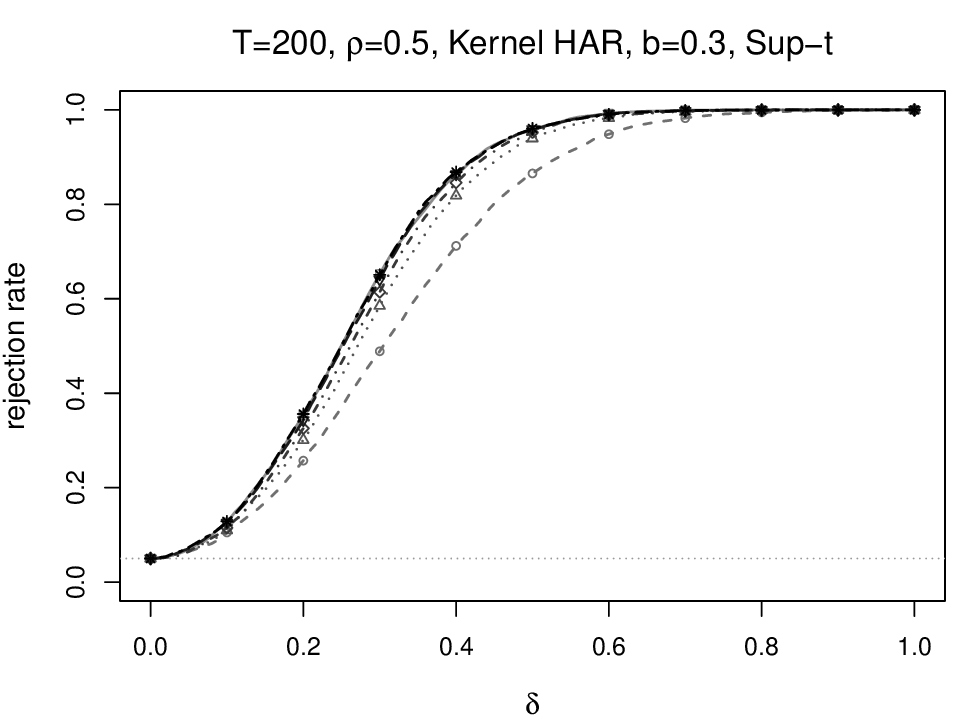}\hfill\includegraphics[width=0.44\linewidth]{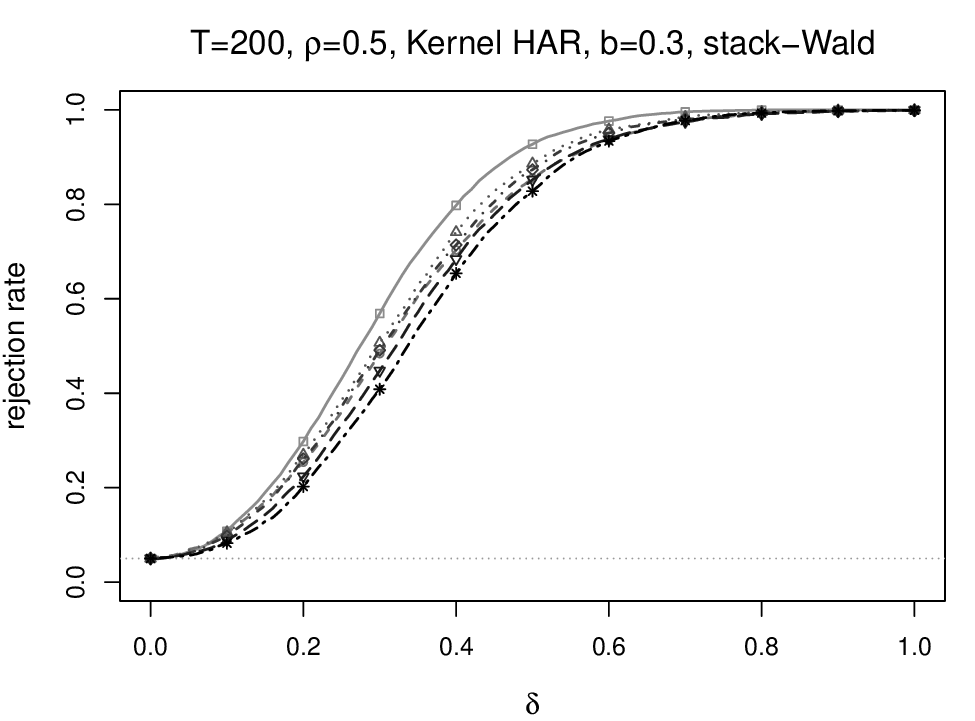}\\
{\footnotesize
 \parbox{0.8\linewidth}{
 \textit{Note:} $m{=}2^{(a)},\,n{=}2^{(a)}$: $\tau \in \{0.3, 0.7\}$; $m{=}2^{(b)},\,n{=}2^{(b)}$: $\tau \in \{0.1, 0.9\}$.
 }}\\
\label{fig_power_diffm_kernel}
\end{figure}

\begin{figure}[H]
\caption{Size-adjusted power with different $m$ and $n$, OS-HAR method}
\centering
\includegraphics[width=0.44\linewidth]{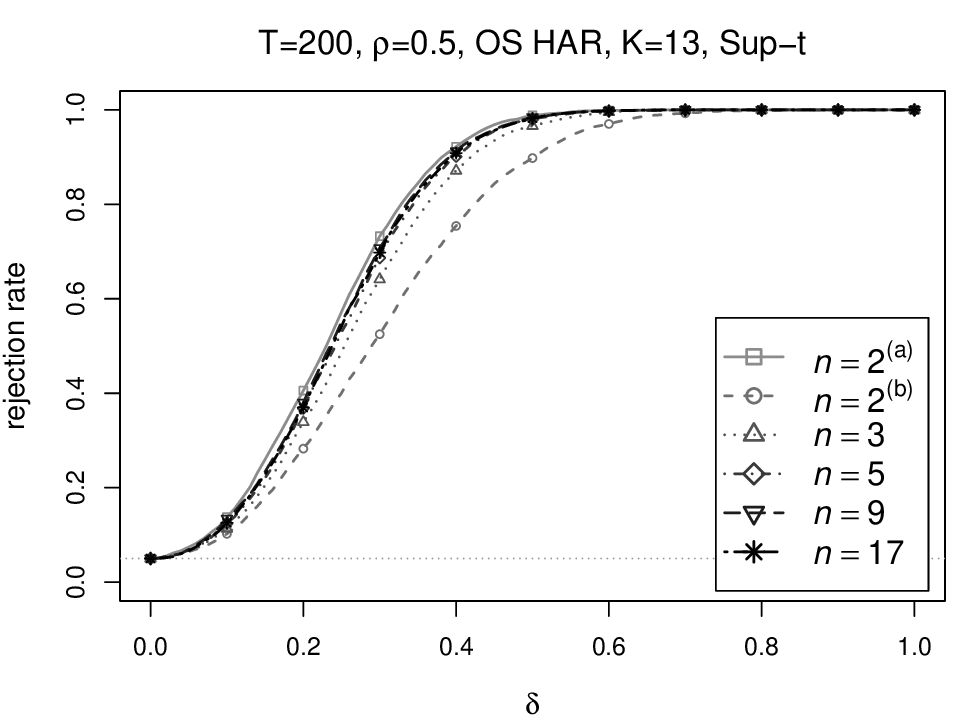}\hfill\includegraphics[width=0.44\linewidth]{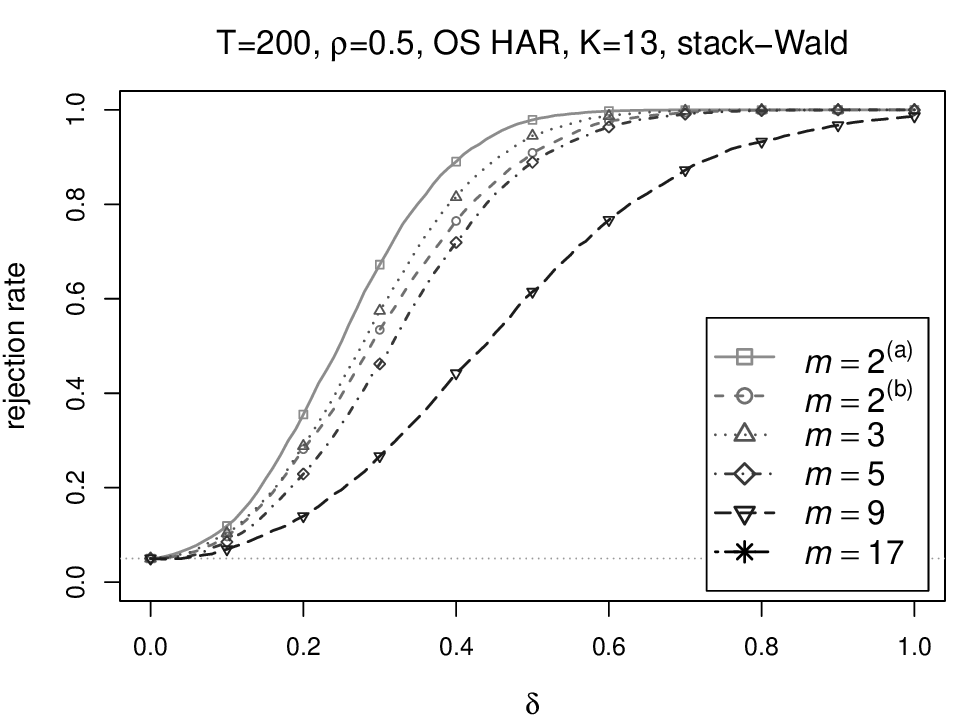}\\
\includegraphics[width=0.44\linewidth]{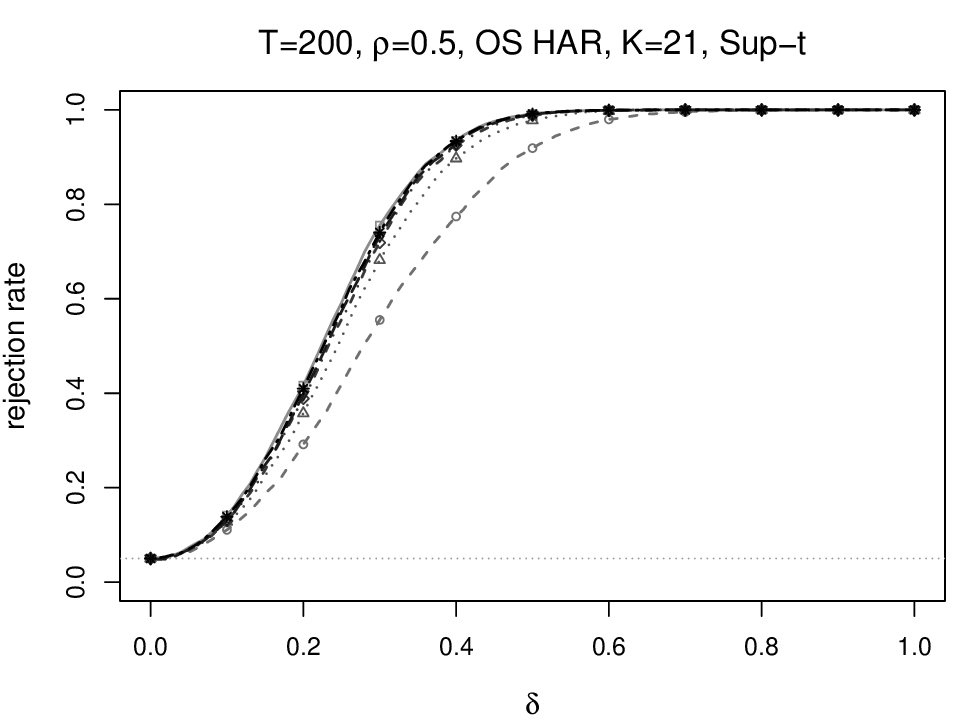}\hfill\includegraphics[width=0.44\linewidth]{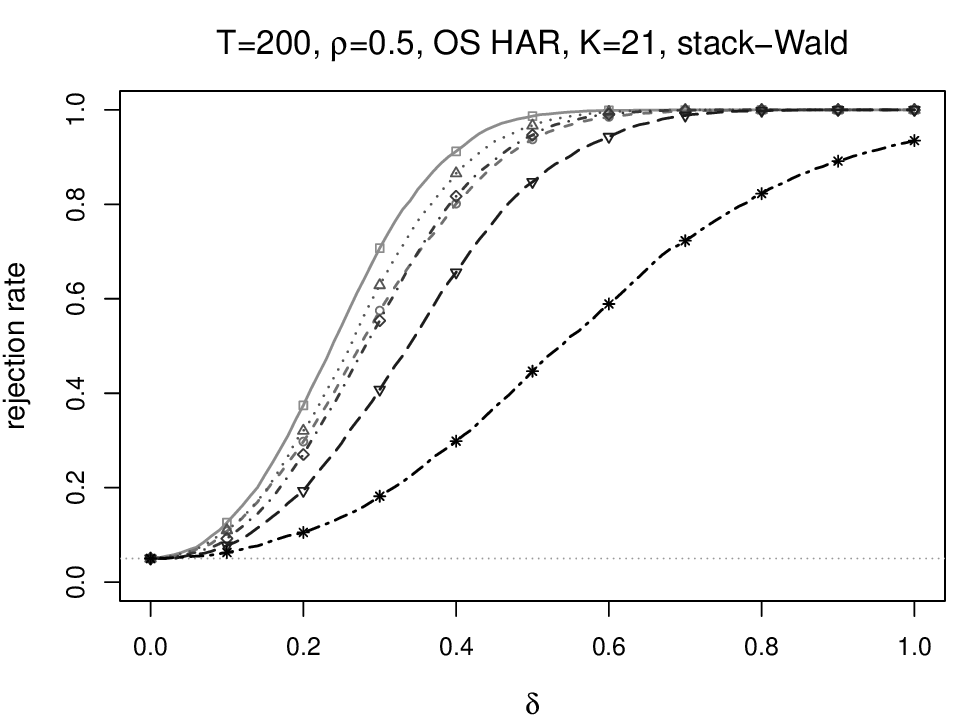}\\
{\footnotesize
 \parbox{0.8\linewidth}{
 \textit{Note:} $m{=}2^{(a)},\,n{=}2^{(a)}$: $\tau \in \{0.3, 0.7\}$; $m{=}2^{(b)},\,n{=}2^{(b)}$: $\tau \in \{0.1, 0.9\}$. For the stack-Wald test with K = 13, the m = 17 case is dropped due to the rank condition.
 }}\\
\label{fig_power_diffm_os}
\end{figure}


\end{document}